%% file: main.tex
\documentclass[11pt,letterpaper]{article}
\usepackage[margin=1in]{geometry}
\usepackage{times}

\usepackage{float}
\usepackage{amsfonts}
\usepackage{amsmath,amssymb}

\usepackage{amstext}
\usepackage{theorem}
\usepackage{graphicx}
\usepackage{url}
\usepackage{graphics}
\usepackage{colordvi}
\usepackage{subfigure}
\usepackage{tabularx}

\usepackage{latexsym}
\usepackage{epic}
\usepackage{epsfig}
\usepackage{hyperref}
\usepackage{multirow}
\usepackage{hhline}
\usepackage{verbatim}
\usepackage[justification=centering]{caption}
\usepackage{enumitem}
\usepackage{color}
\usepackage[cmyk]{xcolor}
\allowdisplaybreaks[1]
\usepackage[numbers]{natbib}
\usepackage{comment}
\usepackage{algorithmic}
\usepackage[linewidth=.5pt]{mdframed}
\usepackage{lipsum}
\usepackage{booktabs} 

\usepackage[linesnumbered,ruled,vlined]{algorithm2e} 

\newenvironment{proofof}[1]{\noindent{\bf Proof of #1.}}%
        {\hspace*{\fill}$\Box$\par\vspace{4mm}}
\newenvironment{proofsketchof}[1]{\noindent{\bf Proof Sketch of #1.}}%
        {\hspace*{\fill}$\Box$\par\vspace{4mm}}

\newcommand{\be}{\begin{enumerate}}
\newcommand{\ee}{\end{enumerate}}
\newcommand{\bd}{\begin{description}}
\newcommand{\ed}{\end{description}}
\newcommand{\bi}{\begin{itemize}}
\newcommand{\ei}{\end{itemize}}

\newtheorem{theorem}{Theorem}[section]
\newtheorem{lemma}[theorem]{Lemma}

\newtheorem{observation}[theorem]{Observation}
\newtheorem{corollary}[theorem]{Corollary}
\newtheorem{claim}[theorem]{Claim}
\newtheorem{proposition}[theorem]{Proposition}
\newtheorem{condition}[theorem]{Condition}

\newtheorem{assumption}{Assumption}[section]
\newtheorem{definition}{Definition}[section]

\newtheorem{remark}{Remark}[section]
\newenvironment{proof}{\par \smallskip{\bf Proof:}}{\hfill\stopproof}
\def\stopproof{\square}
\def\square{\vbox{\hrule height.2pt\hbox{\vrule width.2pt height5pt \kern5pt
\vrule width.2pt} \hrule height.2pt}}

\newcommand{\eps}{\epsilon}

\newcommand{\poly}{\operatorname{poly}}

\mathchardef\hyphen="2D

\SetCommentSty{mycommfont}

\SetAlFnt{\small}
\SetAlCapFnt{\small}
\SetAlCapNameFnt{\small}
\SetAlCapHSkip{0pt}
\IncMargin{-\parindent}

\begin{document}

\title{A Reduction from Multi-Parameter to Single-Parameter Bayesian Contract Design\thanks{Authors are listed in $\alpha$-$\beta$ order.}}


\author{
Matteo Castiglioni \\ Politecnico di Milano \\ \small matteo.castiglioni@polimi.it
\and
Junjie Chen\thanks{Part of the work was done while Junjie Chen was a visiting student at Osaka University.} \\ City University of Hong Kong \\ \small junjchen9-c@my.cityu.edu.hk
\and
Minming Li \\ City University of Hong Kong \\ \small minming.li@cityu.edu.hk
\and
Haifeng Xu\thanks{Xu is supported  by NSF Award  CCF-2303372, Army Research Office Award W911NF-23-1-0030,   Office of Naval Research Award N00014-23-1-2802 and the AI2050 program at Schmidt Sciences (Grant G-24-66104).} \\ University of Chicago\\
\& Google Research\\ \small haifengxu@uchicago.edu
\and 
Song Zuo \\ Google Research \\ \small szuo@google.com
}

	
\date{}

\begin{titlepage}
	\clearpage\maketitle
	\thispagestyle{empty}

\begin{abstract}
The problem of contract design addresses the challenge of \emph{moral hazard} in  principle-agent setups. The agent exerts costly efforts that produce a \emph{random} outcome with an associated reward for the principal. Moral hazard refers to the tension that the principal cannot  observe the agent's effort level   hence needs to incentivize the agent only through rewarding the realized effort outcome, i.e., the \emph{contract}. Bayesian contract design studies the principal's design problem of an optimal contract when facing an unknown agent characterized by a private Bayesian type. In its most general form, the agent's type is inherently ``multi-parameter'' and can arbitrarily affect both the agent's productivity   and effort costs. In contrast,   a natural single-parameter setting of much recent interest simplifies the agent's type to a single value that describes the agent's \emph{cost per unit} of effort,  whereas agents' efforts are assumed to be equally productive.       

The main result of this paper is an almost approximation-preserving polynomial-time reduction from the most general multi-parameter Bayesian contract design (BCD) to  single-parameter BCD. That is, for any multi-parameter BCD instance $I^M$, we construct a single-parameter instance $I^S$ such that any $\beta$-approximate contract (resp. menu of contracts) of $I^S$ can in turn be converted to a $(\beta -\epsilon)$-approximate contract (resp. menu of contracts) of $I^M$. The reduction is in time polynomial in the input size and $\log(\frac{1}{\epsilon})$; moreover, when $\beta = 1$ (i.e., the given single-parameter solution is exactly optimal), the  dependence on $\frac{1}{\epsilon}$ can be removed, leading to a polynomial-time exact reduction. This efficient reduction  is somewhat surprising because in the closely related problem of Bayesian mechanism design, a polynomial-time reduction from multi-parameter to single-parameter setting is believed to not exist. Our result demonstrates the intrinsic difficulty of addressing moral hazard in Bayesian contract design, regardless of being single-parameter or multi-parameter. 

As byproducts, our reduction   answers two open questions in recent literature of algorithmic contract design: (a) it implies that optimal contract design in single-parameter BCD is not in APX unless P=NP even when the agent's type distribution is regular, answering the  open question of \cite{alon2021contracts} in the negative; (b)  it implies that the principal's (order-wise) tight utility gap between using a  menu of contracts and a single contract is $\Theta(n)$ where $n$ is the number of actions, answering the major open question of \cite{guruganesh2021contracts} for the single-parameter case.

\end{abstract}

\end{titlepage}

\newpage 

\input{intro}

\input{prelim}

\input{gap}

\input{reduction}

\input{liability_soda}

\newpage
\bibliographystyle{ACM-Reference-Format}
\bibliography{refer}
\newpage 
\appendix


\newpage

\input{appendix}


\end{document}

%% file: intro.tex
\section{Introduction} 

Mechanism Design   and Contract Design both play fundamental roles in economic theory. Typically, they address different challenges. Mechanism design focuses on incentivizing agents to truthfully reveal their private information, also known as {\em private type}, to the designer. This problem of inducing agents to share  true information  is often referred to as {\em adverse selection}.  The main objective of mechanism design is usually to maximize a designer's revenue or the social welfare under adverse selection. Contract design, on the other hand,  studies the situations where  agents take private actions or exert private efforts that are not observable by the principal. Instead, only certain (often random) outcome of the efforts is observable. It hence  studies the principal's  design problem of incentivizing agents to take certain actions by rewarding their outcomes.   This problem is called {\em moral hazard}. An overview for each of these two fields can also be found in the scientific background by  the Royal Swedish Academy of
Sciences, respectively, for the 2007 \cite{nobel-mechanism} and 2016 \cite{nobel-contract} Nobel Prize in economics.  
While adverse selection and moral hazard are classically studied by these two different fields,  they show up simultaneously in many  applications, hence have been studied jointly in much modern economic analysis. A notable application  domain of this sort is  insurance, where agents of different risk attitudes often take varied actions with different risk levels under insurance contract designed to share  agents' risk costs  --- in fact, this is  also the origin of the term ``moral hazard''. Much economic study in this space addresses moral hazard and adverse selection simultaneous \cite{pauly1978overinsurance,prescott1984pareto,bajari2014moral} (see \cite{cutler2000anatomy} for comprehensive real-world evidences for the existence of both factors); similar modeling of both factors is adopted in other domains such as credit market \cite{theilen2003simultaneous,vercammen2002welfare} and procurement \cite{laffont1986using,laffont2009theory}, as well as in various  abstract principal-agent  models \cite{myerson1982optimal,picard1987design,melumad1989value,caillaud1992noisy,gottlieb2022simple}. 

Given its significant success, it is unsurprising that contract design has gained extensive   interest in computer science. This is because the problem is inherently an optimization problem, and much remains to be understood regarding its  algorithmic foundations in situations with, e.g.,   multiple agents \cite{dutting2023multi}, combinatorial actions \cite{dutting2022combinatorial,dutting2023combinatorial}, Bayesian agent types \cite{xiao2020optimal,guruganesh2021contracts,guruganesh2034contracts,castiglioni2022designing}, etc. These computational studies also pave the way for the large-scale applications of contracts for applications such as strategic machine learning \cite{kleinberg2020classifiers}, AI alignment \cite{hadfield2019incomplete} and incentivizing content creators to produce platform-desired contents \cite{yao2024rethinking}. Notably, many of these algorithmic studies also focused on contract design
under both moral hazard and adverse selection \cite{xiao2020optimal,guruganesh2021contracts,alon2021contracts,guruganesh2034contracts,castiglioni2021bayesian,castiglioni2022designing,alon2022bayesian,gan2022optimal,zhu2023sample}. Besides its more practical modeling of real-world problems, this is also partially because these  models are generally  more challenging to analyze, hence  the  general model is much less well-understood in the economics literature and previous economic studies are often restricted to structured models for specific domains \cite{pauly1978overinsurance,prescott1984pareto,bajari2014moral,theilen2003simultaneous,vercammen2002welfare,picard1987design,caillaud1992noisy}. This is where the advantages of algorithmic lens start to kick in --- recent works have employed \emph{approximation} to understand the efficacy of contracts \cite{xiao2020optimal,guruganesh2021contracts,guruganesh2034contracts} and employed \emph{reductions} to understand the complexity of contracts \cite{dutting2021complexity,castiglioni2021bayesian,castiglioni2022designing} (these also draw much inspiration from the celebrated field of algorithmic mechanism design \cite{nisan1999algorithmic}).

This paper subscribes to the fast-growing field of algorithmic contract design. Also focusing on contract design under both moral hazard and adverse selection, we establish a fundamental algorithmic connection between two very basic models in this space:  one in its most general modeling format (termed the \emph{multi-parameter} setting) and the other  \emph{single-parameter} setting which is  perhaps the simplest possible contract design problem exhibiting both moral hazard and adverse selection. Interestingly, we show that these two extremal settings turn out to be computationally equivalent in a strong sense. More concretely, we study a basic principle-agent interaction where the agent has $n$ actions to choose from, each producing one of $m$ possible outcomes. The stochastic mapping from  actions to outcomes is modeled by a probability matrix $F \in \mathbb{R}^{n\times m}$, where $F_{ij}$ is the probability that  action $a_i$ generates outcome $j$. The vector $c\in \mathbb{R}^n$ describes the agent's cost  for each action.  The principal cannot directly observe the agent's action $a_i$ and therefore has to  incentivize desirable agent actions by designing a payment $p_j(\geq 0)$ for each outcome $j$. The payment vector $p \in \mathbb{R}_+^m$ is called a \emph{contract}.   The principle further faces the   challenge of not knowing the agent's (discrete) private \emph{type} $\theta$ which is drawn from some prior distribution $\mu(\cdot)$.   Due to the Bayesian agent type, we term this problem \emph{Bayesian Contract Design} (BCD), drawing inspiration from the related problem of Bayesian mechanism design. In its most general form, the agent’s
type $\theta$ determines both the transition matrix denoted as $F^{\theta}$, and agent’s  costs denoted as $c^{\theta}$; that is, $\theta$ can be viewed as a tuple of $mn+n$ parameters. This general model contains almost all recently studied BCD problems in the computational literature \cite{xiao2020optimal,guruganesh2021contracts,alon2021contracts,guruganesh2034contracts,castiglioni2021bayesian,castiglioni2022designing,alon2022bayesian,gan2022optimal,zhu2023sample} and  contains many existing economic models  as well, albeit with discrete actions and outcome spaces here. On the other hand, we also study a natural \emph{single-parameter} special case of the above model where the agent's type $\theta \in \mathbb{R}$ is  simplified to   a single value that describes the
agent’s cost per unit of action effort, whereas agents’ efforts are assumed to have the same production efficiency: that is, $F^{\theta} = F$ and $c^{\theta} = \theta \cdot c$ for some publicly known matrix $F$ and vector $c$.  To the best of our knowledge, this basic and elegant model is proposed by \citet{chade2019disentangling}  and later studied from the computational lens by \citet{alon2021contracts,alon2022bayesian}  who  also show interesting connections between this single-parameter BCD problem with  single-parameter Bayesian mechanism design.

As can be seen from the model description, the above general multi-parameter BCD setting may appear significantly more complex than the single-parameter setting. Indeed, besides being a natural special case, this  single-parameter BCD problem is  studied by   \citet{alon2021contracts} also for a key technical reason. Citing the authors' original words, ``\emph{It is well-known from auction design that single-parameter types may allow positive results even when hardness results hold for multi-parameter types}''. Indeed, {though \citet{alon2021contracts,alon2022bayesian} observe    the non-implementability of standard virtual welfare maximizer hence rules out the applicability of   the na\"ive Myersonian approach,  they nevertheless} successfully developed a closed-form optimal contract for the special case of uniform agent type distribution by generalizing Myerson's approach for single-parameter auction design. Additionally, they also design a polynomial-time algorithm for finding an optimal contract for the single-parameter setting with a constant number of actions, which is in sharp contrast to the APX-hardness of  multi-parameter BCD with constantly many actions \cite{guruganesh2021contracts}. 
Given these encouraging positive results for the single-parameter setting,  \citet{alon2021contracts} leave it as a major open question to find a polynomial-time algorithm for the optimal contract in their single-parameter setting.    
This optimism for single-parameter settings is certainly natural because such complexity separation between single-parameter and multi-parameter settings is one of the most well-known insights offered by algorithmic mechanism design \cite{nisan1999algorithmic,daskalakis2015multi}. 
Unfortunately, our main result in this paper shows that the answer to this open question is negative. The positive result in \cite{alon2021contracts}  for the  special class of uniform type distribution does not generalize even to regular type distributions.

\subsection{Main Result and Techniques  } 
Our main result is an almost approximation-preserving polynomial-time reduction  from the above most general multi-parameter BCD problem to single-parameter BCD. This single reduction is applicable to both the optimal design of a single contract used for all types and the optimal design of a menu of contracts, including one contract for each type.\footnote{ One could also ask whether such a reduction exits for   menu of randomized contracts as studied by \cite{castiglioni2022designing,guruganesh2034contracts}. However, this question is not interesting since in both settings the optimal menu of randomized contracts are poly-time solvable and we know any two poly-time solvable problems reduce to each other in polynomial time.   
}
Both forms of contracts are widely studied \cite{myerson1982optimal,guruganesh2021contracts,guruganesh2034contracts,castiglioni2022designing}. Formally, for any multi-parameter BCD instance $I^M$, we construct a single-parameter instance $I^S$ such that any $\beta$-approximate contract (resp. menu of contracts) of $I^S$ can in turn be converted to a $(\beta -\epsilon)$-approximate contract (resp. menu of contracts) of $I^M$. The reduction is in time polynomial in the input size and $\log(\frac{1}{\epsilon})$; moreover, when $\beta = 1$ (i.e., the given single-parameter solution is exactly optimal), the  dependence on $\frac{1}{\epsilon}$ can be removed, leading to a polynomial-time exact reduction. When the given instance $I^M$ has $n$ actions, $m$ outcomes and $K$ agent types, our constructed single-parameter instance $I^S$ has $Kn+1$ actions, $m+1$ outcomes and $K+1$ types.\footnote{This also explains why multi-parameter BCD with constantly many actions is APX-hard \cite{guruganesh2021contracts}, whereas single-parameter BCD with constantly many action admits efficient algorithms. Our polynomial-time reduction from multi-parameter BCD to single-parameter BCD blows up the number of actions by a factor $K$, i.e., the number of agent types. 
}

The existence of such an efficient reduction from multi-parameter BCD to single-parameter BCD  is somewhat surprising. In the closely related field of Bayesian mechanism design, through an extensive line of algorithmic investigations in the past two decades\footnote{The literature on algorithmic mechanism design for multi-parameter settings and multi-item auctions is very rich. Our discussion here can not do justice to the large literature in this area; we hence refer curious readers to an excellent survey by \citet{daskalakis2015multi} with discussions of challenges both  in computation  and structural understandings.},  it is now strongly believed that  a reduction from general multi-parameter settings to single-parameter settings does not exist there.  Hence our result   demonstrates a fundamental difference between these two design problems, and illustrates that the simultaneous presence of moral hazard and adverse selection renders the design problem significantly more challenging regardless of being single-parameter or multi-parameter. The fact that our reduction is (almost) approximation-preserving has multiple useful implications.  On the positive side, it implies that to develop computational or learning-theoretic  (approximate) algorithms with polynomial complexity, it is essentially without loss of generality to focus on the  single-parameter BCD setting and then apply our reduction to extend to multi-parameter settings. On the negative side, it implies that  any previous hardness of approximation  results for multi-parameter settings  (which is rife in the literature) now goes hand in hand with the hardness of approximation for single-parameter settings which is much rarer and, in fact, still remains an open question before our work \cite{alon2021contracts,alon2022bayesian}.  Specifically, our reduction implies the following  new results on hardness of approximation  for single-parameter BCD; both hardness holds even when the agent's type is drawn from a \emph{regular} distribution:
\begin{itemize}
    \item For any $\delta\in (0, 1]$, it is NP-hard  to compute a $\frac{1}{K^{1-\delta}}$--approximation to the optimal single contract, where $K$ is the number of types. This is a corollary of our approximation-preserving reduction and a similar inapproximability result for multi-parameter setting in \cite{castiglioni2021bayesian}.  

 \item For any constant $\rho>0$, it is NP-hard to compute  a $\rho$-approximation to the optimal menu of contracts. This is a corollary of our approximation-preserving reduction and a similar inapproximability result for multi-parameter setting in \cite{castiglioni2022designing}.    
\end{itemize}

 The proof of our main result above turns out to be rather involved. To illustrate the key conceptual ideas of the reduction, we actually start with a simpler research problem of constructing a class of single-parameter BCD instances that have a large principal utility gap between adopting the optimal \emph{menu} of contracts and the optimal \emph{single contract}. We are able to construct such a class of instances that achieves $\Omega(n)$ gap where $n$ is the number of agent actions.  The core idea under our construction is  to carefully choose instance parameters (i.e., action costs $c$, transition matrix $F$ and type distribution $\mu(\cdot)$) so that the best  menu  of contracts
will incentivize each agent type to take certain designated “profitable” action whereas any single contract
will only be able to incentivize very few of these profitable actions. During this construction, we
found that  exponentially-increasing/decaying parameters of form $2^{f(i,k)}$, where $f$ is a carefully designed function that depends on agent action $i$ and type $\theta_k$,  are very useful for enforcing such desired
properties. A direct construction of this instance and its correctness proof are more intuitive, though the existence of such an instance  is in fact also  a corollary of our reduction (as we will illustrate). Therefore we present this direct proof as a ``warm-up'' for our more involved reduction from multi-parameter to single-parameter settings.  
Despite being a starting point for our main result, this warm-up result already resolves an open question posed by \cite{guruganesh2021contracts} for understanding the power of the menu of contracts in terms of securing the principal's utility, compared to merely using a single contract. Specifically, they show that the optimal principal utility from the menu of contracts $OPT_{\text{menu}}$ is at most $O(nK)$ times  the optimal utility from a single contract $OPT_{\text{single}}$, where $K$ is the number of agent types. A major open question left from   \cite{guruganesh2021contracts}  is to pin down the tight ratio between $OPT_{\text{menu}}$ and $OPT_{\text{single}}$ ---  specifically, whether this ratio could be a constant \cite{guruganesh2021contracts}?  This problem is motivated by the well-known algorithmic question of simple-vs-optimal design \cite{hartline2009simple} because a menu of contracts is significantly more complex. Given that a single contract is much more often used in practice,\footnote{As discussed in \citet{bajari2001incentives}, ``the descriptive engineering and construction management literature (...) suggests that menus of contracts are not used. Instead, the vast majority of contracts are variants of simple fixed-price  and cost-plus   contracts".}   it is important to understand how well it performs. In a follow-up work, \citet{guruganesh2034contracts} show that in   general multi-parameter BCD the ratio $OPT_{\text{menu}}/OPT_{\text{single}}$ can be  as large as $\Omega(\max \{ n, \log K \} )$. This leaves the (seemingly   simpler) single-parameter setting the most hopeful situation for having a constant upper bound for $OPT_{\text{menu}}/OPT_{\text{single}}$. Unfortunately, our result here shows an $\Omega(n)$ lower bound for $OPT_{\text{menu}}/OPT_{\text{single}}$  in single-parameter setting.  Notably, this ratio turns out to be order-wise tight, thanks to an upper bound $O(n)$ from \cite{alon2021contracts}.

Inspired by the worst-case instance construction above, we  move on to establish our main result, a polynomial-time approximation-preserving  reduction from multi-parameter to single-parameter BCD. At an intuitive level, given a multi-parameter instance with $n$ agent actions and $K$ agent types, we wish to  construct a single parameter setting with $n\times K$ actions, where ${a}_{i}^{\theta}$ 
represents agent type $\theta$'s action $a_i$ with action-to-outcome transition probabilities $(F^{\theta}_{i,j})_j$.  
This single-parameter instance only has a single matrix $\bar{F}$ that describes the outcome distributions for all these $nK$ actions. To simulate the given multi-parameter instance, we wish any agent type $\theta$ would only take the set of actions 
$\{ {a}_i^{\theta} \}_i$ designated to this particular type, even though in principle all actions are available to him. This is where our idea for constructing the lower bound instance in the above warm-up result becomes a useful starting point. Inspired by that construction, our idea is to use exponentially-decaying/increasing parameters to reweigh each agent type's probability mass, costs and production efficiency. However,  implementing this idea turns out to be significantly more involved here due to the requirement of establishing \emph{precise} connection between the given multi-parameter instance and constructed single-parameter instance (whereas previously all our arguments only need to be order-wise tight to obtain the tight $\Omega(n)$ bound).  Indeed, our ultimate rigorous reduction turns out to need to introduce auxiliary  actions,  outcomes and also agent types in order to complete the reduction.
{Among others, one of the central challenges in our proof is to reduce any approximately optimal solution of the constructed single-parameter instance to a solution of the original multi-parameter instance with almost equally good guarantee. Since many more actions  are available to each agent in this  single-parameter instance, the truth is that the ultimate action chosen by type $\theta$ agent is actually not always inside the designated set $\{ {a}_i^{\theta} \}_i$ as we have wished, hence can be different from that in the multi-parameter instance. In fact, trivially forcing the agent to choose an action from the designated set in the constructed single-parameter instance can potentially cause a large loss to the principal's utility. Our proof  starts from any approximately optimal solution to the  single-parameter instance and  employs highly nontrivial analysis to ensure that its conversion to a solution for the multi-parameter instance  will have the same guarantee, up to a small additive loss. This is  achieved by constructing a sequence of intermediate approximate solutions and is finally completed by employing an $\eps$-IC to IC reduction akin to \cite{castiglioni2022designing,cai2021efficient,daskalakis2012symmetries}.} 

{Finally, we  initiate a diagnosis about the underlying sources that cause the difficulty of Bayesian contract design. Here we show that if we were to remove the limited liability constraint and assume full rank of transition $F$ in the action space (i.e., $F$ has rank $n$) in single-parameter BCD, this restores both the efficacy and computational complexity of single contracts. Specifically, under these assumptions, we show that the best possible  principal utility in single-parameter BCD can be achieved by a single contract; moreover an optimal single contract can be computed in polynomial time.  While this is not a main result of our paper, we view it as a useful exploration towards understanding the difficulties underlying Bayesian contract design. Notably, both conditions seem necessary for proving our positive results. It is an interesting open question of whether  some positive results can be restored by merely either  removing limited liability or assuming full rank (previous hardness instances in \cite{guruganesh2021contracts,castiglioni2021bayesian,castiglioni2022designing} do not have full rank).}





\subsection{Additional Discussion on Related Works} 
Earlier seminal works (e.g, \cite{grossman1992analysis,shavell1979risk,holmstrom1980theory}) laid the foundation of contract theory. 
Our work subscribes to the recent fast-growing literature on algorithmic contract design,
driven in part by many contract-based markets deployed to Internet applications. As discussed in \cite{alon2021contracts}, these markets are of substantial economic value, including crowdsourcing, sponsored content creation, affiliate marketing, freelancing and more. Hence, the computational and algorithmic approaches play an increasingly important role.
To the best of our knowledge, \citet{dutting2019simple} was the first to study the computational aspect of the contract design problem, especially with attention to the approximation between the simple linear contract and optimal contract.
Since then, the complexity and approximability of optimal contract design in succinctly represented settings have received significant interest in theoretical computer science \cite{dutting2021complexity,dutting2023multi,dutting2022combinatorial,dutting2023combinatorial,kleinberg2020classifiers}. \citet{dutting2023multi}  study contract design with multiple agents taking binary action;  \citet{dutting2022combinatorial,dutting2023combinatorial}  consider contract design under combinatorial agent actions and structured principal utility functions. \citet{dutting2023ambiguous} study an ambiguous contract design problem where the principal may commit to a set of contracts. They show that  
the optimal ambiguous contract consists of single-outcome contracts with only one non-zero payment. 

{
On the technical side, our reduction constructs instances with parameters in exponential forms. To our knowledge, this is a new construction approach as a way to  separate different agent types. We are only aware of a similar usage of exponential parameters by \citet{dutting2019simple} in contract design. However, they use  exponential costs and rewards for a fundamentally different purpose from ours.  In our construction, since our problem is Bayesian and the agent has multiple types, the exponentially decaying probabilities, costs and exponential types are constructed to  penalize an agent's significant deviation from designated actions for his particular type; they do not show up in the non-Bayesian setting of \cite{dutting2019simple}.  However, the construction of exponential rewards and costs in \cite{dutting2019simple} is due to the geometry property of linear contracts, i.e., piece-wise linearity of the objective function. They want to maintain a specific form of linear contract so that choosing any action will only result in utility $1$, but the welfare increases as the action increases.
}


 

Algorithmic studies of contract design under both adverse selection and moral hazard are also rich  \cite{xiao2020optimal,guruganesh2021contracts,alon2021contracts,guruganesh2034contracts,castiglioni2021bayesian,castiglioni2022designing,alon2022bayesian,gan2022optimal,zhu2023sample}.
In the  general multi-parameter setting, \citet{guruganesh2021contracts} studied the approximability of linear contracts to single contract, menu of contracts and welfare. Later,  \citet{guruganesh2034contracts} further pursued the approximability of menu of randomized linear contracts. In terms of the hardness results in a multi-parameter setting, \citet{castiglioni2021bayesian} shows that it is NP-hard to compute the optimal single contract, while \citet{castiglioni2022designing} shows that it is also NP-hard to compute the optimal menu of contracts.  \citet{castiglioni2022designing}
 and \citet{gan2022optimal} design polynomial-time algorithms for computing the optimal menu of randomized contracts. In the single-parameter setting, 
\citet{alon2021contracts} characterize the implementable conditions of contract design by considering the dual program. \citet{alon2022bayesian} later studied the problem of linear contract approximating other contracts in the same setup.

%% file: prelim.tex
\section{Preliminaries: Single-Parameter Bayesian Contract Design (BCD) }\label{sec:prelim}

In this section, we  describe a natural  single-parameter Bayesian Contract Design (BCD)  problem, as studied in multiple recent works   
\cite{chade2019disentangling,alon2021contracts,alon2022bayesian}. This will be the major model of our interest, though in Section \ref{section:equvivalentce} we will also describe its multi-parameter generalization before showing how the multi-parameter setup reduces to the following single-parameter setup, despite   seemingly much more complex problem descriptions.  Throughout the paper, we use notations $\langle n\rangle = \{0, 1, 2, \dots, n-1\}$ and $[n] = \{ 1, 2, \dots, n\}$.

This setting features a \emph{principal} and an \emph{agent}. The principal (she) seeks to incentivize the agent (he) to act on her behalf. There are $n$ actions  available for the agent to choose, where $n$ is a finite value. Each action $a_i$, $i \in  \langle n\rangle$ ,  incurs  $c_i$ \emph{units} of cost to the agent, and induces a distribution $F_i \in \Delta^{m}$ over $m  $ outcomes, where $m$ is also finite. Let $F_{i, j}$ denote the probability that the $i^{\textnormal{th}}$ action $a_i$ results in \emph{outcome} $j$. Let $\Omega \triangleq [m] = \{ 1, \cdots, m \}$ be the set of $m$ outcomes.
Each outcome 
{$j\in \Omega$ } will induce a reward $r_j\ge 0$ to the principal. The cost vector $c = [c_0, c_1, c_2, \dots, c_{n-1}]$ for actions, the distributions {$\{ F_i \}_{i \in \langle n\rangle }$ }, and the reward vector {$r = [r_1, r_2, \dots, r_m]$} are all public knowledge. Note that for exposition convenience, we sometimes also  use an action to directly index the distribution and cost, e.g., $F_{a_i} = F_i$ and $c_{a_i}=c_i$.


In BCD, the agent is assumed to be drawn from one of $K$  types in $\Theta = \{\theta_1, \theta_2, \dots, \theta_K\}$. The setting is \emph{single-parameter} in the sense that the agent's type $\theta_k \ge 0$ is a single real value that only affects agent's action cost. Formally,  the agent of type $\theta_k$ suffers cost  $\theta_k \cdot c_i$ for playing action $i$. Let $\mu(\theta)$ denote the probability of agent type $\theta \in \Theta$; w.l.o.g., assume $\mu(\theta) > 0$. While the agent  knows his own private type, the principal does not observe $\theta$ and only knows the distribution $\mu$ over the agent's types.  

In standard contract design models, the principal is not able to observe the agent's action but only observes the probabilistic outcome it produces. To incentivize desirable agent actions, the principal can only set outcome-dependent payments. Thus, a \emph{contract} designed by the principal is a vector {$p = [p_1, p_2, \dots, p_m]$ }where $p_i$ denotes the payment to the agent upon realizing outcome $i$. Given any contract $p$, the agent of type $\theta$  plays his {\it best-response} action. 
The principal's objective is to maximize her own utility by designing (possibly multiple and randomized) contracts. 


\begin{definition}\label{defintion_single_best_res_act}
    (Best-response Action) Given any contract $p \in \mathbb{R}^{m}$, the agent of type $\theta$  plays his {\it best-response} action defined as
    \[
    a(\theta, p) = \arg\max_{i \in \langle n\rangle} \big[ \langle F_i, p \rangle - \theta\cdot c_i \big],
    \]
    where we assume type-breaking is in favor of the principal (see, e.g., \cite{alon2022bayesian,guruganesh2021contracts}). 
\end{definition}

Various formats of contracts have been studied in the literature. In this paper, we primarily focus on two basic contract formats, i.e., a  \emph{menu of contracts} and a \emph{single contract}, distinguished by whether they are allowed to depend on the buyer's type. We start by introducing the former.

  \vspace{2mm}
\noindent 
\textbf{Optimal Design of a Menu of Contracts.} 
In this case,  the principal can generally offer multiple contracts --- i.e., a {\it menu of contracts} --- and let the agent pick his favorite one. We denote this menu of contracts by $P$ consisting of multiple contracts $p \in \mathbb{R}^m$.  
Under a contract $p$, the type-$\theta$ agent  receives expected payment $\sum_{\omega \in [m]} F_{a(\theta, p),\omega}p_\omega = \langle F_{a(\theta, p)}, p \rangle $ by playing the best-response action $a(\theta, p)$.
Consequently, the expected utility that the principal gets from the agent of type $\theta$ is $\langle F_{a(\theta, p)}, r - p \rangle = \sum_{\omega\in[m]} F_{a(\theta, p), \omega} (r_\omega - p_\omega)$, where $p$ is the contract chosen by type $\theta$.

Notably, although the agent can choose a contract $p$ and resultant action $a(\theta, p)$ to play, one could equivalently view the menu as a \emph{contract mechanism} $(P, a)$ (also more generally known as a coordination mechanism  \cite{myerson1982optimal,gan2022optimal}) that elicits the agent's type $\theta$ first, and then assigns to the agent a contract $p^{\theta}$ and a recommended action $a(\theta)=a(\theta, p^{\theta})$, where $P=(p^\theta)_{\theta \in \Theta}$ $a=(a(\theta))_{\theta\in \Theta}$. 
Under this contract mechanism $(P, a)$, the agent has two possible types of deviation: misreporting his private type and not obeying the recommended action. These two types of behavior correspond to {\it adverse selection} and {\it moral hazard}, respectively. Hence the incentive constraints need to take   both into account, leading to the following definition of \emph{incentive compatibility} and its natural relaxation~\cite{myerson1982optimal,castiglioni2022designing,gan2022optimal}.


\begin{definition}[Incentive Compatibility (IC) \cite{myerson1982optimal}]\label{def:IC}
A contract mechanism $(P, a)$ is \emph{incentive compatible} for the Bayesian contract design problem if each agent type $\theta$ maximizes his expected utility by \emph{honestly} reporting his type $\theta$ and \emph{obediently} take the recommended action $a(\theta)$. Formally, 
\[   
   \langle  F_{a(\theta)},  p^\theta \rangle   - \theta \cdot c_{a(\theta)} \ge  \langle F_{i },  p^{\theta'} \rangle   - \theta\cdot c_{i}, \, \, \,   \forall~ i\in \langle n\rangle, \theta, \theta' \in \Theta \quad. 
\] 
Furthermore, when the above inequality is relaxed to $\langle  F_{a(\theta)},  p^\theta \rangle   - \theta \cdot c_{a(\theta)} \ge  \langle F_{i },  p^{\theta'} \rangle   - \theta\cdot c_{i} - \eta$
for some $\eta\geq0$, we say $(P, a)$ is $\eta$-IC. 
\end{definition}

A simple revelation-principle-style argument \cite{myerson1982optimal,alon2021contracts,castiglioni2022designing} shows that, among all menu of contracts, it is without loss of generality to design incentive compatible menus, where the optimal menu needs not to use more than $K$ contracts, with contract $p^\theta$ for agent type $\theta \in \Theta$. Thus, a generic menu of contracts can be represented by $P = \{p^{\theta_1}, p^{\theta_2}, \dots, p^{\theta_K}\}$.  Following the standard literature (see, e.g., \cite{alon2021contracts,carroll2015robustness}),  we impose the \emph{limited liability} constraint $p \in \mathbb{R}_+^m$ for each   contract for most of the paper, though some of our later results will explore the possibility of relaxing this constraint.

It is also natural to impose the  individually-rationality (IR) constraint; that is, each agent type should gain non-negative expected utility by participation. A convenient assumption that can avoid  considering this IR constraint separately is to assume the existence of an \emph{opt-out} action which has $0$ cost. Notably, this assumption is without loss of generality but rather for mathematical convenience  because the IR constraint is mathematically equivalent to having an opt-out action with $0$ cost and $0$ payment. Therefore, following the standard literature \cite{dutting2019simple,castiglioni2021bayesian,castiglioni2022designing}, we also adopt this convenient assumption (see below) throughout this paper for both single-parameter and multi-parameter settings. 
\begin{assumption}\label{assumptionoptouteachtype}
Action $a_0$ is an   opt-out action with zero cost, i.e., $c_0 = 0$.
\end{assumption}
Under Assumption \ref{assumptionoptouteachtype}, the IR constraint is naturally satisfied due to IC requirement and non-negative payments $p\ge 0$ (i.e., limited liability). Then, the problem of designing an optimal IC menu for the principal can be formulated as the following optimization problem. We denote its optimal objective value as $OPT_{\text{menu}}$. 
    \begin{eqnarray}\label{op:opt-menu}
OPT_{\text{menu}}: \quad \max_{P \in \mathbb{R}^{m \times K},\ a \in \langle n \rangle^{K}  } && U(P, a) \triangleq  \sum_{\theta\in \Theta}  \mu(\theta) \cdot   \langle  F_{a(\theta)}, \,  r  -  p^\theta \rangle   \notag\\ \label{menuofcontract}
  \text{s.t.} &&  \langle  F_{a(\theta)},  p^\theta \rangle   - \theta \cdot c_{a(\theta)} \ge  \langle F_{i },  p^{\theta'} \rangle   - \theta\cdot c_{i}, \, \, \,   \forall~ i\in \langle n\rangle, \theta, \theta' \in \Theta \quad  \\
&& p^\theta_\omega \ge 0, \quad \forall~ \theta \in \Theta, \omega \in [m]\notag
\end{eqnarray}
Note that when the first constraint is applied with $\theta = \theta'$, it guarantees that it is optimal for type $\theta$ agent to report his private type truthfully and  $a(\theta)$ is a best response for type $\theta$ under contract $p^{\theta}$.

  \vspace{2mm}
\noindent 
\textbf{Optimal Design of a Single Contract.} A single contract can be seen as a menu of contracts with an additional constraint that all the proposed contracts are equal, i.e., $p^{\theta_1} = p^{\theta_2}=\dots =p^{\theta_K}$. This is equivalent to designing a single contract $p \in \mathbb{R}^m$.  
Similar to menus of contracts, we can equivalently view a contract as a \emph{contract mechanism} $(p,a)$, where $p \in \mathbb{R}_+^m$ is a contract and $a=(a(\theta))_{\theta\in \Theta}$ is a tuple of action recommendations. Definition \ref{def:IC} (IC) can be similarly applied here, except that type misreporting is useless here since every agent receives the same contract. Hence the IC constraint naturally degenerates to obedient action recommendation.  
This leads to the following optimization program  that finds an optimal single contract for the principal. We denote its optimal objective value as $OPT_{\text{single}}$. 
    \begin{eqnarray}
 OPT_{\text{single}}: \, \,  \max_{p\in \mathbb{R}^m, a \in \langle n \rangle^{K} } && U(p, a) \triangleq 
 \sum_{\theta \in \Theta} \mu(\theta) \cdot \langle F_{a(\theta)}, \,  r - p \rangle  \notag \\ \label{OPTofcontract}
\text{s.t.} && \langle F_{a(\theta) }, p \rangle  - \theta\cdot c_{a(\theta)} 
  \geq  \langle F_{i }, p \rangle  - \theta\cdot c_i   ,  \quad \forall~ i \in \langle n\rangle,  \theta \in \Theta \quad \\
&& p_\omega \ge 0, \quad \forall~  \omega\in[m] \notag 
\end{eqnarray}
The first constraint ensures that $a(\theta) \in \langle n\rangle$ is indeed a best response for type $\theta$, while the second constraint expresses limited liability. 
It is worth noting the difference between Program~\eqref{OPTofcontract} and Program~\eqref{menuofcontract}. Besides a richer contract design space, the  incentive constraint of Program  ~\eqref{menuofcontract} considers agent $\theta$'s potential misreporting of type $\theta'$. However, this constraint is not present in Program~\eqref{OPTofcontract}.

\begin{remark}
Another strictly more general format of contracts is the menu of \emph{randomized} contracts, which is the most powerful class of contract mechanisms for optimizing the principal's utility \cite{castiglioni2022designing,gan2022optimal}. Despite its power, such randomized contracts are much less studied in the literature, compared to a single contract or a menu of contracts that we focus on. This is partially due to  concerns and potential impracticability  of using lotteries for determining contracts.  Nevertheless,  in Section~\ref{sec:power-of-single}, we will use the principal's utility under the optimal menu of randomized contract (defined in Section~\ref{sec:power-of-single}) as a benchmark and  look to identify conditions under which a single contract can be as powerful as the most general contract format, i.e., the menu of randomized contracts.   
\end{remark}

  \vspace{2mm}
\noindent 
\textbf{Timeline of the Principal-Agent Interaction.} We conclude the model description by summarizing the interaction timeline between the principal and the agent. 
\begin{enumerate}
    \item The principal designs and commits to a  contract mechanism, being a single contract $(p,a)$ or a menu of contracts $(P,a)$. 
    The designed contract is publicly known.
    \item The agent realizes his private type $\theta$ and reports a type $\theta'$ to the principal. The principal then applies a contract (generically denoted as $p$)  conditioned on $\theta'$.
    \begin{itemize}
        \item In the single-contract case, the principal applies the contract $p$ and recommends action $a(\theta')$;
        \item In the menu-of-contract case, the principal applies the contract $p^{\theta'}$ and recommends action $a(\theta')$;
    \end{itemize}
    
    \item After receiving the contract $p$, the agent takes the best-response action $a$ (not necessarily the recommended action) that maximizes his own utility. This action has cost $\theta \cdot c_a$ and leads to outcome $\omega$ with probability $F_{a, \omega}$. The principal then receives reward $r_\omega$ whereas the agent receives payment $p_{\omega}$. 
\end{enumerate}

%% file: gap.tex
\section{Warm-up: An Instance with Tight $OPT_{\text{menu}}/OPT_{\text{single}}$ Ratio}\label{sec:ratio-n}


{As a warm-up for our reduction in the next section, we start by constructing one instance of single-parameter Bayesian contract design (BCD)   problems which turns out to achieve an order-wise \emph{tight} ratio of $\Theta(n)$  between the revenue of optimal menu  $OPT_{\text{menu}}$  and that of optimal single contract $ OPT_{\text{single}}$. Recall that $n$ here is the number of the agent's actions. 

The reason for presenting this warm-up instance is two-fold. First, the key conceptual ideas for constructing this   instance turns out to be instrumental  for establishing a polynomial-time  reduction from   multi-parameter BCD to single-parameter BCD. At a high level, our idea is to carefully construct  instance parameters (i.e., action costs, efficacy and type distributions) so that any reasonably good menu of contracts 
shall incentivize each agent type to take certain designated ``profitable'' action whereas  any single contract will only be able to incentivize very few of these profitable actions. During our proof development, we found that carefully designed \emph{exponentially-decaying} parameters are very useful for enforcing such desired properties. These also serve as the key starting points for our reduction in Section \ref{section:equvivalentce}, though completing the reduction there will require significantly more developments. Notably, while the result in this section is ``easier'' to show, it is already  interesting and non-trivial. Indeed, the second reason for showing this warm-up construction is that the result itself already answers an open question of \cite{guruganesh2021contracts} for the single-parameter BCD setting. Specifically, \citet{guruganesh2021contracts} studied the power of different formats of contracts in  multi-parameter Bayesian contract design. They show that the menu of contracts can achieve at most  $O(nK)$ times of the optimal revenue from a single contract.\footnote{\citet{guruganesh2021contracts} also show that the $O(nK)$ upper bound can be improved to $O(n \log K)$ with two assumptions: (1) uniform type distribution; (2) type-independent costs.}  A major open question   highlighted in \cite{guruganesh2021contracts} is to pin down  the $OPT_{\text{menu}}/OPT_{\text{single}}$ ratio. In a follow-up work, \citet{guruganesh2034contracts} present a multi-parameter instance in which the $OPT_{\text{menu}}/OPT_{\text{single}}$ ratio is at least $\Omega(\max (n, \log K))$. What the order-wise tight ratio is still  remains open. Our result in this section shows that, in    single-parameter Bayesian contract design, the tight ratio is $\Theta(n)$. Interestingly, this ratio is independent of $K$ whereas a tight ratio for the general multi-parameter setting necessarily depends on $K$ due to the known upper and lower bounds above.\footnote{This does not conflict with our later reduction from multi-parameter  to single-parameter setups since the reduction is a polynomial-time reduction, hence does not preserve the exact order of the game parameters. Indeed, the single-parameter instance we construct will have polynomially-many more actions.}
}





	\begin{theorem}\label{upperlowertheoremsingleparameter}
In single-parameter Bayesian contract design, the principal’s utility gap between the optimal menu of contracts and optimal single contract is $\Theta(n)$ where $n$ is the number of agent actions. 
  
Formally, for any single-parameter Bayesian contract design instance we have $\frac{ OPT_{\text{menu}}}{OPT_{\text{single}}} = O(n)$, and there exist single-parameter instances such that $\frac{ OPT_{\text{menu}}}{OPT_{\text{single}}} = \Omega(n)$.

	\end{theorem}

	\begin{proofsketchof}{Theorem \ref{upperlowertheoremsingleparameter}}
  The upper bound  $O(n)$ in the theorem is a simple corollary of \citet{alon2022bayesian} (Theorem 4), which states that a linear contract achieves expected revenue that is an $n$-approximation to the optimal expected virtual welfare. This   implies  $\frac{OPT_{\text{menu}}}{OPT_{\text{single}}} \le \frac{OPT_{\text{virtual}}}{OPT_{\text{linear}}} \le O(n)$, where the first inequality holds since (i) the principal's utility from a single contract ($OPT_{\text{single}}$) serves as a natural upper bound for the linear contract ($OPT_{\text{linear}}$) case, and (ii) optimal expected virtual welfare ($OPT_{\text{virtual}}$)\footnote{A formal definition of virtual welfare is not needed here for our proof, but it can be found in Section~\ref{sec:power-of-single} later.}  is an upper bound of the optimal principal's utility of menu of contracts ($OPT_{\text{menu}}$) by Corollary 2 in \cite{alon2022bayesian}. 
  
  \vspace{2mm}
\noindent 
\textbf{The main argument: proving lower bound $\Omega(n)$ via instance construction.}  We now move to the more challenging direction of the  theorem, i.e., proving an $\Omega(n)$ lower bound via construction.  In particular, we show that for each  number of actions {$n\ge 3$}, there exists an instance in which the ratio regarding the principal's utility between the menus of contracts and a single contract is at least $n/24$. Without loss of generality, we assume $n$ is odd.\footnote{Otherwise, one can add a trivial action with $0$ cost that results in a $0$-reward outcome with probability $1$.}   Let   $\bar n=(n-1)/2$ and $l = 2n^2$. The detailed instance construction is as follows (see Table \ref{constructionsinglapproximatmenutable} for an illustration).
\begin{itemize}
    \item $m = \bar n + 2$ \textbf{outcomes: }  an outcome $\omega_i$ for each $i \in [\bar n]$ and two additional outcomes $  \omega_{\emptyset}$ and $  \omega_+$.    
	Outcome $  \omega_+$ is the only beneficial outcome and results in principal reward $1$. The principal's reward for all the other outcomes is $0$. Notably,  $  \omega_{\emptyset}$ is a \emph{dummy} outcome, introduced mainly to saturate any remaining probability mass in our designed outcome distributions conditioned on an action. 

 \item $n = 2 \bar n +1$ \textbf{actions: } two actions $\{a_{i,1},a_{i,2}\}$ for each $i\in [\bar n]$, and one additional action $a_0$. 
 
 For each $i \in [\bar n]$, action $a_{i, 1}$ can only induce one of two \emph{non-dummy} outcomes, i.e., $\omega_+, \omega_i$, with probabilities $F_{a_{i,1},\omega_{+}} = F_{a_{i,1},\omega_i}=2^{-il}$. The remaining probability mass is assigned to the dummy outcome $\omega_{\emptyset}$, i.e.,  $F_{a_{i,1},\omega_{\emptyset}}=1-2^{-il+1}$. The cost of action $a_{i,1}$ is $2^{-2il}(1-2^{-in})$. Intuitively, as $i$ increases, the cost of action $a_{i,1}$  decreases exponentially, but so does its efficacy in inducing the only profitable outcome $\omega_+$.  
 Action $a_{i,2}$ is constructed as a competing option to action $a_{i,1}$ --- it leads to many more possible outcomes than   $a_{i,1}$ but has a slightly higher cost. Specifically, for each $i \in [\bar n]$, we set $F_{a_{i,2},\omega_j}=2^{-il}$ for each $j \in [\bar n]$, $F_{a_{i,2},\omega_{+}}=2^{-il}$, and $F_{a_{i,2},\omega_{\emptyset}}=1-2^{-il} (\bar n+1)$. The cost of action $a_{i,2}$  is $2^{-2il}$. 
Finally, action $a_0$ has cost $0$ and trivially always leads to outcome $\omega_{\emptyset}$, i.e., $F_{a_0,\omega_{\emptyset}}=1$. 


\item $K = \bar n$ \textbf{agent types}: a type $\theta_i$ for each $i\in [\bar n]$. Each Type $\theta_i$ is set as $\theta_i=2^{il}$, and the probability of type $\theta_i$ is $\mu(\theta_i)= 2^{in+il}/C$ where $C=\sum_{i \in [\bar n]} 2^{in+il}$ is a normalization constant used to normalized the cumulative probability to $1$. That is, as $i$ increases, the cost of type $\theta_i$ increases exponentially, and so does the probability of having such (inefficient) agent's type.   
\end{itemize} 


\begin{table}[t]
\centering
\renewcommand\arraystretch{1.5}
\begin{tabular}{c|c|c|c|c|c}\hline
Outcomes $\omega$ & $\omega_{\emptyset}$ & $\omega_{+}$ & $\omega_i$ & $\omega_j, j\neq i, j\in[\bar{n}]$ & \multirow{2}{*}{Cost $c$}  \\ \cline{1-5}
Rewards $r$ & $0$ & $1$ & $0$ & $0$ & \multirow{2}{*}{} \\ \hline 
Action $a_{i, 1}, i\in [\bar{n}]$ & $1-2^{-il + 1}$ & $2^{-il}$ & $2^{-il}$ & 0 & $2^{-2il}(1-2^{-in}), i \in [\bar n]$\\ \hline
Action $a_{i, 2}, i\in [\bar{n}]$ & $1-2^{-il}(\bar{n}+1)$ & $2^{-il}$ & $2^{-il}$ & $2^{-il}, j\in[\bar{n}]$ & $2^{-2il}, i \in [\bar n]$ \\ \hline
Action $a_0$ & $1$ & $0$ & $0$ & $0$ & $0$ \\ \hline
\end{tabular}
\caption{Illustration of the constructed instance in Theorem \ref{upperlowertheoremsingleparameter}. There are $n = \bar{n}\times 2 +1$ actions,  $m = \bar{n} + 2$ outcomes, and $K = \bar n $ agent types; type $\theta_i = 2^{il}$ has probability $\mu(\theta_i)$ proportional to $2^{in + il}$.}
\label{constructionsinglapproximatmenutable}
\end{table}

We now present the high-level idea of proving the lower bound $\Omega(n)$ for the above instance. The proof mainly relies on the following Lemma~\ref{menusingleratiolowerbound} and Lemma  \ref{menutosingleratioupperbound}. 

  		\begin{lemma}\label{menusingleratiolowerbound}
		In the constructed instance,    the principal's expected utility from an optimal menu of contracts is at least $\frac{\bar n}{2C} $ where $C = \sum_{i \in [\bar n]}2^{in+il}$. 
		\end{lemma}

\begin{lemma}   \label{menutosingleratioupperbound}
In the constructed instance,    the principal's   expected utility from the optimal single contract is at most  $3/C$ where $C =  \sum_{i \in [\bar n]}2^{in+il}$.
\end{lemma} 
  
We defer the formal proofs of the two lemmas to Appendix \ref{sec:append:gap} but describe their core intuitions here.  Lemma \ref{menusingleratiolowerbound} provides a lower bound on the optimal principal's utility employing a menu of contracts. Intuitively, to achieve a large utility each contract in the menu should incentivize the agent to choose a different action. Hence, for each type $\theta_i$, we construct a contract with payment only on outcome $\omega_i$ and that incentivizes action $a_{i, 1}$. This turns out to be possible  due to our carefully designed exponentially-decreasing  action cost and efficacy in inducing outcome $\omega_i$. 

Lemma \ref{menutosingleratioupperbound} shows an upper bound on the utility of any single contract. The proof hinges on a carefully defined set $K(p)$, which contains all types $i \in [\bar n]$ who provide expected principal utility at least $2^{-in-il-n/2}$. We then prove that the size of $K(p)$ must be small. 
  The intuition is the following: the dummy outcome $\omega_{\emptyset}$  gives nothing in return to the principal, so the payment to outcome $\omega_{\emptyset}$ should be sufficiently low. Then, to obtain high principal utility, the payment to outcome $\omega_{+}$ should be low enough since it is the only outcome that gives a positive reward, hence must be induced by many different agent's types. A crucial point is that different types should be incentivized with different payments and outcome $\omega_+$ does not provide such flexibility. Then, to incentivize agent $i\in K(p)$ to take action $a_{i,1}$ the payment to outcome $\omega_i$ should be sufficiently high. However, if $K(p)$ contains many types, which implies that there are many outcomes $\omega_j$, $j\in K(p)$, that provide sufficiently high payments, a type-$i$'s agent, $i \in K(p)$,  may then instead choose action $a_{i, 2}$ which leads to non-positive utility to the principal. This implies that under an optimal single contract, the principal will achieve small utility from many agents, which leads to her utility upper bound.

Finally, we   combine Lemma~\ref{menusingleratiolowerbound} and Lemma~\ref{menutosingleratioupperbound} to  conclude the proof since the ratio between the utility of an optimal menu $OPT_{\text{menu}}$ and an optimal contract $OPT_{\text{single}}$ is lower bounded as follows:
    \[
    \frac{OPT_{\text{menu}}}{OPT_{\text{single}}} \ge \frac{\frac{\bar n}{2} \frac{1}{\sum_{j \in [\bar n]}2^{jn+jl}}}{ 3 \frac{1}{\sum_{j \in [\bar n]} 2^{jn+jl}}} = \frac{\bar n}{6}=\frac{n-1}{12}\ge \frac{ n}{24}.\]	
\end{proofsketchof}

%% file: reduction.tex
\section{{A Reduction from Multi-Parameter to Single-Parameter BCD}}\label{section:equvivalentce}


In this section, we present our formal reduction from \emph{multi-parameter} Bayesian contract design (BCD) to   \emph{single-parameter} BCD.  
We   start by describing the instance of a general multi-parameter BCD problem, denoted as $I^M$, which  naturally generalizes  the single-parameter BCD problem of Section \ref{sec:prelim}.  The main result of this section will  show that,  despite seemingly being much more complex than   single-parameter BCD, the   general multi-parameter BCD problem is actually no harder computationally; or equivalently,   single-parameter BCD is no easier. This is the case even when approximate design is sought. We do so by proving that the two problems reduce to each other.

{
\vspace{2mm}
\noindent \textbf{Multi-parameter  Bayesian Contract Design.}  The most general BCD setting   generalizes the previous single-parameter setting of Section \ref{sec:prelim} in two  ways. First, the costs for each agent type $\theta$  here can be any   vector $c^{\theta} = (c^{\theta}_0, \cdots, c^{\theta}_{n-1})$ with $n$ arbitrary non-negative parameters, whereas the costs in single-parameter BCD is governed by a single value, i.e., the type $\theta_k$. Second, in this general setup, the action-to-outcome transition matrix $F$ is  allowed to depend on the agent's type $\theta$ in an arbitrary way, hence are denoted by $F^{\theta}$. This setting is  ``multi-parameter'' because the type $\theta$ prescribes many parameters -- i.e., $nm$ parameters in $F^{\theta}$ and $n$ parameters in $c^{\theta}$ -- to describe  $\theta$'s characteristics. The incentive constraints, best response actions,   formats of contracts and timelines for   multi-parameter BCD are all defined in almost the same way, except that all matrix/vector $F$ will be superscripted with its corresponding agent type $\theta$ and all cost notation will generalize from $\theta  c_i$ to $c^{\theta}_i$. We hence do not repeat  the same descriptions here, but only present the following generalized version of  Program \eqref{op:opt-menu} to multi-parameter case as an illustrative example.  The program computes the optimal menu of contracts for multi-parameter BCD and is useful in our proofs.  
\begin{equation}
\begin{aligned} \label{multiparameteroptprob}
    \max_{P \in \mathbb{R}^{m \times K},\ a^m \in \langle n \rangle^{K}  }  \quad & \quad \sum_{\theta} \mu(\theta)  \cdot \langle  F^\theta_{a^m(\theta)}, r  - p^{\theta} \rangle \\
    \text{s.t.} \quad & \quad \langle  F^\theta_{a^m(\theta) },  p^{\theta} \rangle  -c^\theta_{ a^m(\theta)} \ge  \langle  F^\theta_{ i}, p^{\theta'} \rangle -c^\theta_{i}, \quad \forall i \in \langle n\rangle, \theta, \theta' \in \Theta \\
    &\quad p^\theta_\omega \ge 0, \quad \forall~ \theta \in \Theta, \omega \in [m] 
\end{aligned}
\end{equation}

\vspace{2mm}
\noindent \textbf{Remarks on notations and approximation. } To avoid confusion in following derivations, we will use $a^s(\theta, p)$ and $a^m(\theta, p)$, respectively for the single-parameter BCD and multi-parameter BCD settings, to denote type $\theta$'s best response action to a contract $p$,  as in Definition~\ref{defintion_single_best_res_act}.  This is also why the action recommendations are denoted as $a^m$ in Program \eqref{multiparameteroptprob}. For ease of exposition, we assume  all rewards $r_j$ and costs $c_i^{\theta}$ are bounded between $[0,1]$, but all our results easily generalize so long as all these values are upper bounded by some value $B$. Finally, following conventional terminology, $\beta$-approximation means approximating the optimal value to be within a multiplicative factor of at least $\beta$ whereas ``$\epsilon$-optimal''  means approximating the optimal value to be within an additive loss of at most $\epsilon$.
}

Now we are ready to present our main result of this paper,  an ``almost'' approximation-preserving reduction from multi-parameter   to single-parameter Bayesian contract design. 
        

\begin{theorem}\label{equivalencesinglemulti}
Multi-parameter Bayesian contract design (BCD) and single-parameter BCD reduce to each other in polynomial time. 
   Specifically, for any $\xi>0$, there is an algorithm that reduces \emph{any} multi-parameter BCD instance $I^M$ with $n$ actions, $m$ outcomes, $K$ agent types, and total input bit length $L$ to  a single-parameter instance $I^S$ with $m+1$ outcomes, $nK+1$ actions and $K+1$ agent types  in $\poly(L, \log(1/\xi))$ time in the following sense: any $\beta$-approximation of the optimal menu of contracts (resp. single contract) for the constructed instance $I^S$  can be converted, again in $\poly(L, \log(1/\xi))$ time, to an $(\beta-\xi)$-approximation of the optimal menu of contracts (resp. single contract) for the original multi-parameter instance $I^M$.
\end{theorem}



The following corollary shows that if the given   contract for the constructed single-parameter instance $I^S$ is exactly optimal  (i.e., $\beta=1$), then the reduction will also be exact. That is, the algorithm will return an exactly optimal contract to the multi-parameter instance $I^M$ in $\poly(L)$ time.

\begin{corollary}\label{exactlyreducetosignlemulti}[The Exact Optimality Version] 
  By the same reduction as Theorem \ref{equivalencesinglemulti}, there is an algorithm that reduces in polynomial time \emph{any} multi-parameter BCD instance $I^M$ to  a single-parameter instance $I^S$ in $\poly(L)$ time in the following sense: any optimal (i.e., $\beta=1$) menu of contracts (resp. single contract) for $I^S$  can be converted, again in $\poly(L)$ time, to an optimal menu of contracts (resp. single contract) for the original multi-parameter instance $I^M$.
\end{corollary}

 Note that single-parameter Bayesian contract design trivially reduces to multi-parameter setting, thus  Theorem \ref{equivalencesinglemulti} essentially established the approximation-preserving equivalence between multi-parameter and single-parameter BCD for the problem of finding optimal menu of contracts or optimal single contract.  
The remainder of this   section is devoted to this reduction, and is arranged as follows. In Subsection~\ref{sec:reduction-cors}, we discuss a few important implications of Theorem~\ref{equivalencesinglemulti}, including resolving an open question of \cite{alon2021contracts}. We then proceed to proving the theorem. Due to its quite involved proof procedure, we divide the arguments into four subsections: Subsection~\ref{sec:reduction-construct} presents our construction of a reduction from any multi-parameter setup to a single-parameter setup; Subsection \ref{sec:reduction-main} presents the high-level idea of the proof and our main arguments; finally, Subsection \ref{sec:reduction-lemA} and \ref{sec:reduction-lemB} present the proofs of the two core technical lemmas, respectively. 

\subsection{Implications of Theorem \ref{equivalencesinglemulti} }\label{sec:reduction-cors}



As straightforward corollaries, Theorem \ref{equivalencesinglemulti} implies that multiple existing hardness results for multi-parameter BCD (e.g. \cite{castiglioni2021bayesian,castiglioni2022designing}) can be directly applied to single-parameter settings. The following two corollaries show the inapproximability of single-parameter Bayesian contract design, for both the optimal menus of contracts and single contracts .  These corollaries also answer the open problems of \cite{alon2021contracts}, in the negative, about the possibility of designing efficient algorithms for single-parameter Bayesian contract design even under regular agent type distributions. {Recall from \cite{alon2021contracts}, a type distribution is \emph{regular} in BCD if the virtual cost function is monotone.}    

\begin{corollary}\label{cor1} 
     In   single-parameter Bayesian contract design, for any $\delta\in (0, 1]$, it is NP-hard  to compute a $\frac{1}{K^{1-\delta}}$--approximation to the optimal single contract, where $K$ is the number of types. {This hardness holds even when agent's type distribution is regular.}  
\end{corollary}

\begin{proof}
    \citet{castiglioni2021bayesian} show that in the multi-parameter setup, for any $\delta'>0$,  it is NP-hard  to design a single contract that guarantees a  $\frac{1}{K^{1-\delta'}}= K^{-1+\delta'}$ fraction of the optimal solution, where there are $K-1$ types in the multi-parameter instance. \footnote{Notice that even if the factor of approximation in \citep{castiglioni2021bayesian}   is $O((K-1)^{\delta'-1})$, the same result holds for $K^{\delta'-1} $  applying the reduction with a slightly smaller parameter $\delta''$. We intentionally use $K-1$ types since our reduction from multi-parameter to single-parameter instance (in Section \ref{sec:reduction-construct}) adds one additional type.} 
    Hence, to prove the statement, it is sufficient to show that we can reduce in polynomial time a $K^{\delta-1}$-approximate single contract for the single-parameter setting to   a $K^{\delta'-1}$-approximate  contract for the multi-parameter setting, where $\delta'=\delta/2$ and $K$ is the number of types in the single-parameter instance. By  Theorem \ref{equivalencesinglemulti}, letting $\beta = K^{\delta-1}$ and $\xi = K^{\delta-1} - K^{\delta'-1} > 0$, we can use such an algorithm to provide a $ (K-1)^{-1+\delta'}$ approximation to the multi-dimensional problem, proving the statement.

    {To see the above hardness result holds for regular (discrete) type distributions, we observe that the type distribution used by \citet{castiglioni2021bayesian} for multi-parameter instance is a uniform type distribution. We will verify in  the following subsection \ref{sec:reduction-construct} that the resultant single-parameter instance from our multi-parameter to single-parameter reduction will have regular type distribution.}  
\end{proof}

\begin{corollary}\label{cor2}
	In   single-parameter Bayesian contract design, for any constant $\rho>0$, it is NP-hard to compute  a $\rho$-approximation to the optimal menu of contract. {This hardness holds even when agent's type distribution is regular.}
\end{corollary}

\begin{proof}
     \citet{castiglioni2022designing} show that in the multi-parameter setup, for any $\rho'>0$,  it is NP-hard  to design a menu of contracts that provides a $\rho'$-approximation of the optimal solution. 
    To prove the statement, it is sufficient to show that we can reduce in polynomial time any $\rho$--approximation algorithm for the optimal menu of contracts in single-parameter BCD    to  a $\rho'$-approximation algorithm for the multi-parameter setting, where $\rho' = \rho/2 >0 $. By letting $\beta=\rho$ and $\xi = \rho'$,  Theorem \ref{equivalencesinglemulti} implies
  such an algorithm hence proves the corollary. {Similar to the proof of Corollary \ref{cor1},  this hardness result holds for regular (discrete) type distributions because the type distribution used by \citet{castiglioni2022designing} for multi-parameter instance is also a uniform type distribution.   } 
\end{proof}

\subsection{Construction of the Reduction }\label{sec:reduction-construct}

 Theorem \ref{equivalencesinglemulti} is proved by a novel and   non-trivial reduction from multi-parameter to single-parameter settings. Recall that the single-parameter instance has a   common probability matrix and cost vector for every agent type. To construct such a single-parameter instance $I^S$ from given multi-parameter instance $I^M$, one natural way is to stack  up all the probability matrices $F^{\theta}$ and  cost vectors $c^{\theta}$ of different types. However, since now there are more available actions in this constructed common probability matrix for each agent, the optimal solution can be significantly different from that of the original multi-parameter instance $I^M$, hence it appears impossible to establish their connections via this straightforward construction. To prevent the above chaos, we draw inspiration from our proof of Theorem \ref{upperlowertheoremsingleparameter}, where we managed to utilize   exponentially increasing/decaying weights (i.e., $2^{f(k, l)}$ for some function $f(\cdot, \cdot)$ and parameters $k,l$) to successfully  ``restrict'' each type $\theta_k$ agent to choose an action from some designated action set. Our construction here also starts with carefully exploiting these exponentially-changing weights. The guiding principle for our reduction   is that by utilizing these  exponential functions to reweigh the probability matrices $F^{\theta}$ and cost vectors $c^{\theta}$, we   wish that type $\theta$ agent in a single-parameter instance could be constrained to act like in the original multi-parameter instance in the sense that he will only choose those actions  associated with matrix $F^{\theta}$. Unfortunately, this is actually not always the case. {A more careful analysis is needed to handle these situations.} 
 In the following, we describe our reduction in detail.

\begin{figure}[t]
\centering
\includegraphics[scale=0.45]{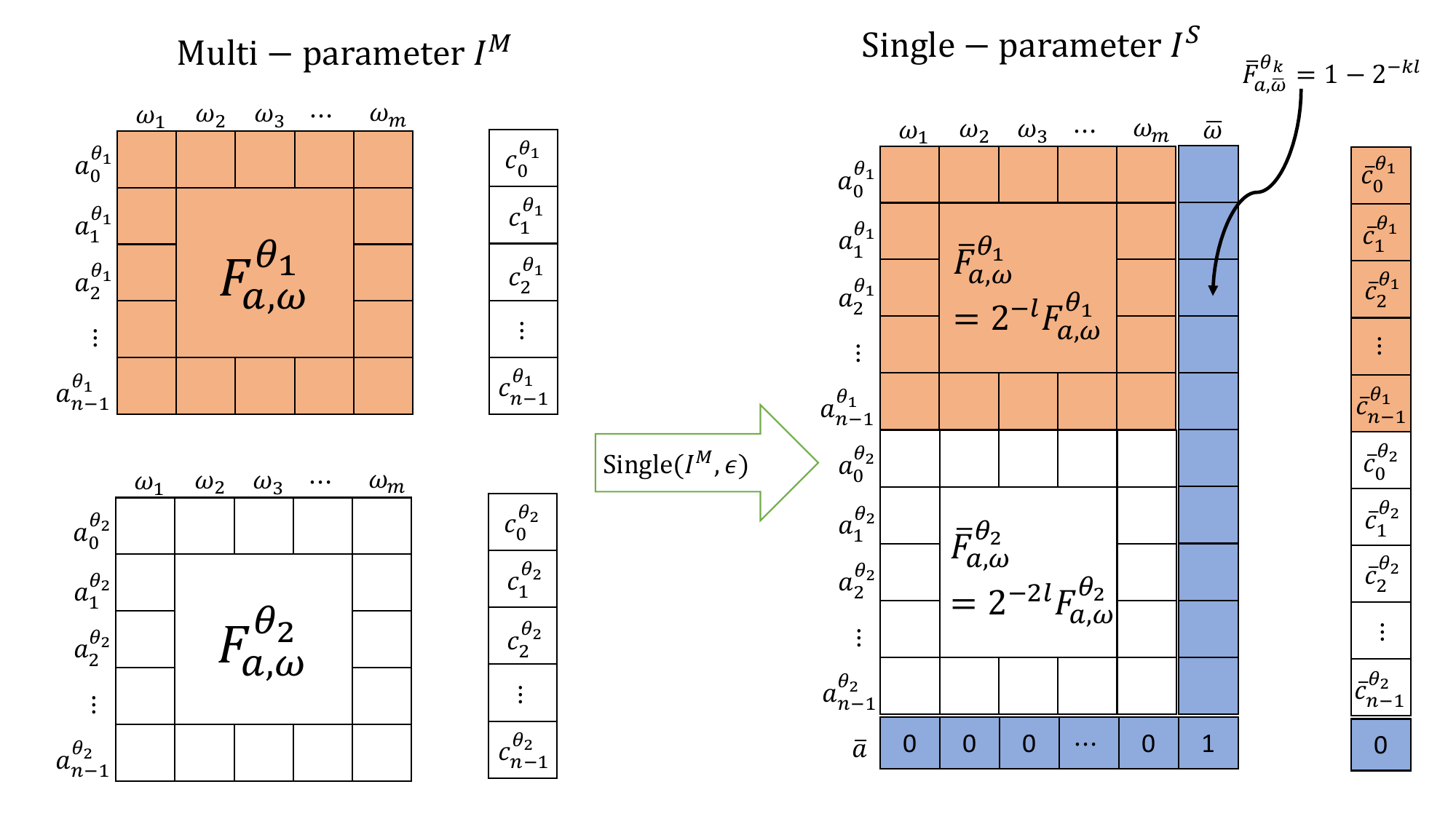}
\caption{The example of the construction procedure $I^{S}=\texttt{Single}(I^M,\epsilon)$, with only two types $\theta_1$ and $\theta_2$. We add one dummy outcome $\bar{\omega}$ and one action $\bar{a}$ with cost $0$. Best viewed in color.}
\label{proceduremultisingleprocess}
\end{figure}

\vspace{2mm}
\noindent 
\textbf{Construction of the Single-parameter Instance.} We consider a procedure $\texttt{Single}(I^M,\epsilon)$ that given a multi-parameter instance $I^M=(\Theta,\Omega,\{F^{\theta}\}_{\theta},\{c^\theta\}_{\theta},\mu)$ returns as output a single-parameter instance $I^{S}=(\bar \Theta,\bar \Omega,\bar F,\bar c,\bar{\mu})$.
A visualization of the construction of $I^{S}=\texttt{Single}(I^M,\epsilon)$ for an example  with two types is depicted in Figure \ref{proceduremultisingleprocess}. 
The single-parameter instance introduces a dummy outcome $\bar{\omega}$ with reward $r_{\bar \omega} = 0$. We overload the notations and, for each type $\theta_k$ of the multi-parameter instance,   define a single-parameter type $\theta_k = 2^{kl}$ with probability $ \bar{\mu}(\theta_k) =\mu(\theta_k) 2^{kl} H(\alpha, l)$,   
where $H(\alpha, l) =   \frac{\alpha}{\sum_{k\in[K]} \mu(\theta_k)2^{kl}}$, $K$ is the number of types in the multi-parameter instance, and $\alpha, l$ are parameters to be specified below. 
   %
   Moreover, the single-parameter instance includes an additional type {$\bar \theta=\frac{2^{2Kl+1}}{\epsilon}$} with $ \bar{\mu}(\bar{\theta}) =1-\alpha$. 
   Notice that $\bar{\mu}$ is  a probability distribution by construction. {We define $\bar{\Omega} \triangleq \Omega \cup \{\bar{\omega}\}$ and $\bar{\Theta} \triangleq \Theta \cup \{\bar{\theta}\}$ to represent sets of outcomes and types for the single-parameter instance.}
   The common  probability matrix $\bar{F}$ for single-parameter setup is defined in \eqref{barFandFdefintion}. Moreover, we also define an auxiliary  matrix $F$ related to the multi-parameter instance in \eqref{barFandFdefintion} below:
   \begin{equation}\label{barFandFdefintion}
   \bar{F} = \left[
   \begin{array}{l}
        \bar{F}^{\theta_1}  \\
        \bar{F}^{\theta_2} \\
        \cdots \\
        \bar{F}^{\theta_K}\\
        \bar{F}^{\bar \theta}
   \end{array}
   \right], \quad {F} = \left[
   \begin{array}{l}
        {F}^{\theta_1}  \\
        {F}^{\theta_2} \\
        \cdots \\
        {F}^{\theta_K}\\
        0
   \end{array}
   \right],
    \end{equation}
   where $\bar{F}^{\theta_k}_{a, \omega} = 2^{-kl} F^{\theta_k}_{a, \omega}$ for any $ \omega \neq \bar{\omega}$ and $\bar{F}^{\theta_k}_{a, \bar{\omega}} = 1- \sum_\omega 2^{-kl} F^{\theta_k}_{a, \omega}$  for any type $\theta_k$ and action $a$ in the multi-parameter instance $I^M$.  
    $\bar{F}^{\bar{\theta}}$
    includes only  one action $\bar{a}$, where we let $\bar{F}^{\bar \theta}_{\bar a, \bar \omega}=1$, and $\bar{F}^{\bar \theta}_{\bar a, \omega}=0$ for each $\omega \neq \bar \omega$. Therefore, $\bar{F} \in  R^{(nK+1) \times (m+1)}$ with $nK+1$ actions and $m+1$ outcomes. The common cost vector $\bar{c} = [\bar{c}^{\theta_1}, \bar{c}^{\theta_2}, \cdots, \bar{c}^{\theta_K},\bar c^{\bar \theta}] \in R^{nK+1}$ is defined as follows: $\bar{c}^{\theta_k} = 2^{-2kl}(c^{\theta_k}+\eps) \in R^{n}$ and $\bar{c}^{\bar \theta}=0$ (i.e., $\bar{c}_{\bar{a}} = 0$) with some small $\eps > 0$.  The  parameters $\alpha, l, \epsilon$ of the construction are  specified to satisfy the following conditions: $l > \log(\frac{4}{\mu_{\min}\eps})$ and $\alpha \le \frac{1}{2^{(K+1)l} + 1}$, where $\mu_{\min} = \min_{k\in [K]} \mu(\theta_k)$. {Existence of these parameters is evident and any of them suffice for our reduction.} 
    

   \begin{figure}[t]
\centering
\includegraphics[scale=0.3]{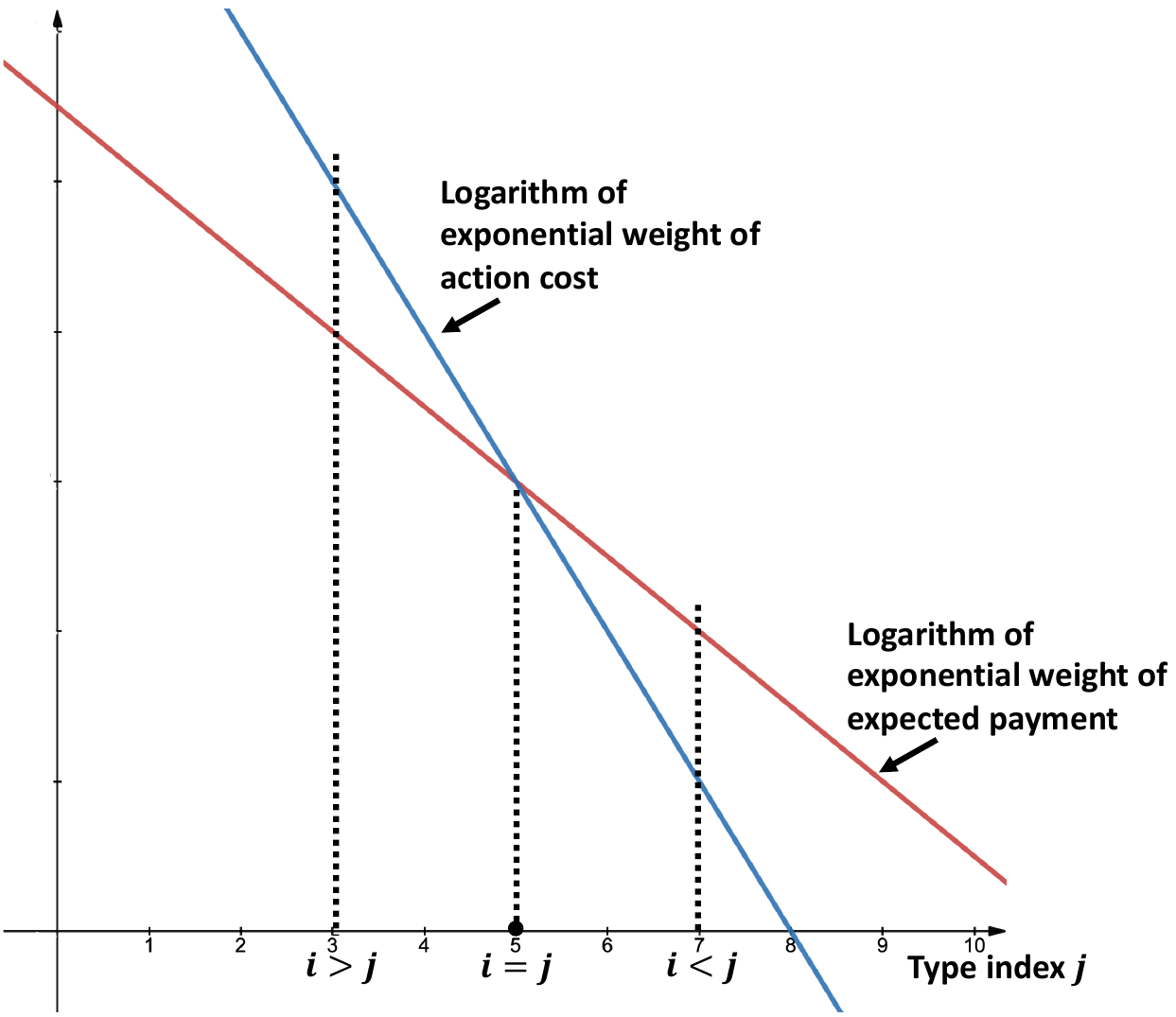}
\caption{Comparisons of the exponential weights between action cost and expected payment.}.
\label{proofhighlevellemmamultitosignelexponweight}
\end{figure}

\begin{remark} 
A useful observation is that if   $\mu(\theta)$ of the given $I^M$ is a uniform distribution, then the constructed type distribution $ \bar{\mu}(\theta)$ for single-parameter instance $I^S$  can be shown to be a discrete regular distribution~\cite{myerson1981optimal,alon2021contracts,yang2022selling}. To see this, we only need to verify that the virtual cost function $\phi(\theta_k) = \theta_k + \frac{(\theta_k-\theta_{k-1}) \bar{M}(\theta_{k-1})}{ \bar{\mu}(\theta_k)}$ (see, e.g., \cite{alon2021contracts}) is monotone,  where $\bar{M}(\theta_k) = \sum_{i=1}^k  \bar{\mu}(\theta_i)$.\footnote{{We have that $\phi(\theta_k) = 2^{kl} + \frac{(2^{kl}-2^{(k-1)l})\sum_{i=1}^{k-1} \mu(\theta_i) 2^{il} H(\alpha, l)}{\mu(\theta_k) 2^{kl} H(\alpha, l)}=2^{kl} + {(1-2^{-l})\sum_{i=1}^{k-1} 2^{il}}$. Hence,  $\phi(\theta_k)$ is increasing for all $\theta_k$, $k \in [K]$. Moreover, we also have $\phi(\theta_K) \le \phi(\bar{\theta})$ because $\phi(\theta_K) = 2^{Kl} + (1-2^{-l})\sum_{i=1}^{K-1} 2^{il} < 2^{Kl} + 2^{Kl} < \bar{\theta} = \frac{2^{2Kl+1}}{\epsilon} $.}} This observation is used for concluding the NP-hardness  in Corollary \ref{cor1} and \ref{cor2} for regular type distributions.      
\end{remark}


\noindent 
\textbf{Renaming   Constructed Actions.} In the constructed single-parameter instance,  for each   $\theta \in \Theta$, the $n$ actions  in $\bar{F}^{\theta}$ comes from type $\theta$'s $n$ actions in  $F^{\theta}$. To distinguish these actions from different agent types, we use $A^{\theta} = \{a^{\theta}_0, a^{\theta}_1, \dots, a^{\theta}_{n-1}\}$ to denote the set of $n$ actions in matrix $\bar{F}^{\theta}$ and $F^{\theta}$, and we use $a^{\theta} = [a^{\theta}_0, a^{\theta}_1, \dots, a^{\theta}_{n-1}]$ to vectorize the actions in $\bar{F}^{\theta}$ and $F^{\theta}$ where $a^{\theta}_i$ is the $i^{\textnormal{th}}$ action in $a^{\theta}$. 
   There is a key difference between the single and multi-parameter instances: in the multi-parameter instance, there are only $n$ different actions available to the agent while, in the constructed single-parameter instance, there are $K\cdot n+1$ different actions available to the agent. 
   In other words, in the original multi-parameter instances, $A^{\theta}$ are the only actions available to type $\theta$ agent, but all the actions in $\bar{A}\triangleq A^{\theta_1} \cup A^{\theta_2} \dots A^{\theta_K} \cup \{\bar{a}\}$ can be played by agent's type $\theta$ in the single-parameter instance.
   Hence, using our distinguished notations of best-responses for multi-parameter and single-parameter settings, we have  $a^m(\theta, p) \in A^{\theta}$ for any $p \in \mathbb{R}^{m}_+$, while $a^s(\theta, p) \in \bar{A}$ for any $p \in \mathbb{R}^{m+1}_+$. 
   Finally, when we use the notation $\bar{F}^{\theta}_{i, j}$ and ${F}^{\theta}_{i, j}$, the index $i$ indicates action  $a^{\theta}_i$. Moreover, to ease the exposition,  we use   $a^m(\theta)$ to denote the best-response action of type $\theta$ agent under menu $P$ for the multi-parameter instance, i.e.,  
$a^m(\theta) =a^m(\theta, p^{\theta})=  \arg\max_a \sum_{\omega\in[m]} F^\theta_{a, \omega} p_\omega^{\theta} -c^\theta_{a}$. Similarly, we define $a^s(\theta)$ under menu $\bar{P}$ for the single-parameter instance.

\subsection{Proof Overview and Main Arguments}\label{sec:reduction-main}

    Throughout the rest of Section \ref{section:equvivalentce}, we use $P$ (resp. $\bar{P}$) to denote the menu for the multi-parameter instance $I^M$ (resp. single-parameter instance $I^S$). 
The proof of Theorem \ref{equivalencesinglemulti} mainly relies on the following two technical lemmas.

\begin{lemma}\label{multitosingleconstrucstion}
    Given  a multi-parameter instance $I^M$  and  an IC menu of contracts $P =( p^\theta)_{\theta \in \Theta}$, 
    there exists an algorithm {that runs in time  polynomial in the instance size $|I^M|$ and $\log(1/\epsilon)$, and} outputs an IC menu $\bar P=(\bar p^\theta)_{\theta \in \bar{\Theta}}$  
    for the single-parameter instance $I^S=\texttt{Single}(I^M,\epsilon)$ that guarantees:
    \begin{equation}\label{multitosingleinequality}
    \sum_{\theta \in \bar{\Theta} }  \bar{\mu}(\theta) \sum_{\omega\in \bar{\Omega}} \bar{F}_{ a^s(\theta), \omega} (r_\omega - \bar{p}^{\theta}_\omega)\ge  H(\alpha, l)  \Big(\sum_{\theta \in \Theta}  \mu(\theta) \sum_{\omega\in \Omega} {F}^{\theta}_{a^m(\theta), \omega} (r_\omega - {p}^{\theta}_\omega)  -2\eps \Big).
    \end{equation}
    Moreover, if $P$ is a single contract, then $\bar P$ is a single contract.
\end{lemma}

\begin{lemma}\label{lemma:single-multi}
    Given a multi-parameter instance $I^M$, the corresponding single-parameter instance $I^S=\texttt{Single}(I^M,\epsilon)$, and an IC menu $\bar{P}=(\bar p^\theta)_{\theta \in \bar{\Theta}}$  
    for $I^S$, we can build in time polynomial in the instance size $|I^M|$ and $\log(1/\epsilon)$ an IC menu $P=( p^\theta)_{\theta \in \Theta}$ 
    that guarantees: 
    \begin{equation}\label{singletomultiinequalitylemma}
    \sum_{\theta\in \bar{\Theta}}  \bar{\mu}(\theta) \sum_{\omega\in \bar{\Omega}} \bar{F}_{ {a}^s(\theta), \omega} (r_\omega - \bar{p}^{\theta}_\omega)\le  H(\alpha, l) \Big(\sum_{\theta \in \Theta}  \mu(\theta) \sum_{\omega\in \Omega} {F}^{\theta}_{a^m(\theta), \omega} (r_w - { p}^{\theta}_\omega) + {\nu} \Big), 
    \end{equation}
    where {$\nu = \frac{13\eps}{4} +4\sqrt{\eps}$}. Moreover, if $\bar{P}$ is a single contract, then $P$ is a single contract.
\end{lemma}

\begin{remark}
It turns out that, when armed with Lemma \ref{multitosingleconstrucstion} and \ref{lemma:single-multi}, we are able to easily recover the $\Omega(n)$ lower bound in Theorem \ref{upperlowertheoremsingleparameter}. This can be achieved by   converting the multi-parameter instance $I^M$ constructed by \citet{guruganesh2034contracts} for achieving  the   $\Omega(n)$ lower bound  for $OPT_{\text{menu}}/OPT_{\text{single}}$ to a single-parameter instance $I^S$ using our reduction, and then employing Lemma \ref{multitosingleconstrucstion} and \ref{lemma:single-multi} to show that this $I^S$ satisfies $\Omega(n) = OPT_{\text{menu}}/OPT_{\text{single}}$. 
We formulate this idea in Proposition \ref{recvoingprooftheorem31} in Appendix \ref{app_recovering_theorem31}. Notably, however, this does not mean Theorem \ref{upperlowertheoremsingleparameter} has a simpler proof because establishing each of the   Lemma \ref{multitosingleconstrucstion} and \ref{lemma:single-multi} is already more complicated than directly proving Theorem \ref{upperlowertheoremsingleparameter} itself. This is somewhat  expected since these two lemmas are meant to tackle more general problems. 
\end{remark}

In the following, we provide an overview of the proofs of the previous two lemmas.
Lemma~\ref{multitosingleconstrucstion} shows that for any contract $P$ of the instance $I^M$, we can construct a contract $\bar{P}$ on the instance $I^S$, such that the principal's utility is at least a $H(\alpha, l)$ fraction of the principal's utility in the original contract (up to an additive loss that depends on $\epsilon$).
The key challenge in the construction is to eliminate the possibility that an agent of type $\theta_i$ chooses an action $a \in A^{\theta_j}$ with $i \neq j$. This could happen since in the single-parameter instance $I^S$, there are more actions available than in instance $I^M$ for each agent's type. This problem is alleviated by our novel construction of probability matrix $\bar{F}$ and  costs  $\bar{c}$. Specifically, we utilize exponential weights, i.e., $2^{-il}$, to distinguish each probability submatrix $\bar{F}^{\theta_i}$. Note that the weight for the cost of $\theta_i$-agent taking an action from $A^{\theta_j}$ is $2^{il-2jl}$. As shown in Figure \ref{proofhighlevellemmamultitosignelexponweight}, if $j<i$, the exponential weight will dramatically increase the cost and give a negative utility to the agent, while if $j>i$, the expected payment from choosing an action in $A^{\theta_j}$ will exponentially decrease, thus giving a small utility. Therefore, by such a novel construction, it is intuitively better for the $\theta_i$-agent to take action from $A^{\theta_i}$ even with more available actions in $I^S$. {We remark that this is not true in general, and dealing with this enlarged action set requires a more careful analysis. }

\begin{figure}[ht]
\centering
\includegraphics[scale=0.5]{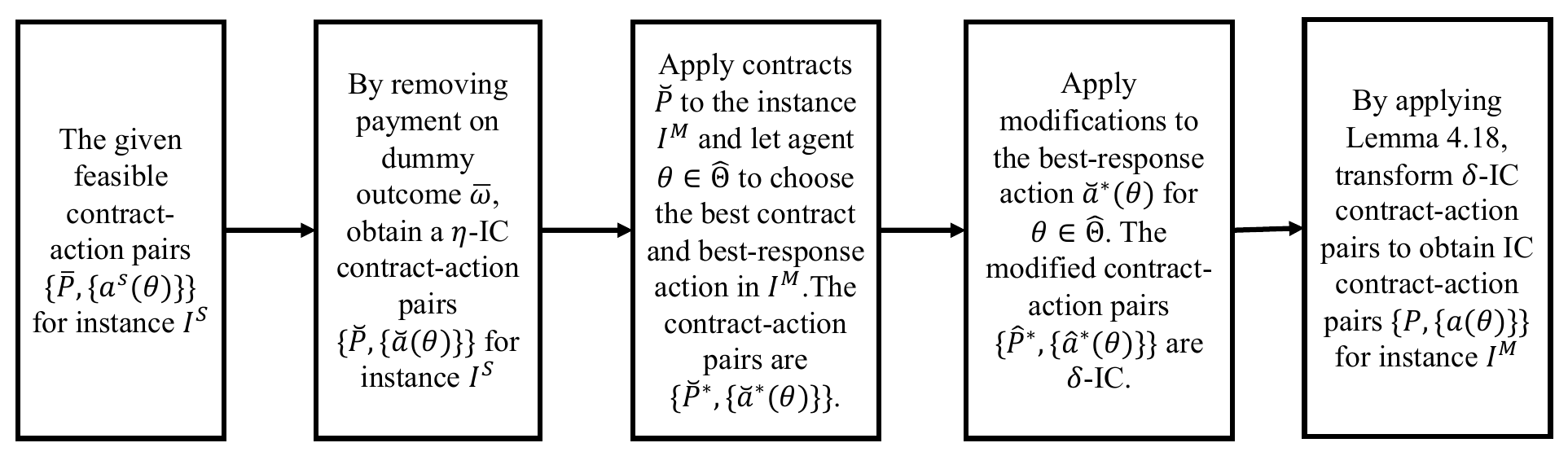}
\caption{The flow of proof of Lemma \ref{lemma:single-multi}. From a given feasible contract $\bar{P}$ for instance $I^S$, we construct a contract $P$ for $I^M$ through a sequence of intermediary contracts $\breve{P}$ and $\breve{P}^*$.} 
\label{prooflemmaequalitvensingletomultiflow}
\end{figure}

  In Lemma~\ref{lemma:single-multi}, we show the reverse direction which is also the more involved direction. In particular, we prove that, given any contract $\bar{P}$ of the single-parameter instance $I^S$, we can construct a contract $P$ for the multi-parameter instance $I^M$ such that the principal's utility in $I^S$ is at most an  $H(\alpha, l)$ fraction of the principal's utility in $I^M$ (up to some small additive loss $\nu$).
The key challenge in deriving this result is that one cannot directly apply the contract $\bar{P}$ to the instance $I^M$, which may cause a severe utility loss to the principal. This is mainly due to the problem that the best-response actions of type-$\theta$ agent in $I^S$ and $I^M$ are mismatched since the agent in $I^S$ has more available actions 
(i.e., agent can choose any action $a\in \bar{A}$),
while in $I^M$ the agent has a constrained action set $A^\theta$. Fortunately, we show that the utility loss is bounded by a small $\nu$. The proof of Lemma~\ref{lemma:single-multi} is intricate and follows by a sequence of constructions of contracts. The flow of the proof is represented in Figure \ref{prooflemmaequalitvensingletomultiflow}. Given a contract $\bar{P}$ that has a positive principal's utility, three key observations constitute the foundation of our proof: 1) the payment on dummy outcome is small (Lemma \ref{pbarthetabarwissmall}), 2) the best-response action of type $\theta_i$ action is either from $A^{\theta_j}$ with $j\ge i$ or $\bar{a}$ (Lemma \ref{lemmaagentincentiveaction2}), and 3) the expected payment on non-dummy outcome is bounded by some constant (Lemma \ref{lm:paymentBoundedjgretateei}). These observations allow us to build $\eta$-IC contracts and ensure that  the principal cannot gain large utility from agent $\theta_i$ choosing action  $\bar{a}$ or an action from $A^{\theta_j}$ with $j > i$ (Lemma \ref{lm:smallLoss}). We denote these agent types as $\hat{\Theta}$. This further implies that the principal will not lose much utility when letting all agent $\theta$ choose an action from $A^{\theta}$, i.e., constructing contract $\breve{P}^*$ for instance $I^M$. Note that some types in $\hat{\Theta}$ may cause large negative utility to the principal under contract $\breve{P}^*$. For those types, we reassign the opt-out action. Finally, the {obtained contract $\hat{P}^*$} is proved to be $\delta$-IC, and can be transformed to IC by losing some small utility.

{Lemma \ref{multitosingleconstrucstion} and \ref{lemma:single-multi} together imply the following (almost) approximation-preserving reduction. }

\begin{proposition} \label{thm:apx} 
 Let $\mathcal{A}^S$  be  a polynomial-time algorithm that finds a $\beta$-approximate menu of contracts in any  single-parameter BCD instance and   $\epsilon>0$ be any constant. From $\mathcal{A}^S$  we can construct
     an algorithm  $\mathcal{A}^M$ for multi-parameter BCD  that runs in time polynomial in the instance size and $\log(1/\epsilon)$, and guarantees the principal's utility to be at least 
    \[ \beta OPT^M- \frac{21}{4}\eps - 4\sqrt{\eps}\]
    where $OPT^M$ is the optimal utility for the multi-parameter instance. Moreover, the same results hold for the single contract case. 
\end{proposition}

\begin{proof}
Given an instance $I^M$, we consider the instance $I^S =\texttt{Single}(I^M,\epsilon)$ constructed by our reduction above.
Let $\bar P^*$ (together with  best-response action set $\{a^{*s}(\theta)\}_{\theta\in \bar{\Theta}}$) be the optimal contract for $I^S$ and $P^*$ (together with best-response action set $\{a^{*m}(\theta)\}_{\theta\in \Theta}$) the optimal contract for $I^M$. We can use the polynomial-time algorithm on $I^S$ and find a menu $\bar P$ (together with  best-response action set $\{a^{s}(\theta)\}_{\theta\in \bar{\Theta}}$) such that 
\begin{align*}
\sum_{\theta_i \in \bar{\Theta}}  \bar{\mu}(\theta_i) \sum_{\omega\in \bar{\Omega}} \bar{F}_{a^s(\theta_i), \omega} (r_\omega - \bar{p}^{\theta_i}_\omega)&\ge \beta [ \sum_{\theta_i \in \bar{\Theta}}  \bar{\mu}(\theta_i) \sum_{\omega\in \bar{\Omega}} \bar{F}_{ {a}^{*s}(\theta_i), \omega} (r_\omega - \bar{p}^{*\theta_i}_\omega)]  \\
&\ge \beta [H(\alpha, l) \Big(\sum_{\theta_i \in \Theta}  \mu(\theta_i) \sum_{\omega\in \Omega} {F}^{\theta_i}_{a^{*m}(\theta_i), \omega} (r_\omega - {p}^{*\theta_i}_\omega)  -2\eps \Big) ],  
\end{align*}
where the second inequality follows by the  Lemma~\ref{multitosingleconstrucstion} and the optimality of contract $\bar{p}^*$

Finally, thanks to Lemma~\ref{lemma:single-multi}, we can build a menu $P$ for the multi-parameter instance from $\bar{P}$ such that 
\begin{align*}
H(\alpha, l)  \Big(\sum_{\theta_i \in \Theta}  \mu(\theta_i) &\sum_{\omega\in \Omega } {F}^{\theta_i}_{a^m(\theta_i), \omega} (r_\omega - { p}^{\theta_i}_\omega) + {\nu} \Big) \\
&\ge \sum_{\theta_i \in \bar{\Theta}}  \bar{\mu}(\theta_i) \sum_{\omega\in \bar{\Omega}} \bar{F}_{ {a}^s(\theta_i), \omega} (r_\omega - \bar{p}^{\theta_i}_\omega)\\
& \ge \beta [H(\alpha, l)  \Big(\sum_{\theta_i \in \Theta}  \mu(\theta_i) \sum_{\omega\in \Omega } {F}^{\theta_i}_{a^{*m}(\theta_i), \omega} (r_\omega - {p^*}^{\theta_i}_\omega)  -2\eps \Big) ].
\end{align*}

Hence,
\[ \sum_{\theta_i \in \Theta}  \mu(\theta_i) \sum_{\omega\in \Omega } {F}^{\theta_i}_{a^m(\theta_i), \omega} (r_\omega - { p}^{\theta_i}_\omega) \ge \beta OPT^M-  
 (2\eps+\nu) = \beta OPT^M-  
\frac{21\eps}{4} - 4\sqrt{\eps}.\]

Finally, notice that since the constructed single-parameter instance   has size polynomial in the size of the multi-dimensional instance and $\log(1/\epsilon)$, the running time of the algorithm is polynomial in these two parameters.
\end{proof}

\vspace{2mm}
\noindent 
\textbf{The Last Mile: From Additive Loss to Multiplicative Approximation. } We are now  ready to complete the proof of Theorem \ref{equivalencesinglemulti} by converting the  additive loss $\frac{21}{4}\eps + 4\sqrt{\eps}$ in Proposition \ref{thm:apx} to an multiplicative loss. To do so, we only need to prove for the case where the optimal objective value is  above $0$ (if it is $0$, then the trivial contract of $0$ payment is optimal already). We first prove it for the menus of contracts and then  extend the proof to single contracts.

Let $A= \langle n\rangle^K$ denote the set of all possible action profiles, one per type. 
 The optimal direct menu that incentivizes an action profile $a\in A$ is the solution to Program~(\ref{multiparameteroptprob}) in which the function $a^m(\cdot)$ is set to be $a^m(\theta)=a(\theta)$. 
    Hence, Program~\eqref{multiparameteroptprob} now becomes  a linear program. 
{By enumerating all such $a \in A$, we know that  $OPT^M$, i.e.,  the optimal objective value of a multi-parameter instance, must be the optimal objective  value of one of these $|A|$ LPs.} 
Therefore,   the bit complexity of $OPT^M$ is upper bounded by a polynomial function of the instance size $|I^M|$.
It implies that there exists a known polynomial function $\tau:\mathbb{N}\rightarrow \mathbb{Q}$ such that either $OPT^M=0$ or $OPT^M\ge 2^{-\tau(|I^M|)}$ (see, e.g., \cite{bertsimas1997introduction}). The claimed result holds trivially when $OPT^M=0$. 
So next we only consider $OPT^M\ge 2^{-\tau(|I^M|)}$.  We apply our reduction from multi to single parameter instance with parameter $\epsilon\le \left(\frac{2^{-\tau(|I^M|)}}{10}\xi\right)^2$.
 It is easy to see that $\log(1/\epsilon)$ is polynomial in the instance size and $\log(1/\xi)$. Then, by Proposition~\ref{thm:apx} we can recover a  menu of contracts $P$ for the multi-parameter instance such that 
	\[  \beta OPT^M- \frac{21}{4}\eps - 4\sqrt{\eps}. \]
This menu of contracts provides the desired approximation since $OPT^M\ge2^{-\tau(|I^M|)}$ implies 
	\begin{align*}
		 \frac{\beta OPT^M- \frac{21}{4}\eps - 4\sqrt{\eps}}{OPT^M}   & \ge  \frac{\beta OPT^M-  10\sqrt{\eps}}{OPT^M}   \\
		 & \ge  \frac{\beta OPT^M-  2^{-\tau(|I^M|)}  \xi}{OPT^M}\\
		 & \ge (\beta-\xi) \frac{OPT^M}{OPT^M}\\
		 &= (\beta-\xi).
	\end{align*}
	This concludes the proof for menus of contracts. The same argument can be used to prove the result for single contracts noticing that the optimal single contract is the solution to Program~(\ref{multiparameteroptprob}) with the additional linear constraint $p^{\theta_1} = \dots = p^{\theta_K}$.


\begin{proofof}{Corollary \ref{exactlyreducetosignlemulti}}\textbf{: Resolving the numerical issue under exact optimality.}   We conclude this section by proving Corollary \ref{exactlyreducetosignlemulti}. In particular, we show that  the approximation $\xi$ can be removed if the single-parameter solution is exactly optimal, i.e., $\beta = 1$. 

Let $A= \langle n\rangle^K$ be the set of tuples of an action per type.
Similar to the last part of the proof of Theorem~\ref{equivalencesinglemulti}, 
the optimal direct menu that incentivizes an action profile $a\in A$ is the solution to an LP.  In the following, we denote with $LP(a)$ the value of the LP induced by action profile $a$.
Let $OPT^M$ be the optimal utility for the multi-parameter instance.
Then, since $OPT^M=\max_{a \in A} LP(a)$ and objective values of LPs have polynomial bit complexity,  there exists a known polynomial function $\tau:\mathbb{N}\rightarrow \mathbb{Q}$ such that for any sub-optimal $a \in A$ we have $\textnormal{LP}(a)\le OPT^M -2^{-\tau(|I^M|)}$ (see, e.g., \cite{bertsimas1997introduction}). By applying Theorem~\ref{equivalencesinglemulti} with $\xi = 2^{-\tau(|I^M|)-1}$,  we can find a menu $P$ for the multi-parameter instance that guarantees principal's utility
\[
(1-\xi)OPT^M \ge OPT^M - \xi OPT^M \ge OPT^M - 2^{-\tau(|I^M|)}/2.
\]
    Let $a^\star=(a^\star(\theta))_{\theta \in \Theta}$ be the tuple of actions played by the agent types under menu $P$. {Note that $a^\star$ can be found efficiently from $P$ by choosing every type's best response (and break tie in favor of principal if needed).} 
    Since all the sub-optimal $a \in A$ are at least $2^{-\tau(|I|)}$ sub-optimal, $a^\star$ must be an optimal action profile.
    Hence, solving Program~(\ref{multiparameteroptprob}) by setting  $a^m(\theta)=a^\star(\theta)$,
    we must recover a menu $P^\star$ with value $\textnormal{LP}(a^\star)=OPT^M$. Notice that since this   is a linear program, we can solve it in polynomial-time.

        The same argument can be used to prove the result for single contracts noticing that the optimal single contract is the solution to Program~(\ref{multiparameteroptprob}) with the additional linear constraint $p^{\theta_1} = \dots = p^{\theta_K}$.
\end{proofof}

\subsection{Proof of Lemma \ref{multitosingleconstrucstion} }\label{sec:reduction-lemA}

We now prove the first main lemma, stating that from a menu of multi-parameter instances, one can construct a menu for single-parameter instances that achieves a good guarantee on the principal's utility.


Depending on the expected payment  in the multi-parameter instance, we divide the analysis into two cases: 
\begin{itemize}
    \item (i) there exist two types $\theta,\theta' \in \Theta$ and an action $a\in A^\theta$ such that $\langle F^{\theta}_{a}, p^{\theta'}\rangle > \frac{2}{\mu_{\min}}$; 
    \item (ii) for any two types $\theta,\theta' \in \Theta$ and action $a\in A^\theta$, it holds  $\langle F^{\theta}_{a}, p^{\theta'}\rangle \le \frac{2}{\mu_{\min}}$.
\end{itemize}
In this section, we use $U^m(a; p, \theta)$ (resp. $U^s(a; p, \theta)$) to denote the utility of type $\theta$ agent taking action $a$ under contract $p$ in the multi-parameter instance (resp. the single-parameter instance).

\paragraph{Case (i).} If (i) holds, then we show that the expected principal's utility is negative in the multi-parameter instance, hence the menu $\bar{P}$ with $0$ payment everywhere satisfies Equation~(\ref{multitosingleinequality}) in the lemma.
Consider two types $\theta,\theta' \in \Theta$, and an action $a\in A^\theta$, such that $\langle F^{\theta}_{a}, p^{\theta'}\rangle > \frac{2}{\mu_{\min}}$.
By the IC constraint, it must be the case that
\[ U^m(a; p^{\theta'}, \theta) = \langle F^{\theta}_{a}, p^{\theta'} \rangle -c^{\theta}_{a}\le \langle F^{\theta}_{a^m(\theta)}, p^{\theta}\rangle -c^{\theta}_{a^m(\theta)} = U^m(a^m(\theta); p^{\theta}, \theta),  \]
implying the following lower bound on the expected payment:
\[ \langle F^{\theta}_{a^m(\theta)}, p^{\theta} \rangle \ge \langle F^{\theta}_{a}, p^{\theta'} \rangle-1 \ge 2/\mu_{\min} -1\ge 1/\mu_{\min}. \]
Hence, the expected principal's utility is at most
$1- \mu(\theta) \frac{1}{\mu_{\min}} \le 0$.
This concludes the first part of the proof. Notice that the same proof applies to the case of single contracts.

\paragraph{Case (ii).} In the rest of the proof, we assume (ii) holds hence $\langle F^{\theta}_{a}, p^{\theta'}\rangle \le \frac{2}{\mu_{\min}}$ for any two $\theta, \theta'$ and any $a \in A^\theta$. We construct the menu $\bar P$ for single-parameter instance $I^S$ as follows: for any $\theta \in \Theta$,  we let $\bar p^\theta_{\omega}= p^\theta_{\omega} +2 \epsilon$ for $\omega \neq \bar{\omega}$ and $\bar p^\theta_{\bar{\omega}} =0$, while we set $\bar p^{\bar \theta}=0$ for type $\bar{\theta}$.


Our first step is to show that the two menus $P$ and $\bar{P}$ incentivize the same best-response actions, i.e., $a^s(\theta) = {a}^m(\theta)$ for every $\theta \in \Theta$. In other words,  $a^s(\theta)$ and $a^m(\theta)$ are the same action (and belong to the set $A^{\theta}$).  Note that in the single-parameter instance,  the agent $\theta$ has two types of deviation:  misreporting the type (i.e.,  choosing a contract $\bar{p}\neq \bar{p}^{\theta}$) and playing an action different from ${a}^m(\theta)$. Next, we verify that the agent will not deviate.

   For any type $\theta_k\in \Theta$, we  show that under contract $\bar{p}^{\theta_k}$, the agent will not deviate to action $a^{\theta_k}_{i} \neq a^m(\theta_k)$. 
   To show that, we only need to verify that the action deviation does not lead to greater utility to the agent, i.e., the following holds:
   \[
    U^s(a^m(\theta_k); \bar{p}^{\theta_k}, \theta_k) = \sum_{\omega \in \bar{\Omega}} \bar{F}^{\theta_k}_{ a^m(\theta_k), \omega} \bar{p}^{\theta_k}_\omega - \theta_k \bar{c}^{\theta_k}_{a^m(\theta_k)} \ge  \sum_{\omega\in \bar{\Omega}} \bar{F}^{\theta_k}_{i, \omega} \bar{p}^{\theta_k}_\omega -\theta_k\bar{c}^{\theta_k}_{i} = U^s(a^{\theta_k}_i; \bar{p}^{\theta_k}, \theta_k).
   \]
   By some substitutions, we have 
   \begin{align*}
       U^s(a^m(\theta_k); \bar{p}^{\theta_k}, \theta_k) = &  \sum_{\omega \in \Omega } 2^{-kl}F^{\theta_k}_{ a^m(\theta_k), \omega} (p^{\theta_k}_\omega +2\eps) - 2^{kl}2^{-2kl} (c^{\theta_k}_{a^m(\theta_k)}+\eps)\\
       = & 2^{-kl} \Big(\sum_{\omega\in \Omega } F^{\theta_k}_{a^m(\theta_k), \omega} p^{\theta_k}_\omega + \eps -  c^{\theta_k}_{a^m(\theta_k)}\Big)  \\ 
       = &  2^{-kl} \Big( U^m(a^m(\theta_k); p^{\theta_k}, \theta_k) + \eps \Big),
    \end{align*}
    and
    \begin{align*}
       U^s(a^{\theta_k}_i; \bar{p}^{\theta_k}, \theta_k) = & \sum_{\omega\in \Omega } 2^{-kl}F^{\theta_k}_{ i, \omega} (p^{\theta_k}_\omega +2\eps) - 2^{kl}2^{-2kl} (c^{\theta_k}_{i}+\eps)\\
       = & 2^{-kl} \Big(\sum_{\omega\in \Omega } F^{\theta_k}_{i, \omega} p^{\theta_k}_\omega + \eps -  c^{\theta_k}_{i}\Big)\\
       = & 2^{-kl} \Big(U^m(a^{\theta_k}_i; {p}^{\theta_k}, \theta_k) + \eps\Big).
   \end{align*}
   Then, $U^s(a^m(\theta_k); \bar{p}^{\theta_k}, \theta_k) \ge U^s(a^{\theta_k}_i; \bar{p}^{\theta_k}, \theta_k)$ holds by the IC constraint in the original multi-parameter instance $I^M$ that guarantees $U^m(a^m(\theta_k); p^{\theta_k}, \theta_k) \ge U^m(a^{\theta_k}_i; {p}^{\theta_k}, \theta_k)$. Tie-breaking rule also prefers $a^m(\theta_k)$ over other actions in $I^S$ due to that $a^m(\theta_k)$ is the best response in $I^M$.
   
   Next, we show that under the same contract $\bar{p}^{\theta_k}$, the agent will not choose an action $a^{\theta_s}_i \in A^{\theta_s}$ with $\theta_s \neq \theta_k$.  
   First, we note that $\theta_k$ agent will not choose action $\bar{a}$ under contract $\bar{p}^{\theta_k}$ since $U^s(a^m(\theta_k); \bar{p}^{\theta_k}, \theta_k) \ge 2^{-kl}\eps$ by the IR constraint in instance $I^M$, while action $\bar{a}$ leads to zero utility. 

   For any action different from $\bar{a}$, the  utility deviating on the action is:
  \begin{align*}
       U^s(a^{\theta_s}_i; \bar{p}^{\theta_k}, \theta_k) = & \sum_{\omega \in \bar{\Omega}} \bar{F}^{\theta_s}_{i, \omega} \bar{p}^{\theta_k}_\omega - \theta_k \bar{c}^{\theta_s}_{i} \\
       = &\sum_{\omega\in \Omega } 2^{-sl}F^{\theta_s}_{ i, \omega} (p^{\theta_k}_\omega +2\eps) - 2^{kl}2^{-2sl} (c^{\theta_s}_{ i}+\eps) \\
       = & 2^{-sl} \Big(\sum_{\omega\in \Omega } F^{\theta_s}_{i, \omega} p^{\theta_k}_\omega +2\eps  - 2^{kl-sl} (c^{\theta_s}_{ i}+\eps)\Big).
   \end{align*} 
   Hence, we have that $U^s(a^m(\theta_k); \bar{p}^{\theta_k}, \theta_k) \ge U^s(a^{\theta_s}_i; \bar{p}^{\theta_k}, \theta_k)$ is equivalent to  \begin{equation}\label{lhshanabdclhsgerhs}
       \sum_{\omega\in \Omega } F^{\theta_s}_{i, \omega} p^{\theta_k}_\omega +2\eps  - 2^{kl-sl} (c^{\theta_s}_{ i}+\eps) \le 2^{sl-kl}  \Big(\sum_{\omega \in \Omega } F^{\theta_k}_{a^m(\theta_k), \omega} p^{\theta_k}_\omega + \eps -  c^{\theta_k}_{a^m(\theta_k)}\Big).
   \end{equation}
   By the IR constraint in instance $I^M$, the right-hand side of Equation~(\ref{lhshanabdclhsgerhs}) is strictly positive. If $k >s$, then, since $l>\log( \frac{4}{\mu_{\min}\epsilon})$ and $c_{i}^{\theta_s} \ge 0$, the left-hand side of (\ref{lhshanabdclhsgerhs}) is negative, i.e., 
   \[
   \sum_{\omega\in \Omega } F^{\theta_s}_{i, \omega} p^{\theta_k}_\omega +2\eps  - 2^{kl-sl} (c^{\theta_s}_{ i}+\eps) \le \frac{2}{\mu_{\min}} + 2\eps -2^{l}\eps <0 .
   \]
   If  $k<s$, the left-hand side of Equation~(\ref{lhshanabdclhsgerhs}) is at most $2/\mu_{\min}+2\epsilon$, while the right-hand side is at least $\frac{3}{\mu_{\min}}$ since $\sum_{\omega\in \Omega } F^{\theta_i}_{ a^m(\theta_i), \omega} p^{\theta_i}_\omega +\eps -  c^{\theta_i}_{ a^m(\theta_i)} \ge \eps$ by the IR constraint of $I^M$.

Next, we show that agent $\theta_k$ will not misreport his type. The agent will not misreport to $\bar{\theta}$, which leads to non-positive utility. Then, suppose $\theta_k$ agent misreports to type $\theta_s$. The agent will not choose action $\bar{a}$; otherwise, he will have $0$ utility. If the agent takes some action $a^{\theta_k}_i$, the utility is 
   \begin{align*}
       U^s(a^{\theta_k}_i; \bar{p}^{\theta_s}, \theta_k) =& \sum_{\omega \in \bar{\Omega}} \bar{F}^{\theta_k}_{i, \omega} \bar{p}^{\theta_s}_\omega -\theta_k \bar{c}^{\theta_k}_{i}\\
       =& \sum_{\omega\in \Omega } 2^{-kl}{F}^{\theta_k}_{i, \omega} ({p}^{\theta_s}_\omega +2\eps)-2^{kl}2^{-2kl}(c^{\theta_k}_{i}+\eps) \\
       =& 2^{-kl}\Big( \sum_{\omega\in \Omega } {F}^{\theta_k}_{i, \omega} {p}^{\theta_s}_\omega  +\eps -c^{\theta_k}_{i} \Big) \\
       =& 2^{-kl}\Big( U^m(a^{\theta_k}_i; {p}^{\theta_s}, \theta_k)  +\eps  \Big).
   \end{align*}
   By the IC constraint in instance $I^M$ we have that $U^m(a^{\theta_k}_i; {p}^{\theta_k}, \theta_k) \ge U^m(a^{\theta_k}_i; {p}^{\theta_s}, \theta_k)$. This implies $U^s(a^m(\theta_k); \bar{p}^{\theta_k}, \theta_k)  \ge U^s(a^{\theta_k}_i; \bar{p}^{\theta_s}, \theta_k)$. If the agent takes action $a^{\theta_t}_i$, the utility is 
    \begin{align*}
       U^s(a^{\theta_t}_i; \bar{p}^{\theta_s}, \theta_k) =& \sum_{\omega \in \bar{\Omega}} \bar{F}^{\theta_t}_{i, \omega} \bar{p}^{\theta_s}_\omega - \theta_k \bar{c}^{\theta_t}_{i} \\
       = &\sum_{\omega \in \Omega } 2^{-tl}{F}^{\theta_t}_{i, \omega} ({p}^{\theta_s}_\omega +2\eps)-2^{kl}2^{-2tl}(c^{\theta_t}_{i}+\eps) \\
       =& 2^{-tl}\Big( \sum_{\omega\in \Omega } {F}^{\theta_t}_{i, \omega} {p}^{\theta_s}_\omega  +2\eps - 2^{kl-tl}(c^{\theta_t}_{i}+\eps) \Big) .
   \end{align*}  
    Hence, we have that $U^s(a^m(\theta_k); \bar{p}^{\theta_k}, \theta_k) \ge U^s(a^{\theta_t}_i; \bar{p}^{\theta_s}, \theta_k)$ is equivalent to  
    \begin{equation*}
      \sum_{\omega\in \Omega } {F}^{\theta_t}_{i, \omega} {p}^{\theta_s}_\omega  +2\eps - 2^{kl-tl}(c^{\theta_t}_{i}+\eps) \le 2^{tl-kl}  \Big(\sum_{\omega\in \Omega } F^{\theta_k}_{a^m(\theta_k), \omega} p^{\theta_k}_\omega + \eps -  c^{\theta_k}_{a^m(\theta_k)}\Big).
   \end{equation*}
   Then, by arguments similar to the ones used to  prove Equation~(\ref{lhshanabdclhsgerhs}) and $\sum_{\omega \in \Omega } {F}^{\theta_t}_{i, \omega} {p}^{\theta_s}_\omega \le \frac{2}{\mu_{\min}}$, we have that $U^s(a^m(\theta_k); \bar{p}^{\theta_k}, \theta_k) \ge U^s(a^{\theta_t}_i; \bar{p}^{\theta_s}, \theta_k)$ holds. In conclusion, we proved that $a^s(\theta_k) = a^m(\theta_k)$. 

   Finally, we show that the best-response action for $\bar{\theta}$ is $a^s(\bar{\theta}) = \bar{a}$. This can be shown by noticing that if $\bar{\theta}$ agent reports any other type (e.g.,  $\theta_s$) and takes any other action (e.g., $a^{\theta_t}_i$), the utility is negative, i.e.,
   \begin{align} \label{asbartheta0}
       U^s(a^{\theta_t}_i; \bar{p}^{\theta_s}, \bar{\theta}) =& \sum_{\omega \in \bar{\Omega}} \bar{F}^{\theta_t}_{i, \omega} \bar{p}^{\theta_s}_\omega -\bar{\theta}\bar{c}^{\theta_t}_{i} \\ 
       =& \sum_{\omega \in \Omega } 2^{-tl}{F}^{\theta_t}_{i, \omega} ({p}^{\theta_s}_\omega +2\eps)-\frac{2^{2Kl+1}}{\eps}2^{-2tl}(c^{\theta_t}_{i}+\eps) \notag \\
       \le& 2^{-tl}\Big( \sum_{\omega\in \Omega } {F}^{\theta_t}_{i, \omega} {p}^{\theta_s}_\omega  +2\eps - \frac{2^{Kl+1}}{\eps} (c^{\theta_t}_{i}+\eps) \Big)  \notag \\
       \le& 2^{-tl}\Big( \frac{2}{\mu_{\min}}  +2\eps - 2^{Kl+1} \Big) \notag\\
       < & 2^{-tl}\Big( \frac{2}{\mu_{\min}}  +2\eps - 2 (\frac{4}{\mu_{\min}\epsilon})^K \Big)  <0 \notag.
   \end{align}  
   Hence, under the constructed contract $\bar{P}$, we have shown that ${a}^s(\theta) = a^m(\theta)$ for $\theta \in \Theta$ and $a^s(\bar{\theta})=\bar{a}$. Since $\bar{\theta}$ agent does not provide utility to the principal, the expected principal's  utility in instance $I^S$ is equal to
   \begin{align*} 
   \sum_{\theta_k\in \bar{\Theta}}  \bar{\mu}(\theta_k) \sum_{\omega\in \bar{\Omega} } \bar{F}^{\theta_k}_{a^s(\theta_k), \omega} (r_\omega - \bar{p}^{\theta_k}_\omega)& = \sum_{\theta_k \in \Theta}\mu(\theta_k) 2^{kl}  H(\alpha, l)2^{-kl} \sum_{\omega\in \Omega} {F}^{\theta_k}_{ a^m(\theta_k), \omega} (r_\omega - {p}^{\theta_k}_\omega -2\eps)\\
   & =H(\alpha, l)  \sum_{\theta_k}  \mu(\theta_k) \sum_{\omega\in \Omega} {F}^{\theta_k}_{a^m(\theta_k), \omega} (r_\omega - {p}^{\theta_k}_\omega -2\eps)\\
   &= H(\alpha, l) \Big(\sum_{\theta_k}  \mu(\theta_k) \sum_{\omega\in \Omega} {F}^{\theta_k}_{a^m(\theta_k), \omega} (r_\omega - {p}^{\theta_k}_\omega)  - 2\eps\Big).
   \end{align*}
   This concludes the proof for menu of contracts.
Note that if $P$ is a single contract, we construct $\bar{P}$ as $\bar p_{\omega}= p_{\omega} +2 \epsilon$ for any $\omega \neq \bar{\omega}$ and $\bar p_{\bar{\omega}} =0$. Hence, the same argument still holds by an analysis similar to the one made for menus of contracts.  Indeed, the argument for single contract does not need to consider the agent's misreporting behavior but the same proof still applies. This concludes the proof of the lemma.


\subsection{Proof of Lemma \ref{lemma:single-multi}  }\label{sec:reduction-lemB} 


Finally, we prove the second main lemma, which is the reverse direction of Lemma \ref{multitosingleconstrucstion}. That is, given a menu of contracts of a single-parameter instance, we are able to construct a ``good" menu for the original multi-parameter instance. The pipeline of proof is in Figure \ref{prooflemmaequalitvensingletomultiflow}.



We use $U^m(a; p, \theta)$ (resp. $U^s(a; p, \theta)$) to denote the utility of type $\theta$ agent taking action $a$ under contract $p$ in  multi-parameter instances (resp. single-parameter instances). Depending on the expected utility of the principal in the single-parameter instance, we divide the proof into two cases: 
\begin{itemize}
    \item (i) The contract $\bar{P}$ gives negative utility to the principal in the instance $I^S$;
    \item (ii) The contract $\bar{P}$ gives non-negative utility to the principal.
\end{itemize}  

The first case is straightforward. We let $P$ be a menu with zero payment everywhere, which obviously satisfies Equation~(\ref{singletomultiinequalitylemma}). This conclusion also holds  if $\bar{P}$ is a single contract.
In the rest of the proof, we focus on the second case and thus always assume the following condition holds.

\begin{condition}\label{conditionnonnegativeutility}
 The menu $\bar P$ provides a non-negative principal utility in $I^S$.
\end{condition}


We first observe that the payment on dummy outcome $\bar{\omega}$ is small. Intuitively, this implies that when constructing a menu $P$ for $I^M$, the utility loss due to payments on dummy outcome is small.

\begin{lemma}\label{pbarthetabarwissmall}
Assume Condition \ref{conditionnonnegativeutility} holds. The payments on dummy outcome is $\bar{p}^{\theta}_{\bar{\omega}} \le \bar{p}^{\bar{\theta}}_{\bar{\omega}} \le \frac{\alpha}{1-\alpha}$ for every $\theta \in \Theta$. Moreover, the best-response action for type $\bar{\theta}$ is ${a}^s(\bar{\theta})=\bar{a}$. 
\end{lemma}


\begin{proof}
%
%
We start the proof by showing that ${a}^s(\bar \theta)=\bar{a}$. Indeed, if it is optimal for the agent $\bar \theta$ to choose another action ${a}^s(\bar \theta)\neq \bar{a}$, it must be the case that $U^s({a}^s(\bar \theta); \bar{p}^{\bar \theta}, \bar\theta) \ge U^s(\bar{a}; \bar{p}^{\bar \theta}, \bar\theta)$, i.e.,
\[ \sum_{\omega} \bar{F}_{{a}^s(\bar \theta), \omega} ( \bar{p}^{\bar \theta}_\omega-\bar \theta \bar{c}_{{a}^s(\bar \theta)}) \ge \sum_{\omega} \bar{F}_{\bar a, \omega} ( \bar{p}^{\bar \theta}_\omega-\bar \theta \bar{c}_{\bar a}) = \bar{F}_{\bar a, \bar \omega}  \bar{p}^{\bar \theta}_{\bar \omega}\ge 0.\]
Hence, the principal's utility from type $\bar{\theta}$ is
\[\sum_{\omega} \bar{F}_{{a}^s(\bar \theta), \omega} ( r_\omega- \bar{p}^{\bar \theta}_\omega)\le \sum_{\omega} \bar{F}_{{a}^s(\bar \theta), \omega} ( 1- \bar{p}^{\bar \theta}_\omega) \le 1 - \bar \theta \bar{c}_{{a}^s(\bar \theta)}\le 1- \frac{2^{2lK+1}}{\epsilon} 2^{-2lK} \epsilon \le -1.  \]
This implies that the expected principal's utility using $\bar P$ is
\begin{align*}
    \sum_{\theta \in \bar{\Theta}}  \bar{\mu}(\theta) \sum_{\omega} \bar{F}_{ {a}^s(\theta), \omega} (r_\omega - \bar{p}^{\theta}_\omega) &=  \bar{\mu}(\bar{\theta})\sum_{\omega} \bar{F}_{{a}^s(\bar \theta), \omega} ( r_\omega- \bar{p}^{\bar \theta}_\omega) + \sum_{\theta \neq \bar \theta}  \bar{\mu}(\theta) \sum_{\omega} \bar{F}_{{a}^s(\theta), \omega} (r_\omega - \bar{p}^{\theta}_\omega)\\
    &\le -  \bar{\mu}(\bar{\theta}) + \sum_{\theta \neq \bar \theta}  \bar{\mu}(\theta)\\
    &= -(1-\alpha) + \alpha < 0,
\end{align*}
where we reach a contradiction since $\bar P$ should lead to non-negative principal utility.
Hence, ${a}^s(\bar \theta)= \bar a$.

Then, since $\bar P$ has non-negative principal utility,  we have that $ \sum_{\theta \in \bar{\Theta}}  \bar{\mu}(\theta) \sum_{\omega} \bar{F}_{ {a}^s(\theta), \omega} (r_\omega - \bar{p}^{\theta}_\omega)\ge 0$. 
This implies that \[- \bar{\mu}(\bar \theta) \bar{p}^{\bar{\theta}}_{\bar \omega}=  \bar{\mu}(\bar \theta) \sum_{\omega} \bar{F}_{{a}^s(\bar \theta), \omega} (r_\omega - \bar{p}^{\bar \theta}_\omega) \ge -\sum_{\theta \neq \bar \theta}  \bar{\mu}(\theta) \sum_{\omega} \bar{F}_{\bar{a}^s(\theta), \omega} (r_\omega - \bar{p}^{\theta}_\omega)\ge -\sum_{\theta \neq \bar \theta}  \bar{\mu}(\theta)=-\alpha   .\]
Hence, $\bar p^{\bar \theta}_{\bar \omega}\le \frac{\alpha}{1-\alpha}$.
Finally, by the IC constraint of type $\bar \theta$ it holds $U^s(\bar{a}; \bar{p}^{\bar{\theta}}, \bar{\theta}) \ge U^s(\bar{a}; \bar{p}^{\theta}, \bar{\theta})$, and we have that $\bar p^\theta_{\bar \omega}\le \bar p^{\bar \theta}_{\bar \omega}$ for each $\theta$. 
This concludes this lemma.
\end{proof}

The next lemma shows that under the contract $\bar{P}$, the best-response action of $\theta_i$ agent  ${a}^s(\theta_i) \notin A^{\theta_j}$ for any $j<i$. 
This implies a monotone-like property. In particular, agent $\theta$ could only have best-response action $a^s(\theta) \in A^{\theta'}$ with $\theta' \ge \theta$. 
Note that $\bar{\theta} > \theta_k$ for all $k \in \{1, 2, \dots, K\}$.

\begin{lemma}\label{lemmaagentincentiveaction2}
    Assume Condition \ref{conditionnonnegativeutility} holds. Under contract $\bar{P}$, the type-$\theta_i$ agent does not play a best-response action $a^s(\theta_i)\in A^{\theta_j}$ with $j<i$.
\end{lemma}
\begin{proof}
    Recall that in the constructed  instance $I^S$, the smallest cost for an action $a\in A^{\theta_j}$ is $2^{-2jl}\eps$. Suppose by contradiction that  the agent $\theta_i$ plays the best-response action $a^s(\theta_i) = a \in A^{\theta_j}$ with $j<i$. By the IR constraint, the utility of the agent is 
    \begin{align}
      {}0 \quad\le  U^s(a^s(\theta_i); \bar{p}^{\theta_i}, \theta_i) 
      = &{} \sum_{\omega\in \bar{\Omega} } \bar{F}^{\theta_j}_{a, \omega}\bar{p}^{\theta_i}_\omega - 2^{il} \bar{c}^{\theta_j}_a \notag\\
        =&\quad (1-2^{-jl})\bar{p}^{\theta_i}_{\bar \omega} + \sum_{\omega\in \Omega}2^{-jl}F^{\theta_j}_{ a, \omega}\bar{p}^{\theta_i}_\omega - 2^{il}\cdot[2^{-2jl}(c^{\theta_j}_a + \eps)] \notag \\
        =&\quad (1-2^{-jl})\bar{p}^{\theta_i}_{\bar \omega} + 2^{-jl}[\sum_{\omega\in \Omega}F^{\theta_j}_{ a, \omega}\bar{p}^{\theta_i}_\omega - 2^{il-jl}(c^{\theta_j}_a + \eps)] .\label{lemma56thetainotchoosej}
    \end{align}
    Since {$l> \log \frac{4}{\mu_{min}\eps} > \log \frac{3}{\mu_{\min} \eps}$,}  the cost term in Equation~\eqref{lemma56thetainotchoosej} is  $2^{il-jl}(c^{\theta_j}_a + \eps) > 2^{l}\eps > \frac{3}{\mu_{\min}}$, where we recall that $\mu_{\min} = \min_k \mu(\theta_k)$. 
    Thanks to {$\alpha \le \frac{1}{2^{(K+1)l}+1} <\frac{1}{\mu_{\min} 2^{Kl} + 1}$,} the expected payment on non-dummy outcomes, i.e., $\langle F^{\theta_j}_{ a}, \bar{p}^{\theta_i} \rangle$ in Equation~(\ref{lemma56thetainotchoosej}), is at least 
    \begin{align*}
        \sum_{\omega \in \Omega} F^{\theta_j}_{ a, \omega}\bar{p}^{\theta_i}_\omega \ge 2^{il-jl}(c^{\theta_j}_a + \eps) - 2^{jl}\bar{p}^{\theta_i}_{\bar{\omega}} \ge 2^{il-jl}(c^{\theta_j}_a + \eps) - 2^{jl}\frac{\alpha}{1-\alpha} \ge \frac{3}{\mu_{\min}} - \frac{2^{jl}}{\mu_{\min}2^{Kl}}  \ge \frac{2}{\mu_{\min}},
    \end{align*}
where the second inequality is due to $\bar{p}^{\theta_i}_{\bar{\omega}} \le \frac{\alpha}{1-\alpha}$ and  second-to-last inequality holds since $\frac{\alpha}{1-\alpha}$ is increasing in $\alpha$, hence upper bounded by substituting $\alpha$ with $\frac{1}{\mu_{\min} 2^{Kl} + 1}$. Therefore, if type $\theta_i$ chooses an action from $A^{\theta_j}$ with $j<i$, the principal gains negative expected utility from that agent, i.e.,
\begin{equation}\label{upeprboundnegativeichoosej}
\langle \bar{F}^{\theta_j}_{a}, r\rangle - \langle \bar{F}^{\theta_j}_{a}, \bar{p}^{\theta_i} \rangle \le 2^{-jl}(\langle F^{\theta_j}_{a}, r \rangle -\langle F^{\theta_j}_{a}, \bar{p}^{\theta_i} \rangle ) \le 2^{-jl} (1 - \frac{2}{\mu_{\min}} ) <0.
\end{equation}
Furthermore, if type $\theta_i$ has best-response action $a^s(\theta_i) \in A^{\theta_k}$ with $k\ge i$, the principal gains utility at most 
$\langle \bar{F}^{\theta_k}_{a}, r \rangle  - \langle \bar{F}^{\theta_k}_{a}, \bar{p}^{\theta_i} \rangle  \le 2^{-kl}$. If $a^s(\theta_i) =\bar{a}$, the principal gains non-positive utility.

Denote as $\Theta^{<}$ the set of types $\theta_s$ having best-response action $a^s(\theta_s) \in A^{\theta_j}$ with $j<s$, and as $\Theta^{\ge}$ the set of types $\theta_s$ having action $a^s(\theta_s) \in A^{\theta_j}$ with $j\ge s$ or $a^s(\theta_s) = \bar{a}$. Hence, $\Theta = \Theta^{<} \cup \Theta^{\ge}$.  The principal's utility is at most 
\begin{align*}
&\sum_{\theta \in \bar{\Theta}}  \bar{\mu}(\theta) \langle \bar{F}_{{a}^s(\theta)},  (r - \bar{p}^{\theta}) \rangle \\
&\le \sum_{\theta_s \in \Theta}  \bar{\mu}(\theta_s) \langle \bar{F}_{{a}^s(\theta_s)},  (r - \bar{p}^{\theta_s}) \rangle \\
&= H(\alpha, l) \Big( \sum_{\theta_s \in \Theta^{<}}\mu(\theta_s) 2^{sl} ( \langle \bar{F}_{a^s(\theta_s)}, r \rangle - \langle \bar{F}_{a^s(\theta_s)}, \bar{p}^{\theta_s} \rangle ) + \sum_{\theta_s \in \Theta^{\ge}}\mu(\theta_s) 2^{sl}(\langle \bar{F}_{{a}^s(\theta_s)}, r \rangle - \langle \bar{F}_{{a}^s(\theta_s)},  \bar{p}^{\theta_s} \rangle )\Big),
\end{align*}
where the inequality holds because the type $\bar{\theta}$ induces non-positive utility to the principal by Lemma~\ref{pbarthetabarwissmall}.

By substituting the values obtained above, we have 
\begin{align*}
&\sum_{\theta_s \in \Theta^{<}} \mu(\theta_s) 2^{sl} ( \langle \bar{F}_{a^s(\theta_s)}, r \rangle  - \langle \bar{F}_{a^s(\theta_s)}, \bar{p}^{\theta_s} \rangle ) + \sum_{\theta_s \in \Theta^{\ge}}\mu(\theta_s) 2^{sl}( \langle \bar{F}_{\bar{a}(\theta_s)}, r \rangle - \langle \bar{F}_{\bar{a}(\theta_s)}, \bar{p}^{\theta_s} \rangle ) \\
    \le {}& {} \mu(\theta_i)2^{il} 2^{-jl} (1 - \frac{2}{\mu_{\min}}) + \sum_{\theta_s \in \Theta^{\ge}} \mu(\theta_s) \\
    < {}& {} \mu(\theta_i) (1-\frac{2}{\mu_{\min}}) + (1-\mu(\theta_i)) \\
    < {}& {}  0 ,
\end{align*}
where the first inequality is due to: (i) The principal gets negative utility from every agent in $\Theta^{<}$ by Equation~\eqref{upeprboundnegativeichoosej}, which implies that we can remove all the types in $\Theta^{<}\setminus \{\theta_i\}$ from the first summation; and (ii) agent $\theta_s \in \Theta^{\ge}$ has best-response action $a^s(\theta_s) \in A^{\theta_k}$ with $k \ge s$ or $\bar a$, where the principal gains utility at most $2^{-kl}$.
Hence, we showed that the principal's utility is negative, which contradicts Condition \ref{conditionnonnegativeutility}. 
\end{proof}

Next, we show that the menu $\bar{P}$ may incentivize the agent $\theta_i$ to choose an action $a^s(\theta_i) \in A^{\theta_j}$ with $j>i$. Observation \ref{actionnotinathetai2} implies one major difference between $I^S$ and $I^M$ instances: $\theta$ agent in $I^M$ instance can only choose an action from $A^{\theta}$, while in $I^S$, $\theta$ agent has more available actions to choose.

\begin{observation}\label{actionnotinathetai2}
    Under contracts $\bar{P}$, $\theta_i$ agent may have best-response action $a^s(\theta_i) \in A^{\theta_j}$ with $j>i$.
\end{observation}
\begin{proof}
    Consider a multi-parameter instance $I^M$ with only two types $\theta_1$ and $\theta_2$. There is a reward $r_\omega = \frac{1}{2}$ for $\omega\in [m]$ and a set of $n$ actions $A^{\theta_i}$ for each type. The cost $c^{\theta_1} = 1$ while  $c^{\theta_2} = 0$. We add one additional default action $a_0$ and default outcome $\omega_0$ to each type where $r_{\omega_0}=c_{a_0} = 0$,  $F^{\theta}_{a_0, \omega_0}=1$, and $F^{\theta}_{a, \omega_0}=0$ for all $a\neq a_0$.  Then, we construct a single-parameter instance by the process $I^S = \texttt{Single}(I^M,\epsilon)$. Under any contract $\bar{P}$ leading to positive principal utility,  by Lemma \ref{lemmaagentincentiveaction2}, $\theta_2$ will only take action from $A^{\theta_2}$ or  action $\bar{a}$ in $I^S$. We define the menu $\bar{P}$ as $\bar{p}^{\bar{\theta}}_{\omega\neq \bar{\omega}} = 0$, $\bar{p}^{\theta_1}_{\omega\neq \bar{\omega}, {\omega_0}} = \bar{p}^{\theta_2}_{\omega\neq \bar{\omega}, {\omega_0}} = \frac{1}{4}$ and $\bar{p}^{\bar{\theta}}_{\bar{\omega}} =\bar{p}^{\theta_1}_{ \bar{\omega}} = \bar{p}^{\theta_2}_{ \bar{\omega}} = \bar{p}^{\bar{\theta}}_{{\omega_0}} =\bar{p}^{\theta_1}_{ {\omega_0}} = \bar{p}^{\theta_2}_{ {\omega_0}} = 0$. Obviously, this contracts $\bar{P}$ leads to non-negative principal utility and ${a}^s(\theta_1), {a}^s(\theta_2) \in a^{\theta_2}$, ${a}^s(\bar{\theta}) = \bar{a}$. 
\end{proof}

In the following, we will circumvent this problem by showing that each type $\theta_i$ has an ‘‘approximate'' best-response  in $A^{\theta_i}$. To this end, as a first step we bound expected payment on non-dummy outcomes.

\begin{lemma}\label{lm:paymentBoundedjgretateei}
    Assume Condition \ref{conditionnonnegativeutility} holds. Under the contracts $\bar{P}$, the expected payment on non-dummy outcomes $\sum_{\omega \in \Omega}{F}_{{a}^s(\theta_i), \omega} \bar{p}^{\theta_i}_\omega$ is upper bounded by $\frac{4}{\mu_{\min}}$ where  $\mu_{\min} = \min_k \mu(\theta_k)$.
\end{lemma}

\begin{proof}
Suppose by contradiction that there exists a type $\theta_i$ such that  $\sum_{\omega \in \Omega}{F}_{{a}^s(\theta_i), \omega} \bar{p}^{\theta_i}_\omega>\frac{4}{\mu_{\min}}$.  Let $\theta_j$ be the type such that the best-response action ${a}^s(\theta_i)\in A^{\theta_j} $. By Lemma~\ref{lemmaagentincentiveaction2} we have that $j\ge i$. Notice that $a^s(\theta_i) \neq \bar{a}$ since $\sum_{\omega \in \Omega }{F}_{\bar{a}, \omega} \bar{p}^{\theta_i}_\omega=\sum_{\omega \in \Omega }\bar{F}_{\bar{a}, \omega} \bar{p}^{\theta_i}_\omega = 0 < \frac{4}{\mu_{\min}}$, not satisfying the assumption for $\theta_i$. 
Then, by the IC constraint, the utility of $\theta_j$ is larger than that of misreporting to $\theta_i$ and following action ${a}^s(\theta_i)$, i.e., $U^s({a}^s(\theta_j); \bar{p}^{\theta_j}, \theta_j) \ge U^s({a}^s(\theta_i); \bar{p}^{\theta_i}, \theta_j)$ and


\[ \sum_{\omega \in \bar{\Omega} }{\bar F}_{{a}^s(\theta_j), \omega} \bar{p}^{\theta_j}_\omega-2^{jl}\bar c_{{a}^s(\theta_j)} \ge\sum_{\omega \in \bar{\Omega} }{\bar F}_{{a}^s(\theta_i), \omega} \bar{p}^{\theta_i}_\omega -2^{jl} \bar c_{{a}^s(\theta_i)}, \] 
which  implies the following lower bound of the expected payment:
\begin{align}\label{eq:largeP}
\sum_{\omega \in \bar{\Omega} }{\bar F}_{{a}^s(\theta_j), \omega} \bar{p}^{\theta_j}_\omega\ge \sum_{\omega \in \bar{\Omega} }{\bar F}_{{a}^s(\theta_i), \omega} \bar{p}^{\theta_i}_\omega -2^{jl} \bar c_{{a}^s(\theta_i)} \ge 2^{-jl} \frac{4}{\mu_{\min}} - 2^{-jl+1}\ge  2^{-jl} \frac{2}{\mu_{\min}},
\end{align}
{where the second-to-last inequality is due to ${a}^s(\theta_i) \in A^{\theta_j}$, the assumption $\sum_{\omega \in \Omega}{F}_{{a}^s(\theta_i), \omega} \bar{p}^{\theta_i}_\omega>\frac{4}{\mu_{\min}}$ and $\bar c_{{a}^s(\theta_i)} \le 2^{-2jl}(1+\eps) \le 2^{-2jl+1}$ by the construction of $\bar c$.} Moreover, by Lemma \ref{lemmaagentincentiveaction2} we know that each agent $\theta_k$ plays an action from $a^s(\theta_k) \in A^{\theta_q}$ with $q\ge k$ or $\bar{a}$. Hence, for agent $\theta_k$ playing $a^s(\theta_k) \in A^{\theta_q}$ with $q\ge k$, it holds:
\begin{align}\label{eq:smallPU}
    2^{kl}\sum_{\omega \in \bar{\Omega} }\bar{F}_{ a^s(\theta_k),\omega} r_\omega  =2^{kl}\sum_{\omega \in \bar{\Omega} }\bar{F}^{\theta_q}_{a, \omega} r_\omega    \le 2^{kl} 2^{-ql}  \sum_{\omega \in \bar{\Omega} } F^{\theta_q}_{a, \omega}\le  \sum_{\omega \in \bar{\Omega} } {F}^{\theta_q}_{a, \omega}\le 1.
\end{align}

Note that Equation~(\ref{eq:smallPU}) still holds for $a^s(\theta_k)=\bar{a}$. Hence, the total principal's utility is at most 
\begin{align*}
\sum_{\theta_i \in \bar{\Theta}}  \bar{\mu}(\theta_i) \langle \bar{F}_{{a}^s(\theta_i)}, (r - \bar{p}^{\theta_i}) \rangle
&\le H(\alpha, l)  \Big(   \sum_{\theta_s \in \Theta}\mu(\theta_s) 2^{sl}(\langle \bar{F}_{{a}^s(\theta_s)}, r \rangle - \langle \bar{F}_{{a}^s(\theta_s)},  \bar{p}^{\theta_s} \rangle )\Big) \\
& \le H(\alpha, l) \Big(  \sum_{\theta_s \in \Theta }\mu(\theta_s) 2^{sl} \langle \bar{F}_{{a}^s(\theta_s)}, r \rangle -  
 \mu(\theta_j) 2^{jl} \langle {\bar F}_{{a}^s(\theta_j)}, \bar{p}^{\theta_j} \rangle   \Big) \\
 & \le H(\alpha, l)  \Big(  \sum_{\theta_s \in \Theta }\mu(\theta_s) -  
 \mu(\theta_j) 2^{jl} \langle {\bar F}_{{a}^s(\theta_j)}, \bar{p}^{\theta_j}  \rangle  \Big)\\
 & \le H(\alpha, l) \Big(  1 -  
 \mu(\theta_j) 2^{jl} \langle {\bar F}_{{a}^s(\theta_j)}, \bar{p}^{\theta_j} \rangle  \Big)\\
 & \le H(\alpha, l)  \Big(  1 -  
 \mu(\theta_j) \frac{2}{\mu_{\min}}   \Big)\\
 &< 0,
\end{align*}
where the first inequality is due to that $\bar{\theta}$ provides non-positive principal's utility, in the third inequality we use Equation~\eqref{eq:smallPU}, and in the fifth we use Equation~\eqref{eq:largeP}.
Hence, we reach a contradiction.
\end{proof}

In the rest of the proof, we show the process used to construct the menu $P$ and its best-response action set $\{a^m(\theta)\}_{\theta \in \Theta}$ for instance $I^M$. We first construct an intermediary menu $\breve{P}$ as follows: for all $\theta \in \bar{\Theta}$, let $\breve{p}^{\theta}_{\omega} = \bar{p}^{\theta}_{\omega}$  for all $\omega \neq \bar{\omega}$ and $\breve{p}^{\theta}_{\bar{\omega}} = 0$. Moreover, we define a tuple of approximate best responses as follows: $\breve{a}(\theta) = {a}^s(\theta)$ for all $\theta$. We will show that the menu $(\breve{p}^{\theta},\breve{a}(\theta))_{\theta \in \Theta}$ is approximately IC for the single-parameter instance. 
It is worthwhile to note that  contracts $\breve{P}$ increase the principal's utility over $\bar{P}$, since the principal does not pay for the outcome $\bar{\omega}$ but the actions remain. The next lemma shows that the contract-action pairs $(\breve{p}^{\theta},\breve{a}(\theta))_{\theta \in \Theta}$ is $\eta$-IC according to Definition \ref{def:IC}.



\begin{lemma}\label{etaicproofcontracts}
    Assume Condition \ref{conditionnonnegativeutility} holds. $(\breve{p}^{\theta},\breve{a}(\theta))_{\theta \in \Theta}$ is $\eta$-IC
 to the single-parameter instance $I^S$, {where $\eta = \frac{\alpha}{1-\alpha}$}.
\end{lemma}

\begin{proof}
We start considering type $\bar{\theta}$.
Under menu $\bar{P}$, it holds $a^s(\bar{\theta}) = \bar{a}$ by Lemma \ref{pbarthetabarwissmall}. By the IC constraints of $\bar{P}$, we have that the utility of type $\bar \theta$ is larger than the utility under any other contract $\bar{p}^{\theta}$ and any action $a \in A^{\theta_j}$ for any $j$, i.e., 
\begin{align*}
 U^s(\bar{a}; \bar{p}^{\bar{\theta}}, \bar{\theta}) \ge U^s(a \in a^{\theta_j}; \bar{p}^{\theta}, \bar\theta).
\end{align*}

Equivalently, it holds:
\begin{align*}
 \sum_{\omega\in \bar{\Omega} }\bar{F}_{\bar{a}, \omega} \bar{p}^{\bar\theta}_\omega - \bar{\theta} \cdot \bar{c}_{\bar{a}}  &\ge \sum_{\omega\in \bar{\Omega} }\bar{F}^{\theta_j}_{a, \omega} \bar{p}^{\theta}_\omega - \bar{\theta} \cdot \bar{c}^{\theta_j}_{a},
\end{align*}
and
\begin{align*}
  {} \bar{p}^{\bar{\theta}}_{\bar{\omega}} + \sum_{\omega\in \Omega }\bar{F}_{\bar{a}, \omega} \bar{p}^{\bar\theta}_\omega - \bar{\theta}\cdot \bar{c}_{\bar{a}}  \ge (1-2^{-jl})\bar{p}^{\theta}_{\bar{\omega}} + \sum_{\omega\in \Omega }\bar{F}^{\theta_j}_{a, \omega} \bar{p}^{\theta}_{\omega} - \bar{\theta} \bar{c}^{\theta_j}_a.
\end{align*}
This implies
\begin{align*}
   \sum_{\omega\in \Omega }\bar{F}_{\bar{a}, \omega} \bar{p}^{\bar\theta}_\omega - \bar{\theta} \cdot \bar{c}_{\bar{a}} \ge  \sum_{\omega\in \Omega }\bar{F}^{\theta_j}_{a, \omega} \bar{p}^{\theta}_{\omega} - \bar{\theta} \bar{c}^{\theta_j}_a - \bar{p}^{\bar{\theta}}_{\bar{\omega}},
\end{align*}
and
\begin{align*}
   \sum_{\omega\in \bar{\Omega} }\bar{F}_{\breve{a}(\bar{\theta}), \omega} \breve{p}^{\bar{\theta}}_\omega - \bar{\theta}\bar{c}_{\breve{a}(\bar{\theta})} \ge  \sum_{\omega\in \bar{\Omega} }\bar{F}^{\theta_j}_{a, \omega} \breve{p}^{\theta}_\omega - \bar{\theta}\bar{c}^{\theta_j}_{a} - \eta,
\end{align*}
where the last inequality holds by the definitions of menu $\breve{P}$ and $\breve{a}(\theta)$, and $\bar{p}^{\bar \theta}_{\bar \omega} \le \eta = \frac{\alpha}{1-\alpha}$ by Lemma \ref{pbarthetabarwissmall}. Hence, we showed that the menu is $\eta$-IC for type $\bar{\theta}$. 

With similar arguments we can show the contract-action pairs are $\eta$-IC for every type $\theta_i \in \Theta$. By the IC constraints of $\bar{P}$, we have that the utility of $\theta_i$ is greater than that of misreporting to any $\theta$ and taking action $a\in A^{\hat{\theta}}$ for any $\hat{\theta}$. Formally, it holds
\begin{align*}
 U^s({a}^s(\theta_i); \bar{p}^{{\theta}_i}, {\theta}_i) \ge U^s(a \in a^{\hat\theta}; \bar{p}^{\theta}, \theta_i).
\end{align*}
Equivalently, it holds:
\begin{align*}
  \sum_{\omega\in \bar{\Omega} }\bar{F}_{{a}^s(\theta_i), \omega} \bar{p}^{\theta_i}_\omega - \theta_i \cdot \bar{c}_{{a}^s(\theta_i)}  \ge \sum_{\omega\in \bar{\Omega} }\bar{F}^{\hat\theta}_{a, \omega} \bar{p}^{\theta}_\omega - \theta_i \cdot \bar{c}^{\hat\theta}_{a},
 \end{align*}
 and
 \begin{align*}
   \bar{F}_{a^s(\theta_i), \bar{\omega}}\bar{p}^{\theta_i}_{\bar{\omega}} + \sum_{\omega\in \Omega }\bar{F}_{{a}^s(\theta_i), \omega} \bar{p}^{\theta_i}_\omega - \theta_i \cdot \bar{c}_{{a}^s(\theta_i)} \ge \bar{F}^{\hat\theta}_{a, \bar \omega} \bar{p}^{\theta}_{\bar{\omega}} + \sum_{\omega\in \Omega }\bar{F}^{\hat\theta}_{a, \omega} \bar{p}^{\theta}_\omega - \theta_i \cdot \bar{c}^{\hat\theta}_{a}.
\end{align*}
This implies
\begin{align*}
    \sum_{\omega\in \Omega }\bar{F}_{{a}^s(\theta_i), \omega} \bar{p}^{\theta_i}_\omega - \theta_i \cdot \bar{c}_{{a}^s(\theta_i)} &\ge  \sum_{\omega\in \Omega }\bar{F}^{\hat\theta}_{a, \omega} \bar{p}^{\theta}_\omega - \theta_i \cdot \bar{c}^{\hat\theta}_{a} - \bar{F}_{a^s(\theta_i), \bar{\omega}}\bar{p}^{\theta_i}_{\bar{\omega}},
\end{align*}
and
\begin{align*}
  \sum_{\omega\in \bar{\Omega} }\bar{F}_{\breve{a}(\theta_i), \omega} \breve{p}^{\theta_i}_\omega - \theta_i \bar{c}_{\breve{a}(\theta_i)} &\ge  \sum_{\omega\in \bar{\Omega} }\bar{F}^{\hat\theta}_{a, \omega} \breve{p}^{\theta}_\omega - \theta_i \cdot \bar{c}^{\hat\theta}_{a}  - \eta,
\end{align*} 
where the last line holds by the definitions of menu $\breve{P}$ and $\breve{a}(\theta)$, and $\bar{F}_{a^s(\theta_i), \bar{\omega}}\bar{p}^{\theta_i}_{\bar{w}} \le \bar{p}^{\bar \theta}_{\bar \omega} \le \eta = \frac{\alpha}{1-\alpha}$ by Lemma \ref{pbarthetabarwissmall}. This concludes the proof.
\end{proof}

By Observation \ref{actionnotinathetai2}, we know that under menu $\bar{P}$, there might exist some agent $\theta_i \in \Theta$  having $a^s(\theta_i) \in A^{\theta_j}$ with $j>i$ or $a^s(\theta_i) =\bar{a}$. We denote as $\hat{\Theta}$ this set of agents. The existence of $\hat{\Theta}$ implies that the constructed menu and actions $(\breve{P}, \{\breve{a}(\theta)\})$ might be not feasible in the instance $I^M$. To construct a feasible menu for $I^M$, we need to constrain each agent of type $\theta$ to have a best-response action $a^s(\theta) \in A^{\theta}$. Our construction of the new menu relies on the next lemma, which shows that the principal's utility from $\theta \in \hat\Theta$ under $(\breve{P}, \{\breve{a}(\theta)\})$ is small. 


\begin{lemma}\label{lm:smallLoss}
    Under contract-action pairs $(\breve{P},\{\breve{a}(\theta)\})$,  the expected utility of principal obtained from an agent $\theta_i \in \hat\Theta$ is at most $\mu(\theta_i) H(\alpha, l)  \gamma$, where {$\gamma = 2^{-l}$}.
\end{lemma}
\begin{proof}
The utility that the principal gets from a type $\theta_i \in \hat{\Theta}$ taking action $\breve{a}(\theta_i)\in A^{\theta_j}$ with $j>i$ is at most 
\begin{equation}\label{utilitylossmodificationbest}
\mu(\theta_i) 2^{il} H(\alpha, l)  2^{-jl} (\sum_{\omega\in \Omega } F^{\theta_j}_{\breve{a}(\theta_i), \omega} (r_\omega - \breve{p}^{\theta_i}_\omega) ) \le \mu(\theta_i)  H(\alpha, l)  2^{il-jl} \le\mu(\theta_i) H(\alpha, l)  \gamma.
\end{equation}
Moreover, if $\breve{a}(\theta_i) = \bar{a}$, the constraint on the utility is still satisfied where the principal's utility is at most $0$. Hence, the lemma holds.
\end{proof}


Then, we construct a new contract-action pair $(\breve{P}^*, \{\breve{a}^*(\theta)\})$. For types $\theta \in \Theta\setminus \hat{\Theta}$, let $\breve{p}^{*\theta} = \breve{p}^{\theta}$, while for each $\theta \in \hat{\Theta}$, let $\breve{p}^{*\theta} = \arg \max_{\breve{p}^{\hat \theta} \in \breve{P}} \max_a [\sum_{\omega\in \Omega }F^{\theta}_{a, \omega} \breve{p}^{\hat\theta}_\omega - c_a^{\theta}]$. The corresponding action set $\{\breve{a}^*(\theta)\}$ is modified similarly. In particular, $\breve{a}^*(\theta) = \breve{a}(\theta)$ for $\theta\in \Theta\setminus \hat{\Theta}$, and otherwise $\breve{a}^*(\theta) \in \arg\max_a  \sum_{\omega\in \Omega }F^{\theta}_{a, \omega} \breve{p}^{*\theta}_\omega - c_a^{\theta}$. In the definition of $\breve{a}^*(\theta)$ for types in $\hat{\Theta}$, tie-breaking is in favor of the principal's utility.  In other words, we restrict the type-$\theta$ agent to choose an action $\breve{a}^*(\theta)\in A^{\theta}$, which provides the maximum utility to the agent.  Lemma \ref{lm:smallLoss} ensures that each type $\theta\in\hat{\Theta}$ provides at most $\mu(\theta) H(\alpha, l)  \gamma$ utility to the principal under $(\breve{P},\{\breve{a}(\theta)\})$. So, if the principal's utility from these types is not too negative under $(\breve{P}^*,\{\breve{a}^*(\theta)\})$, the principal suffers a small loss from this modification. This is exactly what we prove in the next sequence of lemmas.
We start showing that the expected utility {of types in $\hat \Theta$} is small under the new contract-action pair $(\breve{P}^*,\{\breve{a}^*(\theta)\})$.


\begin{lemma}\label{newpstarthetainbarthetaissmall}
    Under contract-action pairs $(\breve{P}^*,\{\breve{a}^*(\theta)\})$, each agent type $\theta \in \hat\Theta$ in instance $I^M$ has expected utility at most $\delta=3\epsilon$.
\end{lemma}

\begin{proof}
Recall that in the instance $I^M$, when taking action $a\in A^{\theta_i}$,  the agent $\theta_i \in \hat\Theta$ utility under contract $\breve{p}^{*\theta_i}$ is calculated as $U^m(a; \breve{p}^{*\theta_i}, \theta_i) = \sum_{\omega\in \Omega }F^{\theta_i}_{a, \omega} \breve{p}^{*\theta_i}_\omega - c_a^{\theta_i}$. 
Moreover, recall that by construction $\breve{p}^{*\theta_i} \in \breve{P}$.  Suppose that under $\breve{P}$, agent $\theta_i \in \hat{\Theta}$ chooses an action $\breve{a}(\theta_i) \in A^{\theta_j}$ with $j>i$.  Since $(\breve{P}, \{\breve{a}(\theta)\})$ is $\eta$-IC in $I^S$, we have that $\theta_i$ agent's utility is greater than the utility choosing a contract $\breve{p}^{*\theta_i} \in \breve{P}$ (i.e., misreporting) and taking action $\breve{a}^*(\theta_i)$ up to an additive error $\eta$. Formally, it holds
\begin{align*}
U^s(\breve{a}(\theta_i) \in A^{\theta_j}; \breve{p}^{\theta_i}, \theta_i) \ge U^s(\breve{a}^*(\theta_i) \in A^{\theta_i}; \breve{p}^{*\theta_i}, \theta_i) - \eta, 
\end{align*}
and equivalently
\begin{align}
 \sum_{\omega\in \bar{\Omega} }\bar{F}^{\theta_j}_{\breve{a}(\theta_i), \omega} \breve{p}^{\theta_i}_\omega  - \theta_i \bar{c}^{\theta_j}_{\breve{a}(\theta_i)} \ge \sum_{\omega\in \bar{\Omega} }\bar{F}^{\theta_i}_{\breve{a}^*(\theta_i), \omega} \breve{p}^{*\theta_i}_\omega  - \theta_i \bar{c}^{\theta_i}_{\breve{a}^*(\theta_i)}- \eta . \label{etaicbreveagreterthanbrevestar}
\end{align}
This implies
\begin{align*}
 2^{-jl}[\sum_{\omega\in \Omega }F^{\theta_j}_{\breve{a}(\theta_i), \omega} \breve{p}^{\theta_i}_\omega - 2^{il-jl}(c_{\breve{a}(\theta_i)}^{\theta_j}+\eps)] \ge 2^{-il} [\sum_{\omega \in \Omega }F^{\theta_i}_{\breve{a}^*(\theta_i), \omega} \breve{p}^{*\theta_i}_\omega - (c_{\breve{a}^*(\theta_i)}^{\theta_i}+\eps) ] - \eta, 
\end{align*}
and then
\begin{align*}
 2^{il-jl}[\sum_{\omega\in \Omega } F^{\theta_j}_{\breve{a}(\theta_i), \omega} \breve{p}^{\theta_i}_\omega - 2^{il-jl}(c_{\breve{a}^*(\theta_i)}^{\theta_j}+\eps)] + 2^{il}\eta + \eps \ge [\sum_{\omega \in \Omega }F^{\theta_i}_{\breve{a}^*(\theta_i), \omega} \breve{p}^{*\theta_i}_{\omega} - c_{\breve{a}^*(\theta_i)}^{\theta_i} ] .
\end{align*}

Hence, the utility of type $\theta_i \in \hat{\Theta}$ in instance $I^M$ under contract $\breve{P}^*$ is upper bounded by 
\begin{align*}
U^m(\breve{a}^*(\theta_i) ;\breve{p}^{*\theta_i},  \theta_i) &\le 2^{il-jl}[\sum_{\omega\in \Omega} F^{\theta_j}_{\breve{a}(\theta_i), \omega} \breve{p}^{\theta_i}_\omega - 2^{il-jl}(c_{\breve{a}^*(\theta_i)}^{\theta_j}+\eps)] + 2^{il}\eta + \eps \\
&\le 2^{-l}\frac{4}{\mu_{\min}} + \frac{2^{il}\alpha}{1-\alpha} +\eps,
\end{align*}
where the second inequality holds by Lemma \ref{lm:paymentBoundedjgretateei} since $\sum_{\omega\in \Omega} F^{\theta_j}_{\breve{a}(\theta_i), \omega} \breve{p}^{\theta_i}_\omega = \sum_{\omega\in \Omega} F^{\theta_j}_{{a}^s(\theta_i), \omega} \bar{p}^{\theta_i}_\omega$ by the construction of $\breve{P}$.   {Since $l> \log (\frac{4}{\mu_{\min} \eps})$ and $\alpha \le \frac{1}{2^{Kl + l} + 1}$, we further have}
\begin{equation}\label{upperboundutilityagentnewcontract}
2^{-l}\frac{4}{\mu_{\min}} + \frac{2^{il}\alpha}{1-\alpha} +\eps \le \eps + \frac{2^{il}}{2^{Kl+l}} + \eps \le 2\eps + \frac{\mu_{\min} \eps}{4} \le 3\eps.
\end{equation}
That implies that, under the new contract-action pair $(\breve{p}^{*\theta}, \breve{a}^*(\theta))$, for each $\theta \in \hat{\Theta}$ the agent in the multi-parameter instance can receive at most $\delta= 3\eps$ utility. This concludes the proof for the case in which $\breve{a}(\theta_i) \neq \bar{a}$.

Finally, if $\breve{a}(\theta_i) =\bar{a}$, the left-hand side of \eqref{etaicbreveagreterthanbrevestar} is $U^s(\bar{a}; \breve{p}^{\theta_i}, \theta_i) = 0$ by the construction of $\breve{P}$. The above analysis still applies, and thus, the lemma holds. 
\end{proof}

Then, we divide the types $\hat{\Theta}$ into two classes $\hat{\Theta}_1$ and $\hat{\Theta}_2$ with $\hat{\Theta} = \hat{\Theta}_1 \cup \hat{\Theta}_2$. In particular, we let $\hat{\Theta}_1 \subseteq \hat{\Theta}$ be the set of types $\theta$ such that, under the pair  $(\breve{P}^*,\{\breve{a}^*(\theta)\})$,  the principal gains non-negative utility from type $\theta$, and we let $\hat{\Theta}_2 = \hat{\Theta} \setminus \hat{\Theta}_1$, i.e., the principal gets negative utility from type $\theta\in \hat{\Theta}_2$. One potential issue of our construction is that the principal may receive a large negative utility from type $\theta\in \hat{\Theta}_2$, which causes a large utility loss with respect to the original menu. To mitigate this issue, we build a new contract action pair $(\hat{P}^*, \{\hat{a}^*(\theta)\})$. Let $\hat{P}^* = \breve{P}^*$. For actions $\hat{a}^*$, we let $\hat{a}^*(\theta) = \breve{a}^*(\theta)$ except that for any type $\theta \in \hat{\Theta}_2$, we let $\hat{a}^*(\theta) = a^{\theta}_0$, where we recall that $a^{\theta}_0$ is the opt-out action with zero cost by Assumption \ref{assumptionoptouteachtype}. 



\begin{lemma}\label{principallossdeltaaftermodification}
    Under the contract-action pairs $(\hat{P}^*, \{\hat{a}^*(\theta)\})$, the principal receives at least $-\delta$ utility from any $\theta \in \hat{\Theta}_2$ in instance $I^M$, i.e.,
    \[ \sum_{\omega \in \Omega} F^\theta_{a^\theta_0,\omega} [r_{\omega}-\hat{p}^{*\theta}_\omega]\ge -\delta \quad \forall \theta \in \hat{\Theta}_2 . \]
\end{lemma}
\begin{proof}
    Recall that in the construction of $(\breve{P}^*, \{\breve{a}^*(\theta)\})$, it holds that $\breve{a}^*(\theta) \in \arg\max_a  \sum_{\omega\in \Omega }F^{\theta}_{a, \omega} \breve{p}^{*\theta}_\omega - c_a^{\theta}$ is an action providing the largest utility under contract $\breve{p}^{*\theta}$. This implies that
    \[
    \delta \ge U^m(\breve{a}^*(\theta) ;\breve{p}^{*\theta},  \theta) = \sum_{\omega\in \Omega }F^{\theta}_{\breve{a}^*(\theta), \omega} \breve{p}^{*\theta}_\omega - c_{\breve{a}^*(\theta)}^{\theta} \ge \sum_{\omega\in \Omega }F^{\theta}_{a^\theta_0, \omega} \breve{p}^{*\theta}_\omega - c_{a^\theta_0}^{\theta} = \sum_{\omega\in \Omega }F^{\theta}_{a^\theta_0, \omega} \breve{p}^{*\theta}_\omega,
    \]
    where the first inequality follows by Lemma~\ref{newpstarthetainbarthetaissmall}.
    Hence, under $(\hat{P}^*, \{\hat{a}^*(\theta)\})$, by $\hat{P}^* = \breve{P}^*$ it holds that the principal gains utility from $\theta \in \hat{\Theta}_2$
    \[
    \sum_{\omega\in \Omega}F^{\theta}_{a^{\theta}_0, \omega} r_\omega - \sum_{\omega\in \Omega }F^{\theta}_{a^{\theta}_0, \omega} \hat{p}^{*\theta}_\omega \ge -\delta.
    \]
\end{proof}

Below, we show that the contract-action pairs $(\hat{P}^*, \{\hat{a}^*(\theta)\})$ is $\delta$-IC in the multi-parameter instance $I^M$.



\begin{lemma}
    The contract-action pairs  $(\hat{P}^*, \{\hat{a}^*(\theta)\})$ is $\delta$-IC in the multi-parameter instance $I^M$
\end{lemma}
\begin{proof}
For each type $\theta\notin \hat{\Theta}$, we know by construction that, the modified contract-action pair $(\hat{p}^{* \theta}, \hat{a}^{*}(\theta))$ is exactly the pair $(\breve{p}^{\theta}, \breve{a}(\theta))$. Moreover, by Lemma \ref{etaicproofcontracts}, we know $(\breve{p}^{ \theta}, \breve{a}(\theta))$ is $\eta$-IC in instance $I^S$. This implies that for $\theta_i\notin \hat{\Theta}$, the following holds for any $\hat{\theta} \in \Theta$ (i.e., any $\breve{p}^{\hat\theta} \in \breve{P}$) and any action $a\in A^{\theta_i}$:
\begin{align*}
 U^s(\breve{a}(\theta_i); \breve{p}^{\theta_i}, \theta_i) \ge U^s(a \in A^{\theta_i}; \breve{p}^{\hat\theta}, \theta_i) - \eta.
 \end{align*}
 Equivalently, it holds
\[ \sum_{\omega\in \bar{\Omega} }\bar{F}_{\breve{a}(\theta_i), \omega} \breve{p}^{\theta_i}_\omega - \theta_i \bar{c}_{\breve{a}(\theta_i)} \ge  \sum_{\omega\in \bar{\Omega} }\bar{F}^{\theta_i}_{a, \omega} \breve{p}^{\hat\theta}_\omega - \theta_i \cdot \bar{c}^{\theta_i}_{a}  - \eta. \]
This implies that
\[ 2^{-il} [\sum_{\omega\in \Omega }{F}^{\theta_i}_{\breve{a}(\theta_i), \omega} \breve{p}^{\theta_i}_\omega - ({c}^{\theta_i}_{\breve{a}(\theta_i)}+\eps)] \ge 2^{-il}[\sum_{\omega\in \Omega }{F}^{\theta_i}_{a, \omega} \breve{p}^{\hat\theta}_\omega - ({c}^{\theta_i}_{a}+\eps)  - 2^{il}\eta],\]
where the inequality follows by $\breve{a}(\theta_i) \in A^{\theta_i}$.
Thus,
\[ \sum_{\omega\in \Omega }{F}^{\theta_i}_{\breve{a}(\theta_i), \omega} \breve{p}^{\theta_i}_\omega - {c}^{\theta_i}_{\breve{a}(\theta_i)} \ge \sum_{\omega\in \Omega } {F}^{\theta_i}_{a, \omega} \breve{p}^{\hat\theta}_\omega - {c}^{\theta_i}_{a}  - 2^{il}\eta, \]
and then
\[
 \sum_{\omega\in \Omega }{F}^{\theta_i}_{\hat{a}^*(\theta_i), \omega} \hat{p}^{*\theta_i}_\omega - {c}^{\theta_i}_{\hat{a}^*(\theta_i)} \ge \sum_{\omega\in \Omega } {F}^{\theta_i}_{a, \omega} \breve{p}^{\hat\theta}_\omega - {c}^{\theta_i}_{a}  - 2^{il}\frac{\alpha}{1-\alpha}, \]
 where the inequality follows by the construction of $\hat{P}^*$ and $\hat{a}^*(\theta)$ for $\theta\notin \hat{\Theta}$.
 Hence, we can conclude
 \begin{align}
  \sum_{\omega\in \Omega }{F}^{\theta_i}_{\hat{a}^*(\theta_i), \omega} \hat{p}^{*\theta_i}_\omega - {c}^{\theta_i}_{\hat{a}^*(\theta_i)} \ge \sum_{\omega\in \Omega } {F}^{\theta_i}_{a, \omega} \breve{p}^{\hat\theta}_\omega - {c}^{\theta_i}_{a}  - \frac{\mu_{\min}\eps}{4}, \label{revisedetatomueps4cic}
\end{align}
where the inequality holds since $\frac{\alpha}{1-\alpha}$ is increasing in $\alpha \le 1$ and $\alpha \le \frac{1}{2^{(K+1)l}+1}$.
Hence, by $l>\log(\frac{4}{\mu_{\min}\eps})$ we have 
\[
2^{il}\frac{\alpha}{1-\alpha} \le \frac{2^{il}}{2^{(K+1)l}} \le \frac{1}{2^l} \le \frac{\mu_{\min}\eps}{4}.
\]
Since by the construction of $\hat{P}^*$ (i.e., $\breve{P}^*$) it holds $\hat{p}^{*\theta} \in \breve{P}$ for all $\theta$ \footnote{Actually, we know that menu $\breve{P}$ is for instance $I^S$ and the outcome set is $\bar{\Omega}$, while $\hat{P}^*$ and $\breve{P}^*$ are for $I^M$ and the  set is $\Omega$. However, we know that $\breve{p}_{\bar{\omega}} = 0$ by construction and has no effects on the analysis for instance $I^M$. Hence, to be more precise here, we only consider the payments of contracts $\breve{P}$ over outcomes $\Omega$.} , Equation \eqref{revisedetatomueps4cic} implies that for any $\hat{\theta} \in \Theta$ it holds
\begin{align*}
U^m(\hat{a}^*(\theta_i); \hat{p}^{*\theta_i}, \theta_i) & = \sum_{\omega\in \Omega }{F}^{\theta_i}_{\hat{a}^*(\theta_i), \omega} \hat{p}^{*\theta_i}_\omega - {c}^{\theta_i}_{\hat{a}^*(\theta_i)} \\
&\ge \sum_{\omega\in \Omega } {F}^{\theta_i}_{a, \omega} \hat{p}^{*\hat\theta}_\omega - {c}^{\theta_i}_{a}  - \frac{\mu_{\min}\eps}{4} \\
&=  U^m(a\in A^{\theta_i}; \hat{p}^{*\hat{\theta}}, \theta_i)-\frac{\mu_{\min}\eps}{4}.
\end{align*}
Therefore, $(\hat{P}^*, \{\hat{a}^*(\theta)\})$ is $\delta $-IC for any $\theta \notin \hat{\Theta}$ since $\frac{\mu_{\min}\eps}{4} \le  3\eps= \delta $.
    
   Now, we prove that the menu is $\delta$-IC for a type $\theta \in \hat{\Theta}$.   
   For type $\theta \in \hat \Theta$, the contract-action pair $(\breve{p}^{*\theta}, \breve{a}^*(\theta))$ before modification is IC in instance $I^M$ (by the construction of $\breve{P}^*$) and the agent's type $\theta$ can achieve utility at most $\delta$ (by Lemma \ref{newpstarthetainbarthetaissmall}). Hence, under the new contract pair $(\hat{p}^{*\theta}, \hat{a}^*(\theta))$, the agent $\theta$ loses at most $\delta$ utility since the agent's utility playing action $a^\theta_0$ is non-negative. Hence, the modified contract-action pair $(\hat{p}^{*\theta}, \hat{a}^*(\theta))$ is $\delta$-IC.
\end{proof}

Finally, we apply the following lemma from \cite{castiglioni2022designing} to modify the $\delta$-IC contract into an IC contract guaranteeing a decrease of the principal's utility at most $2\sqrt{\delta}$. 
\begin{lemma}\label{matteorllemmma1}
    (Lemma 1 in \cite{castiglioni2022designing}) Given a $\delta$-IC menu in a multi-parameter instance $I^M$, there exists an IC contract, which can be constructed in polynomial time, that guarantees a decrease of the principal's utility at most $2\sqrt{\delta}$.
\end{lemma}
We directly apply Lemma \ref{matteorllemmma1} to the contract-action pair $(\hat{P}^*, \{\hat{a}^*(\theta)\})$, which outputs the final contract $P$ for the multi-parameter instance $I^M$, with the best-response actions $\{a^m(\theta)\}_{\theta \in \Theta}$. {Note that Lemma~\ref{matteorllemmma1} guarantees that if $\hat{P}^*$ is a single contract, the output $P$ is also a single contract.}  Indeed, in the menu constructed through Lemma \ref{matteorllemmma1}, we first construct an auxiliary menu of contracts  $\tilde{P}$ that for every type $\theta$, where $\tilde{p}^\theta_\omega = (1-\sqrt{\delta})\hat{p}^{*\theta}_\omega + \sqrt{\delta} r_\omega, \forall \omega \in \Omega $, then the contract finally  proposed to type $\theta$ is $p^{\theta} \in \arg\max_{\tilde{p} \in \tilde{P}} \max_a \sum_{\omega\in \Omega }F^{\theta}_{a, \omega}\tilde{p}_\omega - c^{\theta}_a$, and the best-response action is $a^m(\theta) \in \arg\max_{a} \sum_{\omega\in \Omega }F^{\theta}_{a, \omega}p^{\theta}_\omega - c^{\theta}_a$ by definition. As usual, tie-breaking is in favor of principal's utility.

Recall that the principal receives increased utility when constructing contract $\breve{P}$ from $\bar{P}$.
Constructing the  contract-action pairs $(\hat{P}^*, \{\hat{a}^{*}(\theta)\})$, by Lemma \ref{lm:smallLoss} and Lemma \ref{principallossdeltaaftermodification}, the principal's expected utility due to a type $\theta\in \hat{\Theta}$ at most decreases from $\mu(\theta) H(\alpha, l)  \gamma $ to $-\mu(\theta) H(\alpha, l)  \delta$, while the utility from a type $\theta \notin \hat\Theta$ does not change. Hence, 
the contract-action pairs $(\hat{P}^*, \{\hat{a}^{*}(\theta)\})$ applied to instance $I^M$ guarantees that
  \[
  \sum_{\theta\in \bar{\Theta}}  \bar{\mu}(\theta) \sum_{\omega\in \bar{\Omega} } \bar{F}_{ {a}^s(\theta), \omega} (r_\omega - \bar{p}^{\theta}_\omega) \le H(\alpha, l) \Big( \sum_{\theta \in \Theta}  \mu(\theta) \sum_{\omega\in \Omega } {F}^{\theta}_{\hat{a}^*(\theta), \omega} (r_\omega - \hat{p}^{*\theta}_\omega) + \gamma + \delta \Big).
  \]
Finally, by building the (IC) menu $P$ for $I^M$, the principal in the instance $I^M$ loses at most $2\sqrt{\delta}$ by  Lemma~\ref{matteorllemmma1}. Hence, 
\[
  \sum_{\theta\in \bar{\Theta}}  \bar{\mu}(\theta) \sum_{\omega\in \bar{\Omega} } \bar{F}_{ {a}^s(\theta), \omega} (r_\omega - \bar{p}^{\theta}_\omega) \le H(\alpha, l) \Big( \sum_{\theta \in \Theta}  \mu(\theta) \sum_{\omega\in \Omega } {F}^{\theta}_{{a}^m(\theta), \omega} (r_\omega - {p}^{\theta}_\omega) + \gamma + \delta + 2\sqrt{\delta} \Big).
\]
By  {$ l >\log(\frac{4}{\mu_{\min} \eps})$ and $\alpha \le \frac{1}{ 2^{(K+1)l} +1}$}, we have 
\[\gamma + \delta + 2\sqrt{\delta} = 2^{-l} + 3\eps + 2\sqrt{3\eps} < \frac{\mu_{\min} \eps}{4} + 3\eps + 2\sqrt{3\eps} < \frac{13\eps}{4} +4\sqrt{\eps} = \nu. \]
This concludes the proof, showing that 
\[
  \sum_{\theta\in \bar{\Theta}}  \bar{\mu}(\theta) \sum_{\omega\in \bar{\Omega} } \bar{F}_{ {a}^s(\theta), \omega} (r_\omega - \bar{p}^{\theta}_\omega) \le H(\alpha, l) \Big( \sum_{\theta \in \Theta}  \mu(\theta) \sum_{\omega\in \Omega} {F}^{\theta}_{{a}^m(\theta), \omega} (r_\omega - {p}^{\theta}_\omega) +\nu \Big).
\]





%% file: liability_soda.tex

\section{{When Can A Single Contract Be Effective?}}\label{sec:power-of-single}
So far we have shown that the single-parameter Bayesian contract design problem is not only computationally hard but does not admit ``simple'' optimal solutions due to the big gap between menu and single contract. In this section, we dive into the problem formulation and look to understand what makes the problem challenging. {Towards that end,   we found that the combination of  insistence on contracts with  limited liability  (i.e., $p\ge 0$) and the rank of $F$ is a potential source of hardness.} Specifically,   after removing the limited liability constraints\footnote{ In non-Bayesian cases with known agent  type, \citet{carroll2015robustness} and \citet{gottlieb2022simple} point out that if without limited liability constraints, the optimal design is trivially ``selling the firm to the agent" which guarantees the full surplus to the principal. However, this is no longer the case for Bayesian contract design since the agents may have misreporting behaviors. } {and assuming rank-$n$ for $F$}, a single contract suffices to achieve optimality in the sense that it achieves optimal revenue that is as good as the best possible revenue achievable by the  
menu of \emph{randomized} contracts \cite{castiglioni2022designing,gan2022optimal}. \emph{While this is not the main result of our paper, we view it as a useful exploration along the interesting   direction of understanding the fundamental difficulties underlying single-parameter Bayesian contract design.}

\vspace{2mm}
\noindent
\textbf{Optimal Design of a Menu of \emph{Randomized} Contracts.} We start by  introducing the menu of \emph{randomized} contracts, which naturally generalizes the menu of (deterministic) contracts as described in Section \ref{sec:prelim}. 
It turns out that lotteries  help to increase the principal's  utility and can be used to design menus of \emph{randomized} contracts; that is, after the agent reports his type, he is not assigned a contract deterministically but gets a chance to draw a contract from a pool of contracts with pre-determined probabilities.  Such menus of randomized contracts are first studied by \citet{castiglioni2022designing}.  \citet{gan2022optimal} proved a revelation-principle-like result for this principal-agent problem by generalizing an earlier result of \citet{myerson1982optimal}. Specifically, it suffices to consider menus of  randomized contracts that include a randomized contract for each type. This randomized contract samples from at most $n$ contracts,  each of which  corresponds to one agent action. In other words, once the agent reported the type $\theta$, the principal draws a contract $p^{a, \theta}$ from the randomized contract $P^\theta = \{p^{a_0, \theta}, p^{a_1, \theta}, \dots, p^{a_{n-1}, \theta}\}$ with probability $\pi(a;\theta)$. By slightly overloading the notation, we denote the menu as $P = \{p^{a_0, \theta_1}, p^{a_1, \theta_1}, \dots, p^{a_{n-1}, \theta_K}\}$. The optimization problem faced by the principal  can then be formulated as follows using a natural generalization of previous IC constratins.  
\begin{align}
\max_{P, \pi}  & \quad  U(P, \pi) \triangleq \sum_{\theta \in \Theta}  \sum_{i \in \langle n\rangle} \pi(a_i; \theta)  \mu(\theta) \cdot  \langle  F_{i }, \,  r  - p^{a_i, \theta} \rangle   \notag   \\ \label{menuofRANDcontract}
\text{s.t.} & \quad \sum_{i \in \langle n\rangle} \pi(a_i;\theta) \Big[ \langle F_{i }, p^{a_i, \theta} \rangle  -\theta \cdot c_i \Big] \ge \sum_{i \in \langle n\rangle} \pi(a_i;\theta') \max_{i'}\Big[ \langle F_{i' }, p^{a_i, \theta'} \rangle  -\theta \cdot c_{i'} \Big], \forall~ \theta, \theta' \in \Theta  \\
& \quad \sum_i \pi(a_i; \theta) = 1,  \notag \quad \forall~ \theta  \in \Theta \\  \notag 
& \quad p^{a_i, \theta}_\omega\ge 0, \, \, \forall i, \omega, \theta;  \qquad \pi(a_i; \theta) \geq 0, \, \,   \forall i,  \theta
\end{align}
When the first constraint is applied to the case of $\theta = \theta'$, it guarantees that, for every $i$,  action $a_i$ is the best response for  type-$\theta$ agent under contract $p^{a_i, \theta}$. Overall, the first constraint ensures that it is optimal for the agent  to truthfully report his private type and follow the recommended action by the sampled contract. Note that Program \eqref{menuofRANDcontract} further relaxes Program \eqref{menuofcontract} by relaxing each type $\theta$ 's best response   from a fixed action $a(\theta)$ to a {distribution $\pi(\theta)$ over $n$ actions} (coupled with its corresponding randomized contract). 

Notably, the use of $\max$ instead of $\sup$  in  Program \eqref{menuofRANDcontract}  is valid.  While \citet{gan2022optimal,castiglioni2022designing} show that in multi-parameter BCD settings, an exactly optimal contract need not to always exist unless infinite payment is used, it is not the case in single-parameter BCD, as shown by our following result (proof deferred to Appendix \ref{maximumexistformneuofrandomizecondtraint}).
\begin{proposition}\label{pporopostionmenuofrandomizedondtraintexist}
   In single-parameter Bayesian contract design, an optimal menu of randomized contracts with finite payments always exists.  
\end{proposition}

We now examine single-parameter BCD \emph{without} the limited liability constraint {under the assumption of rank-$n$ $F$}.   This corresponds to situations where the contract is allowed to ask the agent to  compensate  the principal upon bad outcomes, i.e., the transfer $p$ is negative. Our main finding is the following. 


\begin{proposition}\label{theoremfullrank}
Consider single-parameter Bayesian contract design  without the limited liability constraint. If the action-to-outcome transition matrix $F$ is of rank $n$ where $n$ is the number of actions, then a single contract suffices to achieve optimal principal utility, i.e., achieves optimal objective value of Program \eqref{menuofRANDcontract} {without limited liability}. Moreover, an optimal single contract can be computed in polynomial time. 
\end{proposition}

It is worth comparing Proposition \ref{theoremfullrank} with the result in \cite{gottlieb2022simple}. \citet{gottlieb2022simple} found that under limited liability constraints, if the probability matrix has {\em multiplicative separability} property, then a single contract is enough to achieve the same optimal utility as the menu of deterministic contracts. The multiplicative separability implies that agents of different types share exactly the same order of incentives over action, which is not usually the case in practice. In contrast, our condition on the rank of the probability matrix $F$ is easier to satisfy --- essentially it means the number of actions is less than the number of outcomes in which  case the non-degenerated $F$ matrix shall have rank $n$. {The additional cost of obtaining our stronger conclusion is to give up the limited liability constraint.}


\subsection{Proof Sketch of Proposition \ref{theoremfullrank} }

 The high-level idea of the proof is as follows. In a Bayesian contract design, there are two main constraints,  moral hazard and adverse selection, which jointly make the contract design rather challenging \cite{chade2019disentangling}. We start by constructing a natural upper bound for the menu-of-deterministic-contract problem by removing  moral hazard, and we do so by ``forcing'' the agent to take the prescribed action $a(\theta)$ if he was to take contract $p^{\theta}$. 
	Interestingly, this natural relaxation also serves as an upper bound for the optimal utility of the menu of randomized contracts. By an analysis similar to \cite{myerson1981optimal,elkind2007designing}, we show that the optimal solution to this relaxed problem turns out to  be equivalent to maximizing the virtual welfare, which can be solved in polynomial time.  Solving this relaxed problem   returns us the optimal actions and expected payments on each action for the agent. Finally, we  utilize the full-rank property to reconstruct the solution with one single contract for the general instance, which has the same expected payments and optimal actions.  

We start with a few technical preparations.  Without loss of generality \cite{alon2021contracts,alon2022bayesian}, we  assume the actions are indexed such that $0=c_0\le c_1 \le ...\le c_{n-1}$ and agent types are indexed such that $\theta_1 < \theta_2 <\dots <\theta_K$. Under these orders, the following concepts of  monotone action ``allocation'' are useful.  

\begin{definition}[\cite{alon2021contracts,alon2022bayesian}]\label{defmonotonineaction}
 The deterministic action allocation $a(\cdot): \Theta \to \langle n \rangle$ is monotone if for any $\theta>\theta'$, it holds that $c_{a(\theta)}\le c_{a(\theta')}$, i.e, the lower-type agent prefers actions with weakly larger cost. Analogously,  the monotonicity of a randomized allocation $\pi(\cdot): \Theta \to \Delta(\langle n \rangle)$ is defined such that if for any $\theta>\theta'$, it holds that $c_{\pi(\theta)}\le c_{\pi(\theta')}$,  where $c_{\pi(\theta)} = \sum_a\pi(a;\theta)  c_{a}$. 
\end{definition}

Similar to mechanism design \cite{myerson1981optimal,alaei2012bayesian,cai2016duality,alon2021contracts} the following notion of  virtual cost and virtual welfare naturally show up in our analysis. 
\begin{definition}[Virtual Cost and Welfare]\label{virtualtypediscretedef}
 The  (discrete) virtual cost for type $\theta_k$ is defined as $$\phi(\theta_k) = \theta_k + \frac{(\theta_k-\theta_{k-1}) M(\theta_{k-1})}{\mu(\theta_k)}, \qquad \text{ where } M(\theta_k) = \sum_{i=1}^k \mu(\theta_i).$$ 
 Given a feasible randomized action allocation $\pi=(\pi(\theta))_\theta$,  the expected virtual welfare is thus defined as $E_\theta[R_{\pi(\theta)} - c_{\pi(\theta)}\phi(\theta)]$,  where $c_{\pi(\theta)} = \sum_a\pi(a;\theta)  c_{a}$ and $R_{\pi(\theta)} = \sum_a  \pi(a; \theta) \sum_\omega F_{a,\omega}  r_\omega = \sum_a  \pi(a; \theta) R_a$. The expected virtual welfare for a deterministic action allocation $a=(a(\theta))_{\theta}$ is defined similarly.
\end{definition}

\vspace{2mm}\noindent 
\textbf{An Upper Bound via Removing Moral Hazard.} Let us return to the design of a menu of contracts as in Program~\eqref{menuofcontract}. We relax the problem by removing the moral hazard component and  allowing the principal to force the agent to take the action associated with his chosen contract. In this hypothetical situation, the agent can only misreport type (i.e., choosing a contract) but has to follow the action subscribed.
We denote this new design problem as $U_R(P)$, which naturally serves an upper bound for $U(P, a)$ in Program (\ref{menuofcontract}). This problem can be  formulated as follows. 
\begin{eqnarray}\label{removemonoralhazaradmenufofcontract}
\max_{P, a\in \langle n \rangle^K} && U_R(P) \triangleq  E_{\theta}\Big[\sum_{\omega} F_{a(\theta), \omega} r_\omega - \sum_\omega F_{a(\theta), \omega} p^{\theta}_\omega \Big] \\
\text{s.t.} && \sum_\omega F_{a(\theta), \omega} p^\theta_\omega  - \theta \cdot c_{a(\theta)}\ge \sum_\omega F_{a(\theta'), \omega} p^{\theta'}_\omega  - \theta\cdot c_{a(\theta')}, \quad \forall~ \theta, \theta' \notag\\
&& \sum_\omega F_{a(\theta), \omega} p^\theta_\omega  - \theta\cdot c_{a(\theta)}\ge 0, \quad \forall~ \theta \notag
\end{eqnarray}
The first constraint implies that the only deviation for the agent is misreporting type, and it is optimal for the agent to report his type truthfully.  Note that, different from Program (\ref{menuofcontract}), here we must  explicitly include the individual rational constraints  (i.e., the second constraint) since the limited liability constraint is removed.  
 The following Lemma \ref{upperboundlemmarandomized} shows that in the relaxed problem above, there is no need for the principal to design randomized contracts, intuitively, because the principal can force the agent to choose the action that benefits the principal the most. In other words, the optimal utility of $U_R(P)$ also upper bounds the utility of any menu of randomized contracts $U(P, \pi)$. The proof is deferred to Appendix \ref{app_proof_lemma_53}.

\begin{lemma}\label{upperboundlemmarandomized}
    The optimum of Program~\eqref{menuofRANDcontract} {without limited liability} is upper-bounded by the optimum of Program~\eqref{removemonoralhazaradmenufofcontract}. Moreover, the optimum of Program~\eqref{removemonoralhazaradmenufofcontract} is equal to the maximum virtual welfare, which can be found in polynomial time by solving the following problem:
    \begin{align*}
    \max_{\pi}{ } & \quad  E_\theta[R_{\pi(\theta)} - c_{\pi(\theta)}\phi(\theta)] \\ 
    \textnormal{subject to} { } & \quad c_{\pi(\theta_K)} \le c_{\pi(\theta_{K-1})} \le \dots \le c_{\pi(\theta_1)}
\end{align*}
Finally, the action allocation $\pi$ that maximizes the virtual welfare is a deterministic allocation $a=(a(\theta))_{\theta}$.
\end{lemma}

\vspace{2mm}
\noindent
\textbf{Completing the Proof.} Finally, we use the rank condition to construct a single contract, concluding the proof.  By Lemma \ref{upperboundlemmarandomized}, we know that an upper bound of the principal's utility is the maximum virtual welfare and the action profile $\{a(\theta)\}_{\theta}$ that maximizes the virtual welfare can be found in polynomial time. Given $\{a(\theta)\}_{\theta}$, Problem  (\ref{removemonoralhazaradmenufofcontract}) becomes an LP and can be solved efficiently. The solution of the LP is an optimal menu of contracts $P=\{p^\theta\}$ for Problem  (\ref{removemonoralhazaradmenufofcontract}).  In the following, we show how to compute a single contract $p$ that has the same optimal utility. To do so, we only need to ensure that it is optimal for the agent of type $\theta$ to choose action $a(\theta)$ and that each agent's type gets the same expected payment under the two contracts. We start defining the expected payment $T$ such that $T_{a(\theta)} = F_{a(\theta)}p^{\theta}$ for each $\theta \in \Theta$ and $T_i = 0$ for $a \notin \{a(\theta)\}_{\theta}$. Note that if $a(\theta) = a(\theta')=a$ for $\theta \neq \theta'$, then it must be $F_{a}p^{\theta} = F_{a}p^{\theta'}$; otherwise, one type can misreport to another and gain higher utility, violating IC in (\ref{removemonoralhazaradmenufofcontract}). 

It is easy to verify that under expected payment $T$ and best-responses $\{a(\theta)\}$, the IC and IR constraints are satisfied. Since $F$ is of rank $n$, there must exist a single contract $p$ such that $ F p=T$. Finally, we need to discuss the tie-breaking problem. In fact, if there exists some $\theta$ type agent such that the agent  gives higher principal utility by deviating to some action $\hat{a}\neq a(\theta)$ due to tie-breaking, then it contradicts the optimum of maximum virtual welfare.
This concludes the proof of the theorem.

%% file: appendix.tex
\section{Missing Proofs in Section \ref{sec:ratio-n} }\label{sec:append:gap}

\subsection{Proof of Lemma \ref{menusingleratiolowerbound} }
{At a high level, we would like to design a menu of contract $P = \{ p^1, \cdots, p^{\bar n}\}$ such that $p^i$ is preferred by type $\theta_i$ and induces him to play action $a_{i,1}$. To do so, } consider the menu that proposes to each type $\theta_i$ with $i \in [\bar n]$ a contract with payment $p^{i}_{\omega_i}=1-2^{-in-1}$ and $p^{i}_{\omega}=0$ for each $\omega \neq \omega_i$.

 
First, we verify the IC constraints; that is, any type $\theta_i$ will honestly pick $p^i$ and takes action $a_{i,1}$. 
{Let $U(a; j \leftarrow i)$ denote the utility of agent type $\theta_i$ when he chooses contract $p^j$ by reporting  type $\theta_j$  and then takes action $a$. To verify the IC constraints, we will show $U(a_{i, 1}; i \leftarrow i) \geq U(a; j \leftarrow i) $ for any $a$ and $j$. We observe that } 
\[ U(a_{i, 1}; i \leftarrow i)  = \sum_{\omega} F_{a_{i, 1}, w}p^{i}_w - \theta_i c_{a_{i, 1}} = (1-2^{-in-1})2^{-il} - 2^{il}\cdot 2^{-2il}(1-2^{-in}) = 2^{-il-in-1}.\]
We start showing that $a_{i,1}$ is type-$\theta_i$'s best response under contract $p^i$, i.e., $U(a_{i, 1}; i \leftarrow i) > U(a; i \leftarrow i)  $ for any $a \not = a_{i, 1}$. This follows a case analysis:
\begin{itemize}
    \item  $U(a_{i, 1}; i \leftarrow i) > U(a_{i, 2}; i \leftarrow i) $ because $a_{i, 1}$ and $ a_{i, 2}$ generate outcome $\omega_i$ with equal probability,   hence lead to the same expected agent's payment under contract $p^{i}$. However, the cost of action $a_{i, 1}$ is strictly lower;
    \item $U(a_{i, 1}; i \leftarrow i) > U(a_0; i \leftarrow i) $ because $U(a_0; i \leftarrow i) = 0$ as $a_0$ has  cost $0$ and  payment $0$ under $p^{i}$; 
    \item $U(a_{i, 1}; i \leftarrow i) > U(a_{j, 1}; i \leftarrow i) $ for any $j \not = i$ since $a_{j, 1}$ is costly but never leads to outcome $\omega_i$, that is the only outcome that is paid under $p^i$; 
    \item Finally, we show that  $U(a_{i, 1}; i \leftarrow i) > U(a_{j, 2}; i \leftarrow i) $ for every $j \not = i$. Observe that
\[ U(a_{j, 2}; i \leftarrow i)  = \sum_{\omega} F_{a_{j, 2}, w}p^{i}_w - \theta_i c_{a_{j, 2}} = 2^{-jl} (1-2^{-in-1}) - 2^{il} 2^{-2jl}.\]
If $j<i$, then we have $ U(a_{j, 2}; i \leftarrow i)  = 2^{-jl} (1-2^{-in-1} - 2^{il-jl}) < 2^{-jl} (1- 2^{il-jl}) < 0$. Hence, $U(a_{i, 1}; i \leftarrow i) > U(a_{j, 2}; i \leftarrow i) $.   If $j>i$, then we have $ U(a_{j, 2}; i \leftarrow i)  < 2^{-jl} = 2^{-il-(j-i)l}$. Thus, proving  $U(a_{i, 1}; i \leftarrow i) > U(a_{j, 2}; i \leftarrow i) $  reduces to showing $(j-i)l > in+1$, which holds by substituting the values of $l$, i.e., $(j-i)l \ge l = 2n^2 >n^2 + 1 \geq in+1$. 
\end{itemize}

Next, we  consider the case that agent of {type $\theta_i$} choosing contract $p^{j}$ for some $j \neq i$ and action $a_{s, 1}$ or $a_{s, 2}$. If $s = j$, then the agent will only choose $a_{s, 1}$ due to its smaller cost; otherwise, the agent will choose $a_{s, 2}$ since $\omega_j$ is the only paid outcome but $F_{a_{s, 1}, \omega_j} = 0$. The agent will not play $a_0$ neither since the expected utility is $0$. Then, we can consider only these two cases: 

\begin{itemize}
    \item Suppose  type-$\theta_i$ agent plays action $a_{j,1}$. Then, its cost is $2^{il}[2^{-2jl}(1-2^{-jn})]$ and the induced expected payment to agent is $2^{-jl}(1-2^{-jn-1})$. The expected utility is 
  $ U(a_{j, 1}; j \leftarrow i)  = 2^{-jl}(1-2^{-jn-1}) - 2^{il}[2^{-2jl}(1-2^{-jn})]$. We want to show  $ U(a_{i, 1}; i \leftarrow i)   > U(a_{j, 1}; j \leftarrow i) $. 
The argument is divided into two cases.
If $j<i$, {to prove  
\begin{align*}
    U(a_{i, 1}; i \leftarrow i)   > U(a_{j, 1}; j \leftarrow i), 
\end{align*}
it is equivalent to showing that 
\begin{align*}
2^{-jl}(1-2^{-jn-1}) - 2^{il}[2^{-2jl}(1-2^{-jn})] < 2^{-il-in-1}, 
\end{align*}
and then
\begin{align}
(1-2^{-jn-1}) -2^{il-jl}(1-2^{-jn}) <  2^{jl-il-in-1}.\label{uiigeujilefthidnsizde}
\end{align}
}
Since $i>j\ge 1$, we have that the above inequality holds since the left-hand side of (\ref{uiigeujilefthidnsizde}) is negative. In particular, it holds
\begin{align*}
(1-2^{-jn-1}) -2^{il-jl}(1-2^{-jn}) & < 1-2^{l}(1-2^{-jn}) \\
& \le 1-2^{l}(1-2^{-n}) < 0.
\end{align*}


If $j>i\ge 1$, then we have 
\begin{eqnarray*}
  U(a_{j, 1}; j \leftarrow i)  & = &   2^{-jl}(1-2^{-jn-1}) - 2^{il}[2^{-2jl}(1-2^{-jn})]  \\
  & = & 2^{-jl} - 2^{-jl}2^{-jn-1} - 2^{il}[2^{-2jl}(1-2^{-jn})] \\ 
    &< &  2^{-jl}   \\
    & =&   2^{-il - (j-i)l}   \\
    & < & 2^{-il-in-1} =  U(a_{i, 1}; i \leftarrow i) 
\end{eqnarray*}



 \item Suppose the agent takes action $a_{s, 2}$ for $s\neq j$. The agent's expected utility is 
 \begin{eqnarray*}
     U(a_{s, 2}; j \leftarrow i) &=& 2^{-sl}(1-2^{-jn-1}) - 2^{il}\cdot 2^{-2sl}  \\ 
     &=&  2^{-sl}(1-2^{-jn-1} - 2^{il-sl}).
 \end{eqnarray*} If $s\le i$, then $1\le2^{il-sl}$ and $U(a_{s, 2}; j \leftarrow i) < 0 \leq U(a_{i, 1}; i \leftarrow i) $.  
If $s>i$,  then we have $2^{-sl} = 2^{-il - (s-i)l} < 2^{-il-in-1}$. Therefore, $ U(a_{s, 2}; j \leftarrow i) < 2^{-sl} < 2^{-il-in-1} =  U(a_{i, 1}; i \leftarrow i)$. 
\end{itemize} 


We have shown that type-$\theta_i$ agent will choose contract $p^i$ and take action $a_{i, 1}$. Hence, the principal derives expected utility from agent's type $\theta_i$ at least  $2^{-il}-2^{-il}(1-2^{-in-1})=2^{-il-in-1}$.
Thus, the overall expected principal's utility is at least $\sum_{i \in [\bar n]} 2^{-il-in-1} \cdot 2^{in+il}/C =  \bar n/(2C) $.

\subsection{Proof of Lemma \ref{menutosingleratioupperbound} }
    Given an arbitrary single contract $p$, let $K(p)$ denote the set of types $\theta_i$, $i \in [\bar n]$, such that the  principal derives expected utility at least $2^{-in-il-n/2}$ from this type-$\theta_i$ agent. 
    If $K(p) = \emptyset$, then the principal's utility is at most 
    \[
    \sum_{i\in[\bar n]} 2^{-in-il-n/2} \cdot \frac{2^{in+il}}{C} = \frac{\bar n 2^{-n/2}}{C} < 3/C.
    \]
    
    Next, we consider the case when $K(p)$ is not empty. As a first step, we show that for any contract $p$, it holds that $|K(p)|\le 2$. Suppose by contradiction that there exists a contract $p$ such that $|K(p)|\ge 3$.
    We first show that any types $\theta_i \in K(p)$ plays action $a_{i,1}$ due to our definition of $K(p)$. The following case analysis shows that the agent $\theta_i \in K(p)$ cannot take an action other than  $a_{i,1}$. 
    \begin{itemize}
        \item  $a_{i,2}$ cannot be best responses for the agent $\theta_i \in K(p)$ because the expected principal's utility under $a_{i,2}$   is at most {the welfare of action $a_{i,2}$, i.e, $\sum_{\omega}F_{a_{i, 2}, \omega}r_{\omega} - \theta_i c_{a_{i, 2}} = 2^{-il}-2^{il}2^{-2il} = 0$,} 
        contradicting $\theta_i \in K(p)$. 
        \item  $a_{j,1}$ and $a_{j,2}$ with $1\le j<i$ cannot be  a best response for agent $\theta_i \in K(p)$. Indeed, the principal's utility from $a_{j,1}$ or $a_{j,2}$ is at most $2^{-jl} - 2^{il}2^{-2jl}(1-2^{-jn})$. Then, we note that   
\[
    2^{-jl} - 2^{il}2^{-2jl}(1-2^{-jn}) < 2^{-in-il-n/2}\]
    if and only if
\[
     1-2^{(i-j)l}(1-2^{-jn}) < 2^{-in-(i-j)l-n/2}.
\]
The above inequality holds since $2^{(i-j)l}(1-2^{-jn}) > 2^{l}(1-2^{-n}) = 2^{2n^2}(1-2^{-n})>1$, which implies that the left hand side of the inequality is negative. 
        \item Finally, $a_{j,1}$ and $a_{j,2}$ with $j>i$ cannot be best responses for agent $i \in K(p)$. In this case, the upper bound of principal utility is $2^{-jl} \le 2^{-in-il-n/2}$, which holds due to 
\[2^{-jl} \le 2^{-in-il-n/2} \]
if and only if 
\[ 2^{il-jl} \le 2^{-in-n/2}.  \]  
This inequality  holds true because $2^{il-jl} \le 2^{-2n^2} \le 2^{-n^2-n/2} \le 2^{-in-n/2}$. 
    \end{itemize}  
Hence, we have shown that all types $\theta_i\in K(p)$ take action $a_{i,1}$.
Moreover, since the principal derives expected utility at least $2^{-in-il-n/2}$ from type $\theta_i$ agent, the expected payment when the played action is $a_{i,1}$ is at most $2^{-il}(1-2^{-in-n/2})$. 



    Let $i_1$ be the smallest index of types in $K(p)$.
    We show that  in an optimal contract, the payment on the dummy outcome $\omega_{\emptyset}$ must be sufficiently small, and in particular $p_{\omega_{\emptyset}}\le 2^{-l}$. 
    Indeed, we have that action $a_{i_1,1}$ is incentivized over $a_0$ for type $\theta_{i_1}$,  hence by the IC constraint 
		\[2^{-i_1l} p_{\omega_{i_1}} +2^{-i_1l} p_{\omega_{+}}+p_{\omega_{\emptyset}} (1-2^{-i_1l+1})-2^{-i_1l}(1-2^{-i_1 n}) \ge p_{\omega_{\emptyset}}. \]
	Since the expected payment to type $\theta_{i_1}$ when the agent plays $a_{i_1,1}$ is at most $ 2^{-i_1l}(1-2^{-i_1n-n/2}) \le 2^{-i_1l}$, 
 this implies that $p_{\omega_{\emptyset}}\le 2^{-i_1l}\le 2^{-l}$.
Moreover, we show that the payment on outcome $\omega_+$ satisfies:
		\[p_{\omega_{+}}\le 1- 2^{-i_1n-n/2}.\]
    Indeed, since the expected payment to type $\theta_{i_1}$ is at most $2^{-i_1l}(1-2^{-i_1n-n/2})$,  it must be the case that \[ 2^{-i_1l} p_{\omega_{+}} \le 2^{-i_1l}(1-2^{-i_1n-n/2}),\]
    which directly implies that  
    $p_{\omega_{+}} \le 1- 2^{-i_1n-n/2}$. 
    Let $i_2$ be the second smallest index of types in $K(p)$. Since type $\theta_{i_2}$ plays $a_{i_2,1}$ under contract $p$, it must be the case that $a_{i_2,1}$ is preferred by $\theta_{i_2}$ over $a_0$. Hence,
		\[p_{\omega_{i_2}} 2^{-i_2l}+p_{\omega_{+}} 2^{-i_2l}+ (1-2^{-i_2l+1}) p_{\omega_{\emptyset}} - 2^{-i_2l}(1-2^{-i_2n})\ge p_{\omega_{\emptyset}},\]
    from which we can lower bound the payment for outcome $\omega_{i_2}$ as follows:
    \begin{align}\label{eq:contr}
    p_{\omega_{i_2}} &\ge \Big[ p_{\omega_{\emptyset}} - (1-2^{-i_2l+1}) p_{\omega_{\emptyset}} - p_{\omega_{+}} 2^{-i_2l} + 2^{-i_2l}(1-2^{-i_2n})\Big] 2^{i_2 l} \notag \\
    & = 2 p_{\omega_{\emptyset}} -p_{\omega_{+}}  + (1-2^{-i_2n}) \notag\\
    &\ge (1-2^{-i_2n}) -(1-2^{-i_1n-n/2})\ge  2^{-i_1n-n/2}- 2^{-i_1n-n} \notag\\
    &= 2^{-i_1n-n/2} (1-2^{-n/2})\ge  2^{-i_1n-n/2-1}
    \end{align}
    Finally, take the third smallest index of types $i_3 \in K(p)$. Since $a_{i_3,1}$ is preferred by $\theta_{i_3}$ over $a_{i_3,2}$ under contract $p$, it must hold:
	\begin{align*}
&  (p_{\omega_{i_3}}+p_{\omega_{+}})2^{-i_3l}+p_{\omega_{\emptyset}} (1-2^{-i_3l+1})-2^{-i_3l}(1-2^{-i_3n}) \\ 
 &  \qquad \ge  (p_{\omega_{i_3}}+p_{\omega_{+}})2^{-i_3l}+\sum_{j\neq i} p_{\omega_j} 2^{-i_3l}+ p_{\omega_{\emptyset}} (1-2^{-i_3l}(\bar n+1))-2^{-i_3
  l},  
  \end{align*}
    which implies an upper bound of the payment for outcome $\omega_{i_2}$ as follows:
    \begin{align*}
        p_{\omega_{i_2}}& \le 2^{i_3l}\Big( p_{\omega_{\emptyset}} (1-2^{-i_3l+1})-2^{-i_3l}(1-2^{-i_3n}) - p_{\omega_{\emptyset}} (1-2^{-i_3l}(\bar n+1)) + 2^{-i_3l}\Big) \\
        &\le  2^{i_3l}\Big( p_{\omega_{\emptyset}} (2^{-i_3l}(\bar{n}+1) - 2^{-i_3l+1}) + 2^{-i_3l}2^{-i_3n}\Big) \\
        & \le p_{\omega_{\emptyset}} (\bar n+1) + 2^{-i_3n} \\
        & \le 2^{- i_3 n} \left(2^{-n^2}(\bar n+1)\right) + 2^{-i_3n}\le 2^{- i_3 n}  + 2^{-i_3n} \\
        & \le 2^{-i_2n-n/2},
    \end{align*}
where in the fourth inequality, we use $p_{\omega_{\emptyset}}\le 2^{-l}$, and the last inequality is by $i_{3}\ge i_2 +1$.
This reaches a contradiction to Equation \eqref{eq:contr} since $i_2 \ge i_1+1$.

 Now, we are prepared to upper bound  the principal's expected utility under the optimal contract $p^*$. Notice that for each agent's type $\theta_i$, $i \in  [\bar n]$, the expected principal's utility is at most $2^{-il-in}$. Hence, for each $\theta_i \in K(p^*)$, it holds that the expected principal's utility multiplied by the probability of the type is at most
 \[ \mu(\theta_i) 2^{-il-in}=\frac{2^{in+il}}{C} \cdot  2^{-il-in}  \le 1/C,\] 
 while for all the types $\theta_i \notin K(p)$, we have that the expected principal utility multiplied by the probability of  the type is at most
 \[  \mu(\theta_i) 2^{-il-in-n/2} = \frac{2^{in+il}}{ C} \cdot   2^{-il-in-n/2} \le \frac{2^{-n/2}}{C} \]
 by definition of $K(p^*)$.
Since $|K(p^*)|\le 2$, the expected principal's utility is at most
\[ 2 \cdot \frac{1}{ C} + (\bar n-2) \cdot  \frac{2^{-n/2}}{ C }\le 3/C. \]
This conlcudes the proof.

\section{Missing Proofs in Section \ref{section:equvivalentce}}
\subsection{Recovering Theorem \ref{upperlowertheoremsingleparameter}} \label{app_recovering_theorem31}
\begin{proposition}\label{recvoingprooftheorem31} (Recovering Theorem \ref{upperlowertheoremsingleparameter}) 
    In the single-parameter setup, the approximation ratio between a menu and a single contract is lower bound by $\Omega(n)$, where $n$ is the number of actions.
\end{proposition}
\begin{proof}
    Theorem $3.1$ in \cite{guruganesh2034contracts} shows that there exists an instance $I^M$ with $3$ agent's types and $\bar{n}$ 
    actions such that $\frac{OPT^M_{\textnormal{menu}}}{OPT^M_{\textnormal{single}}} \ge \Omega(\bar{n})$, where $OPT^M_{\textnormal{menu}}$ (resp. $OPT^M_{\textnormal{single}}$ ) denotes the utility in a multi-parameter instance of an optimal menu of contracts (resp. single contract). 
    We build a single-parameter instance $I^S = \texttt{Single}(I^M,\epsilon)$ with $n=3\bar{n}+1$ actions, and let $OPT^S$ denote the optimal utility for $I^S$. By Lemma \ref{multitosingleconstrucstion}, we know that for both single contracts and menus of contracts, 
    \[
    OPT^S_{*} \ge H(\alpha, l)[OPT^M_{*} - 2\eps], \quad \forall * \in \{\textnormal{menu}, \textnormal{single}\},
    \]
    where $OPT^S_{menu}$ and $OPT^S_{single}$ denote the optimal utility for menus and single contracts, respectively. 
    By Lemma \ref{lemma:single-multi}, we have 
    \[
    OPT^S_{*} \le H(\alpha, l)[OPT^M_{*} + \frac{13}{4}\eps + 4\sqrt{\eps}], \quad \forall * \in \{\textnormal{menu}, \textnormal{single}\}.
    \]
    Then, we have that
    \[
    \frac{OPT^S_{\textnormal{menu}}}{OPT^S_{\textnormal{single}}} \ge \frac{H(\alpha, l) [ OPT^M_{\textnormal{menu}} -2\eps]}{H(\alpha, l) [OPT^M_{\textnormal{single}} + \frac{13}{4}\eps + 4\sqrt{\eps}]} \ge \Omega(\bar{n})=\Omega(n),
    \]
{where the second inequality follows taking $\eps$ sufficiently small and the last equality is by $n=3\bar{n}+1$.}
\end{proof}

\section{Missing Proofs in Section \ref{sec:power-of-single} }

\subsection{Proof of Proposition \ref{pporopostionmenuofrandomizedondtraintexist}}\label{maximumexistformneuofrandomizecondtraint}

The problem of design an optimal menu of randomized contracts is formulated as follows:
\begin{eqnarray}
\sup_{P, \pi} && U(P, \pi) \triangleq \sum_{\theta \in \Theta}  \sum_{i \in \langle n\rangle} \pi(i; \theta)  f(\theta) \cdot  \langle  F_{i }, \,  r  - p^{i, \theta} \rangle   \notag   \\
\text{s.t.} && \sum_{i \in \langle n\rangle} \pi(i;\theta) \Big[ \langle F_{i }, p^{i, \theta} \rangle  -\theta \cdot c_i \Big] \ge \sum_{i \in \langle n\rangle} \pi(i;\theta') \max_{i'}\Big[ \langle F_{i' }, p^{i, \theta'} \rangle  -\theta \cdot c_{i'} \Big], \quad \forall~ \theta, \theta' \in \Theta,  \\
&& \sum_i \pi(i; \theta) = 1,  \notag \quad \forall~ \theta  \in \Theta \\  \notag 
&& p^{i, \theta}_\omega\ge 0, \, \, \forall i, \omega, \theta;  \qquad \pi(i; \theta) \geq 0, \, \,   \forall i,  \theta
\end{eqnarray}

Following the method in \cite{gan2022optimal}, we use $z^{i, \theta} = \pi(i; \theta)p^{i, \theta}$ to replace the terms and relax the program to an LP,
\begin{eqnarray}
\sup_{z, \pi} && U(z, \pi) \triangleq \sum_{\theta \in \Theta}  \sum_{i \in \langle n\rangle}   f(\theta) \cdot  \langle  F_{i }, \,  \pi(i; \theta)r  - z^{i, \theta} \rangle   \notag   \\
\text{s.t.} && \sum_{i \in \langle n\rangle}  \Big[ \langle F_{i }, z^{i, \theta} \rangle  -\theta \cdot \pi(i;\theta)c_i \Big] \ge \sum_{i \in \langle n\rangle}  \max_{i'}\Big[ \langle F_{i' }, z^{i, \theta'} \rangle  -\theta \cdot \pi(i;\theta')c_{i'} \Big], \quad \forall~ \theta, \theta' \in \Theta,  \\
&& \sum_i \pi(i; \theta) = 1,  \notag \quad \forall~ \theta  \in \Theta \\  \notag 
&& z^{i, \theta}_\omega\ge 0, \, \, \forall i, \omega, \theta;  \qquad \pi(i; \theta) \geq 0, \, \,   \forall i,  \theta
\end{eqnarray}
Note that there are two types of results: i) {\it regular case}: $z^{i, \theta}\ge 0$ while $\pi(i; \theta) > 0$; and ii) {\it irregular case}: $z^{i, \theta}> 0$ but $\pi(i; \theta) = 0$. For the regular cases, we can recover the contract by setting $p^{i, \theta} = z^{i, \theta}/\pi(i; \theta)$. 
So, it is the irregular case that hinders the existence of maximum. For the irregular cases, we show that we can construct a solution for $U(P, \pi)$ that gives the same utility of $U(z, \pi)$.

Suppose to have an irregular solution. Without loss of generality, we only consider the construction for a type $\theta$. Denote the set $R$ as the irregular contracts index: $i\in R$ is irregular such that $z^{i, \theta} \neq 0$ but $\pi(i; \theta) = 0$, while $i\in \langle n\rangle \setminus R$ are regular.
The type $\theta$ utility in $U(z, \pi)$ is
\[
U(\theta, \theta) = \sum_{i \in \langle n\rangle}  \Big[ \langle F_{i }, z^{i, \theta} \rangle  -\theta \cdot \pi(i;\theta)c_i \Big] = \sum_{i \in \langle n\rangle \setminus R }  \Big[ \langle F_{i }, z^{i, \theta} \rangle  -\theta \cdot \pi(i;\theta)c_i \Big] + \sum_{i \in R} \langle F_{i }, z^{i, \theta} \rangle.
\]
The utility of $\theta'$ misreporting to $\theta$ is 
\[
U(\theta', \theta) = \sum_{i \in \langle n\rangle}  \max_{i'}\Big[ \langle F_{i' }, z^{i, \theta} \rangle  -\theta' \cdot \pi(i;\theta)c_{i'} \Big] = \sum_{i \in \langle n\rangle\setminus R}  \max_{i'}\Big[ \langle F_{i' }, z^{i, \theta} \rangle  -\theta' \cdot \pi(i;\theta)c_{i'} \Big] + \sum_{i \in R} \langle F_{i }, z^{i, \theta} \rangle.
\]
We construct the solution for $U(P, \pi)$ that achieves the same principal utility:
\begin{itemize}
    \item Let $\pi(i; \theta)$ from $U(z, \pi)$ be the same for $U(P, \pi)$;
    \item Let contract $p^{i, \theta} = 0$ for $i\in R$;
    \item For each $i\in \langle n\rangle \setminus R$, let $p^{i, \theta}_{\omega} = z^{i, \theta}_{\omega}/\pi(i; \theta) + \sum_{i \in R} \langle F_{i }, z^{i, \theta} \rangle$.
\end{itemize}
We first show that the utility of agent $\theta$ does not change:
\begin{align*}
     \sum_{i \in \langle n\rangle} \pi(i;\theta) \Big[ \langle F_{i }, p^{i, \theta} \rangle  -\theta \cdot c_i \Big] 
    &=  \sum_{i \in \langle n\rangle \setminus R} \pi(i;\theta) \Big[ \langle F_{i }, z^{i, \theta}/\pi(i; \theta) + \sum_{i \in R} \langle F_{i }, z^{i, \theta} \rangle \rangle  -\theta \cdot c_i \Big] \\
     &= \sum_{i \in \langle n\rangle \setminus R} \pi(i;\theta) \Big[ \langle F_{i }, z^{i, \theta}/\pi(i; \theta)\rangle  + \sum_{i \in R} \langle F_{i }, z^{i, \theta} \rangle   -\theta \cdot c_i \Big] \\
    &= \sum_{i \in \langle n\rangle \setminus R }  \Big[ \langle F_{i }, z^{i, \theta} \rangle  -\theta \cdot \pi(i;\theta)c_i \Big] + \sum_{i \in R} \langle F_{i }, z^{i, \theta} \rangle\\
     &= U(\theta, \theta).
\end{align*}
Next, we verify that for any $\theta'$ misreporting to $\theta$, the IC constraints will not be violated. The utility of $\theta'$ misreporting to $\theta$ in $U(P, \pi)$ is 
\begin{align*}
      \sum_{i \in \langle n\rangle} \pi(i;\theta) \max_{i'}\Big[ \langle F_{i' }, p^{i, \theta} \rangle  -\theta' \cdot c_{i'} \Big] 
    &=  \sum_{i \in \langle n\rangle \setminus R} \pi(i;\theta) \max_{i'} \Big[ \langle F_{i' }, z^{i, \theta}/\pi(i; \theta) + \sum_{i \in R} \langle F_{i }, z^{i, \theta} \rangle \rangle  -\theta' \cdot c_{i'} \Big] \\
     &= \sum_{i \in \langle n\rangle \setminus R} \pi(i;\theta) \max_{i'} \Big[ \langle F_{i' }, z^{i, \theta}/\pi(i; \theta)\rangle  + \sum_{i \in R} \langle F_{i }, z^{i, \theta} \rangle   -\theta' \cdot c_{i'} \Big] \\
     &= \sum_{i \in \langle n\rangle \setminus R } \max_{i'}  \Big[ \langle F_{i }, z^{i, \theta} \rangle  -\theta' \cdot \pi(i;\theta)c_{i'} \Big] + \sum_{i \in R} \langle F_{i }, z^{i, \theta} \rangle\\
    &= U(\theta', \theta).
\end{align*}

Finally, we verify that the principal's utility gain from type $\theta_i$ agent does not change.
\begin{align*}
      \sum_{i \in \langle n\rangle} \pi(i; \theta)   \cdot  \langle  F_{i }, \,  r  - p^{i, \theta} \rangle 
    &=  \sum_{i \in \langle n\rangle \setminus R} \pi(i; \theta)   \cdot  \langle  F_{i }, \,  r  - ( z^{i, \theta}/\pi(i; \theta) + \sum_{i \in R} \langle F_{i }, z^{i, \theta} \rangle)\rangle \\
    &=  \sum_{i \in \langle n\rangle \setminus R} \pi(i; \theta)   \cdot  \Big( \langle  F_{i }, \,  r  -  z^{i, \theta}/\pi(i; \theta)  \rangle - \sum_{i \in R} \langle F_{i }, z^{i, \theta} \rangle \Big) \\
    &=  \sum_{i \in \langle n\rangle \setminus R} \pi(i; \theta)   \cdot  \langle  F_{i }, \,  r  -  z^{i, \theta}/\pi(i; \theta)  \rangle - \sum_{i \in R} \langle F_{i }, z^{i, \theta} \rangle \\
    &=  \sum_{i \in \langle n\rangle \setminus R}    \langle  F_{i }, \,  \pi(i; \theta) r  -  z^{i, \theta}  \rangle - \sum_{i \in R} \langle F_{i }, z^{i, \theta} \rangle \\ 
    &=  \sum_{i \in \langle n\rangle}    \langle  F_{i }, \,  \pi(i; \theta) r  -  z^{i, \theta}  \rangle .
\end{align*}
Hence, we show that there exists a solution for the original problem $U(P, \pi)$ to achieve the same utility as $U(z, \pi)$. It implies that the maximum exists.

\subsection{Proof of Lemma \ref{upperboundlemmarandomized}}\label{app_proof_lemma_53}

This section is devoted to proving Lemma \ref{upperboundlemmarandomized}.  Similar to $U_R(P)$, we define a relaxed problem $U_R(P, \pi)$ where we remove the moral hazard by forcing the agent to follow the recommended action even misreporting types.
\begin{align*}
	\max_{P, \pi} &\quad U_R(P, \pi) \triangleq \sum_{\theta} \mu(\theta)\Big\{ \sum_a \pi(a; \theta) \sum_\omega F_{a,\omega} (r_\omega - p_\omega^{a, \theta})\Big\} \\
	\text{s.t.} &\quad \sum_a \pi(a;\theta) \Big[\sum_\omega F_{a, \omega}p^{a, \theta}_\omega -\theta \cdot c_a \Big]\ge \sum_a \pi(a;\theta') \Big[\sum_\omega F_{a, \omega}p^{a, \theta'}_\omega -\theta \cdot c_{a} \Big], \quad \forall~ \theta, \theta'\\
	&\quad {\sum_a \pi(a;\theta) \Big[\sum_\omega F_{a, \omega}p^{a, \theta}_\omega -\theta \cdot c_a \Big] \ge 0 } \quad \forall \theta\\
	&\quad \sum_a \pi(a; \theta) = 1 \quad \forall \theta \\
	&\quad \pi(a; \theta) \ge 0 \quad \forall \theta, a 
\end{align*}
Similar to Program~\eqref{removemonoralhazaradmenufofcontract}, we need to write down the individual constraint explicitly. Then, we rewrite the program by expanding the formulation as follows:
\begin{align*}
	\max_{P, \pi} &\quad  U_R(P, \pi) \triangleq \sum_{\theta} \mu(\theta)\Big\{ \sum_a  \sum_\omega F_{a,\omega} \pi(a; \theta) r_\omega - \sum_a  \sum_\omega F_{a,\omega} \pi(a; \theta)p_\omega^{a, \theta}\Big\} \\
	\text{s.t.} &\quad \sum_a \sum_\omega F_{a, \omega} \pi(a; \theta)p_\omega^{a, \theta} -\theta \cdot \sum_a \pi(a;\theta) c_a \ge \sum_a \sum_\omega F_{a, \omega}  \pi(a; \theta')p_\omega^{a, \theta'} -\theta \cdot \sum_a\pi(a;\theta')  c_{a} , \ \ \forall~ \theta, \theta'\\
	&\quad {\sum_a \pi(a;\theta) \sum_\omega F_{a, \omega}p^{a, \theta}_\omega -\theta \cdot\sum_a \pi(a;\theta) c_a  \ge 0 } \quad \forall \theta\\
	&\quad  \sum_a \pi(a; \theta) = 1 \quad \forall \theta\\
	& \quad \pi(a; \theta) \ge 0   \quad \forall \theta, a
\end{align*}
Since distribution $\pi(\cdot;\theta)$ and $P^\theta$ are only related to type $\theta$ agent, and the agent always follows the recommendation by assumption, we can define an equivalent program with variables  $z^\theta = \sum_a \sum_\omega F_{a, \omega} \pi(a; \theta)p_\omega^{a, \theta}$. Indeed, note that given a $z^\theta$ and arbitrary distribution $\pi(a; \theta)$, one can always find a menu $p^\theta$ such that $\sum_a \sum_\omega F_{a, \omega} \pi(a; \theta)p_\omega^{a, \theta} = z^\theta$. Therefore, replacing $\sum_a \sum_\omega F_{a, \omega} \pi(a; \theta)p_\omega^{a, \theta}$ with $z^\theta$ does not change the optimality of the problem. 
Hence, we can reformulate the program defining $U_R(P,\pi)$ as follows:
\begin{align*}
	\max_{\pi, z} &\quad U_R(z, \pi) \triangleq \sum_{\theta} \mu(\theta)\Big\{ \sum_a  \sum_\omega F_{a,\omega} \pi(a; \theta) r_\omega - z^\theta\Big\} \\
	\text{s.t.} &\quad z^\theta -\theta \cdot \sum_a \pi(a;\theta) c_a \ge z^{\theta'} -\theta \cdot \sum_a\pi(a;\theta')  c_{a} , \quad \forall~ \theta, \theta'\\
	&\quad {z^\theta -\theta \cdot \sum_a \pi(a;\theta) c_a \ge 0} \quad \forall~ \theta \\ 
	&\quad  \sum_a \pi(a; \theta) = 1 \quad \forall~ \theta
\end{align*}
If we can show that in the optimal solution, $\pi(a; \theta)$ recommends deterministically an action, i.e., $\pi(a; \theta)=1$ for some $a$ while $\pi(k; \theta)=0$ for all $k\neq a$, we then obtain a deterministic mechanism, i.e., a feasible solution to Program~\eqref{removemonoralhazaradmenufofcontract}. This will prove the first part of the proposition. 

We use the following notations for expected costs and rewards of $\theta$-type agent induced by a distribution $\pi(\theta)$: $c_{\pi(\theta)} = \sum_a\pi(a;\theta)  c_{a}$ and $R_{\pi(\theta)} = \sum_a  \pi(a; \theta) \sum_\omega F_{a,\omega}  r_\omega = \sum_a  \pi(a; \theta) R_a$.
Here, we rewrite the above problem with the new notation:
\begin{subequations} \label{pr:random}
	\begin{align}
		\max_{\pi, z} && U_R(z, \pi) \triangleq \sum_{\theta} \mu(\theta)\Big\{ R_{\pi(\theta)} - z^\theta\Big\} \\
		\text{s.t.} && z^\theta -\theta \cdot c_{\pi(\theta)} \ge z^{\theta'} -\theta \cdot c_{\pi(\theta')} , \quad \forall~ \theta, \theta'\\
		&& {z^\theta -\theta \cdot c_{\pi(\theta)} \ge 0} \\ 
		&& \sum_a \pi(a; \theta) = 1
	\end{align}
\end{subequations}



{In the following, we first show that  the allocation $\pi$ being monotone is a sufficient and necessary condition for Program (\ref{pr:random}) to be implementable. This will greatly simplify our analysis to find the deterministic optimal solution later.}
\begin{claim}\label{monotoneclallocationcontract}
	The contracts in $U_R(z, \pi)$ is implementable if and only if the allocation $\pi$ is monotone as in Defintion \ref{defmonotonineaction}.
\end{claim}
\begin{proof}
	The necessary condition is straightforward. Given any $\theta < \theta'$, we have 
	\begin{align*}
		z^\theta -\theta \cdot c_{\pi(\theta)} &\ge z^{\theta'} -\theta \cdot c_{\pi(\theta')} 
	\end{align*}
	and
	\begin{align*}
		z^{\theta'} -\theta' \cdot c_{\pi(\theta')} &\ge z^{\theta} -\theta' \cdot c_{\pi(\theta)}
	\end{align*}
	which imply that $\theta (c_{\pi(\theta)} - c_{\pi(\theta')}) \le z^{\theta} - z^{\theta'} \le \theta' (c_{\pi(\theta)} - c_{\pi(\theta')})$. Hence, $c_{\pi(\theta)} \ge c_{\pi(\theta')}$. 
	
	Now, we prove the other direction.
	Suppose that we are given $c_{\pi(\theta_K)} \le c_{\pi(\theta_{K-1})} \le \dots \le c_{\pi(\theta_1)}$. We construct a feasible solution in the following way:
	\begin{itemize}
		\item For the largest type $\theta_K$, set $z^{\theta_K} = \theta_K \cdot c_{\pi(\theta_K)}$;
		\item For private type $\theta_i$ with $i < K$,  set $z^{\theta_i} = \theta_i \cdot c_{\pi(\theta_i)} + \max_{j > i} \{ z^{\theta_j} - \theta_i \cdot c_{\pi(\theta_j)} \}$. By induction, it is not hard to see that this is equivalent to set $z^{\theta_i} = \theta_i \cdot c_{\pi(\theta_i)} + z^{\theta_{i+1}} - \theta_i \cdot c_{\pi(\theta_{i+1})}$. First, we know $z^{\theta_{K-1}} = \theta_{K-1} \cdot c_{\pi(\theta_{K-1})} + z^{\theta_{K}} - \theta_{K-1} \cdot c_{\pi(\theta_{K})}$. Suppose that 
		\[z^{\theta_k} = \theta_k \cdot c_{\pi(\theta_k)} + \max_{j > k} \{ z^{\theta_j} - \theta_k \cdot c_{\pi(\theta_j)} \} = \theta_k \cdot c_{\pi(\theta_k)} + z^{\theta_{k+1}} - \theta_k \cdot c_{\pi(\theta_{k+1})}.\] 
		Hence, we have 
		\begin{align*}
			\max_{j > {k}} \{ z^{\theta_j} - \theta_{k-1} \cdot c_{\pi(\theta_j)} \} &= \max_{j > {k}} \{ z^{\theta_j} - (\theta_{k-1} - \theta_k)c_{\pi(\theta_j)}  - \theta_k c_{\pi(\theta_j)} \} \\
			&= z^{\theta_{k+1}} - (\theta_{k-1} - \theta_k)c_{\pi(\theta_{k+1})}  - \theta_k c_{\pi(\theta_{k+1})} \\
			&= z^{\theta_{k+1}} - \theta_{k-1} c_{\pi(\theta_{k+1})}  
		\end{align*}
		since $(\theta_{k-1} - \theta_k)c_{\pi(\theta_j)}$ is decreasing as $j$ decreases. Since $z^{\theta_k} - \theta_k \cdot c_{\pi(\theta_k)} = z^{\theta_{k+1}} - \theta_k \cdot c_{\pi(\theta_{k+1})}$, we then have $z^{\theta_k} - \theta_{k-1} \cdot c_{\pi(\theta_k)} \ge z^{\theta_{k+1}} - \theta_{k-1} \cdot c_{\pi(\theta_{k+1})}$. Hence, $z^{\theta_{k-1}}  = \theta_{k-1} \cdot c_{\pi(\theta_{k-1})} + z^{\theta_{{k}}} - \theta_{k-1} \cdot c_{\pi(\theta_{{k}})}$.
	\end{itemize}

	The IR constraint is obviously satisfied. It remains to show that the contract satisfies the IC constraint. First, consider two private types $\theta_j > \theta_i$. The agent $\theta_i$ will not misreport type $\theta_j$ since by construction $z^{\theta_i} - \theta_i \cdot c_{\pi(\theta_i)} =  \max_{j > i} \{ z^{\theta_j} - \theta_i \cdot c_{\pi(\theta_j)} \}$, i.e., $z^{\theta_i}- \theta_i \cdot c_{\pi(\theta_i)} \ge   z^{\theta_j} - \theta_i \cdot c_{\pi(\theta_j)}, \forall j > i$. On the other hand, consider $\theta_j < \theta_i$. Truthfulness can be shown inductively proving that for any $ \theta_{j}<\theta_{j+1} <\theta_i$, we have that $z^{\theta_{j}} - \theta_i\cdot c_{\pi(\theta_{j})} \le z^{\theta_{j+1}} - \theta_i\cdot c_{\pi(\theta_{j+1})}$ holds. This is guaranteed by our construction since 
	\begin{align*}
		\theta_{j}\cdot c_{\pi(\theta_{j})} + z^{\theta_{j+1}} -\theta_{j}\cdot c_{\pi(\theta_{j+1})} - \theta_i\cdot c_{\pi(\theta_{j})} &\le z^{\theta_{j+1}} - \theta_i\cdot c_{\pi(\theta_{j+1})} 
	\end{align*}
	holds if and only if 
	\begin{align*}
		(\theta_{j} - \theta_i) c_{\pi(\theta_{j})} &\le(\theta_{j} - \theta_i) c_{\pi(\theta_{j+1})},
	\end{align*}
	which holds due to $\theta_{j} < \theta_{i}$ and $c_{\pi(\theta_{j})} \ge c_{\pi(\theta_{j+1})}$. Similar arguments can be applied to prove that $\theta_i$ will not misreport to $\theta_{i-1}$. Therefore, by induction, we  can conclude that $\theta_i$ will not report $\theta_j\neq \theta_i$.
\end{proof}

By standard arguments as in \citet{myerson1981optimal,alon2021contracts}, we can show that the optimal design achieves a principal's utility equal to the maximum virtual welfare subject to monotone constraints.
\begin{claim}
	The optimum of Problem~(\ref{pr:random}) is equal to the maximum virtual welfare, i.e., to the value of the following program:
	\begin{align*}
		\max_{\pi}{ } & \quad  E_\theta[R_{\pi(\theta)} - c_{\pi(\theta)}\phi(\theta)] \\ 
		\textnormal{subject to} { } & \quad c_{\pi(\theta_K)} \le c_{\pi(\theta_{K-1})} \le \dots \le c_{\pi(\theta_1)}
	\end{align*}
\end{claim}
\begin{proof}
	Given any feasible (i.e., satisfying the monotonicity by Claim~\ref{monotoneclallocationcontract}) $\pi$ such that $c_{\pi(\theta_K)} \le c_{\pi(\theta_{K-1})} \le \dots \le c_{\pi(\theta_1)}$, by the incentive compatibility constraint, for any $\theta_k < \theta_{k+1}$, 
	we have that 
	$ z^{\theta_k} - \theta_k\cdot c_{\pi(\theta_k)} \ge z^{\theta_{k+1}} - \theta_k\cdot c_{\pi(\theta_{k+1})}$ and  $z^{\theta_{k+1}} - \theta_{k+1}\cdot c_{\pi(\theta_{k+1})} \ge z^{\theta_k} - \theta_{k+1}\cdot c_{\pi(\theta_k)}$.
	Hence, 
	\begin{equation}\label{sandwichinequality22}
		\theta_{k+1}\cdot (c_{\pi(\theta_k)} - c_{\pi(\theta_{k+1})} )\ge z^{\theta_k} - z^{\theta_{k+1}} \ge \theta_k\cdot (c_{\pi(\theta_k)}  - c_{\pi(\theta_{k+1})}).
	\end{equation}
	which implies $z^{\theta_k} \ge z^{\theta_{k+1}}$.
	Thus, there must exists some $\tilde{\theta}_{k} \in [\theta_k, \theta_{k+1}]$ such that $z^{\theta_k} - z^{\theta_{k+1}} = \tilde{\theta}_{k} (c_{\pi(\theta_k)} - c_{\pi(\theta_{k+1})} )$. By simple calculation, we get that 
	\begin{equation}\label{tmpTxck}
		z^{\theta_k} = z^{\theta_K} + \sum_{i=k+1}^K \tilde{\theta}_{i-1} \big[ c_{\pi(\theta_{i-1})} - c_{\pi(\theta_i)} \big].
	\end{equation}
	Given the allocation  $\pi(\cdot;\cdot)$, the expected utility of the principal is 
	$$E_\theta[R_{\pi(\theta)} - z^\theta] =  E_{\theta}[R_{\pi(\theta)}] - E_{\theta}[z^\theta]$$
	Hence, maximizing the expected utility is equivalent to minimizing $E_{\theta}[z^\theta]$.  Note that we can choose any $\tilde{\theta}_i \in [\theta_i, \theta_{i+1}]$ in Equation~\eqref{tmpTxck}. 
	{This holds since, given that $\pi$ is a feasible allocation, as long as Equation~(\ref{sandwichinequality22}) holds we have that IC is satisfied. Hence, the choice of $\tilde{\theta}_i$ within the interval only affects the payoff  but not the behavior of an agent in Equation~(\ref{pr:random}). A similar argument can be found in \cite{elkind2007designing}.}
	To minimize  $E_{\theta}[z^\theta]$, we choose
	$\tilde{\theta}_k = \theta_{k}$. That is the right-hand side of Equation~(\ref{sandwichinequality22}) will be binding. 
	
	The expected payment is
	\begin{small}
		\begin{eqnarray*}
			E_\theta[z^\theta] &=& \sum_{k=1}^K \mu(\theta_k) z^{\theta_K} + \sum_{k=1}^K \mu(\theta_k)\cdot \sum_{i=k+1}^K \theta_{i-1}[c_{\pi(\theta_{i-1})} - c_{\pi(\theta_i)}] \\ 
			&=&z^{\theta_K} + \sum_{k=1}^{K-1} \theta_k[c_{\pi(\theta_k)} - c_{\pi(\theta_{k+1})}]\cdot \sum_{i=1}^k \mu(\theta_i) \\
			&=&z^{\theta_K} - \sum_{k=1}^{K-1} \theta_k[c_{\pi(\theta_{k+1})} - c_{\pi(\theta_k)}]\cdot M(\theta_k) \\ 
			&=& z^{\theta_K} - \Big\{ \theta_{K-1}M(\theta_{K-1}) c_{\pi(\theta_K)} - \theta_1 M(\theta_1) c_{\pi(\theta_1)} -\sum_{k=2}^{K-1} c_{\pi(\theta_k)} [\theta_k M(\theta_k) -\theta_{k-1} M(\theta_{k-1})] \Big\} \\
			&=& z^{\theta_K} - \Big\{ \theta_{K-1}M(\theta_{K-1}) c_{\pi(\theta_K)} - \theta_1 M(\theta_1) c_{\pi(\theta_1)} -\sum_{k=2}^{K-1} c_{\pi(\theta_k)} [(\theta_k-\theta_{k-1}) M(\theta_{k-1}) + \theta_k \mu(\theta_k)] \Big\}\\
			&=& z^{\theta_K} - \Big\{ \theta_{K-1}M(\theta_{K-1}) c_{\pi(\theta_K)} -\sum_{k=1}^{K-1} c_{\pi(\theta_k)} [(\theta_k-\theta_{k-1}) M(\theta_{k-1}) + \theta_k \mu(\theta_k)] \Big\} \\
			&=& z^{\theta_K} - \Big\{ \theta_K M(\theta_K) c_{\pi(\theta_K)} -c_{\pi(\theta_K)} [ (\theta_K   -  \theta_{K-1})M(\theta_{K-1}) + \theta_K \mu(\theta_{K})  ] \\
			&&\quad -\sum_{k=1}^{K-1} c_{\pi(\theta_k)} [(\theta_k-\theta_{k-1}) M(\theta_{k-1}) + \theta_kf(\theta_k)] \Big\} \\
			&=& z^{\theta_K} -  \theta_K c_{\pi(\theta_K)} +\sum_{k=1}^{K} c_{\pi(\theta_k)} \mu(\theta_k) [\frac{(\theta_k-\theta_{k-1}) M(\theta_{k-1})}{\mu(\theta_k)} + \theta_k],  
		\end{eqnarray*}
	\end{small}where the fourth equality is by summation by parts and the sixth equality follows defining $M(\theta_0) = 0$. 
	Finally, the expected principal's utility is 
	\begin{equation}
		E_\theta[R_{\pi(\theta)} - z^\theta] = E_\theta[R_{\pi(\theta)} - c_{\pi(\theta)}\phi(\theta)] -(z^{\theta_K} -  \theta_K c_{\pi(\theta_K)}).
	\end{equation}
	Then, recall that by individual rationality, we have $z^{\theta_K} -  \theta_K c_{\pi(\theta_K)} \ge 0$. Thus, to maximize the principal's expected utility, we can set the price for the highest private type $\theta_K$ as $z^{\theta_K} =  \theta_K c_{\pi(\theta_K)}$. Note that this will not affect the feasibility, since  $z^{\theta_k} \ge z^{\theta_{k+1}}$ implies that $z^{\theta_K}$ is the smallest payment and decreasing all the payments by the same amount does not affect the IR constraints (since $ z^{\theta_k} - \theta_k\cdot c_{\pi(\theta_k)} \ge z^{\theta_{K}} - \theta_k\cdot c_{\pi(\theta_{K})} \ge  z^{\theta_{K}} - \theta_K\cdot c_{\pi(\theta_{K})} \ge 0$) and the IC constraints. Finally, the problem is equivalent to 
	\begin{align*}
		\max_{\pi}{ } & \quad  E_\theta[R_{\pi(\theta)} - c_{\pi(\theta)}\phi(\theta)] \\ 
		\textnormal{subject to} { } & \quad c_{\pi(\theta_K)} \le c_{\pi(\theta_{K-1})} \le \dots \le c_{\pi(\theta_1)}
	\end{align*}
\end{proof}


Next, we are ready to show the optimal solution $\pi$ is a deterministic function $a(\cdot)$.
Notice that $E_\theta[R_{\pi(\theta)} - c_{\pi(\theta)}\phi(\theta)] = E_\theta[\sum_a \pi(a; \theta)(R_{a} - c_{a}\phi(\theta))]$. If $\phi(\theta)$ is non-decreasing, we can deterministically assign $\theta$ the action $a(\theta) = \arg\max_{a} \Big\{R_a-c_a\phi(\theta)\Big\}$ where tie-breaking prefers smaller cost. This solution is optimum because we maximize the virtual welfare for each $\theta$. We only need to show that for any $\theta_i < \theta_j$, we have $a(\theta_i) \ge a(\theta_j)$ and thus, the monotonicity constraint is satisfied. Suppose $a(\theta_i) < a(\theta_j)$. Then, $R_{a(\theta_i)} - c_{a(\theta_i)} \phi(\theta_i) \ge R_{a(\theta_j)} - c_{a(\theta_j)} \phi(\theta_i)$  and $R_{a(\theta_j)} - c_{a(\theta_j)} \phi(\theta_j) > R_{a(\theta_i)} - c_{a(\theta_i)} \phi(\theta_j)$. The second inequality is strict due to the tie-breaking rule.  Hence, $\phi(\theta_i) \ge \frac{R_{a(\theta_j)} - R_{a(\theta_i)}}{c_{a(\theta_j)} - c_{a(\theta_i)}} > \phi(\theta_j)$, which contradicts the non-decreasing property. 

If $\phi(\theta)$ is not monotone, we can use the ironing technique developed in \cite{elkind2007designing} to gain a monotone $\bar{\phi}(\theta)$. Then, we deterministically assign $\theta$ the action $a(\theta) = \arg\max_{a} \Big\{R_a-c_a\bar{\phi}(\theta)\Big\}$ where tie-breaking prefers smaller cost. This allocation is feasible and achieves the maximum objective.

Finally, it is easy to know that all the steps above are polynomial. Hence, we designed a polynomial-time algorithm that maximizes the virtual social welfare, concluding the proof.


